\documentclass[prd,aps,10pt,twocolumn,superscriptaddress,floatfix,nofootinbib]{revtex4-1}
\pdfoutput=1 
\usepackage[T1]{fontenc}
\usepackage{amsmath}
\usepackage{amssymb}
\usepackage{amsthm}
\usepackage{amsfonts}
\usepackage{graphicx}
\usepackage{enumitem}
\usepackage{bm}
\usepackage{rotating}
\usepackage{hyperref}
\usepackage{pbox}
\usepackage[dvipsnames]{xcolor}
\usepackage{array}
\newcommand{\dd}{\mathrm{d}}

\usepackage{braket}
\usepackage{dsfont}
\usepackage{mathtools}
\usepackage{leftidx}
\usepackage{lmodern}
\usepackage[normalem]{ulem}
\usepackage{braket}
\usepackage{tensor}
\usepackage{mathrsfs}
\usepackage{dsfont}
\usepackage[margin=2cm]{geometry}
\usepackage[title]{appendix}
\usepackage{tcolorbox}
\usepackage{csquotes}

\usepackage{color,soul}
\newcommand{\bb}[1]{\boldsymbol{#1}}
\usepackage{tikz}
\usetikzlibrary{quantikz2}
\usetikzlibrary{external}
\tikzset{external/force remake}

\usepackage{tikz}
\usetikzlibrary{arrows.meta,decorations.pathmorphing,math}

\newcommand{\tcb}[1]{\leavevmode{\color{Blue}{#1}}}

\newcommand{\green}{}

\usepackage[colorinlistoftodos, textsize=tiny]{todonotes}

\renewcommand{\Re}{\mathrm{Re}}
\renewcommand{\Im}{\mathrm{Im}}
\renewcommand{\bar}{\overline}
\renewcommand{\tilde}{\widetilde}

\newcommand{\sx}{\mathsf{x}}
\newcommand{\sy}{\mathsf{y}}

\newcommand{\rr}[1]{\left(#1\right)}
\newcommand{\bx}{{\bm{x}}}
\newcommand{\by}{{\bm{y}}}
\newcommand{\bk}{{\bm{k}}}

\newcommand{\supp}{\text{supp}}

\newcommand{\R}{\mathbb{R}}
\newcommand{\C}{\mathbb{C}}
\newcommand{\M}{\mathcal{M}}

\newcommand{\A}{\mathcal{A}}
\newcommand{\W}{\mathcal{W}}
\newcommand{\CS}{C^\infty_0(\M)}
\newcommand{\Sol}{\mathsf{Sol}}

\newcommand{\fock}{{\mathfrak{F}(\mathcal{H})}}

\newtheorem{proposition}{Proposition}

\begin{document}

\title{Transmission of quantum information through quantum fields in curved spacetimes}

\author{Michael Kasprzak}
\thanks{Author to whom any correspondence should be addressed}

\email{mk7592@princeton.edu}
\affiliation{Department of Physics, Princeton University, Princeton, New Jersey 08544, USA}

\author{Erickson Tjoa}
\email{erickson.tjoa@mpq.mpg.de}
\affiliation{Max-Planck-Institut f\"ur Quantenoptik, Hans-Kopfermann-Stra\ss e 1, D-85748 Garching, Germany}

\begin{abstract}

We construct a relativistic quantum communication channel between two localized qubit systems, mediated by a relativistic quantum field, that can achieve the theoretical maximum for the quantum capacity in arbitrary curved spacetimes using the Unruh-DeWitt detector formalism. Using techniques from algebraic quantum field theory, we express the quantum capacity of the quantum communication channel purely in terms of the correlation functions of the field and the causal propagator for the wave equation. Consequently, the resulting quantum channel, and hence the quantum capacity, are by construction manifestly covariant, respect the causal structure of spacetime, and are independent of the details of the background geometry, topology, and the choice of Hilbert space (quasifree) representations of the field.  
\end{abstract}

\maketitle

\section{Introduction}

In the spirit of it-from-(qu)bit, quantum information theory has been extensively used to understand generic physical phenomena from quantum many-body physics to high-energy physics, from applied to foundational physics. Relativistic quantum information (RQI) is one of the many attempts to realize this spirit, which thus far focuses on the role of relativity in quantum information-processing tasks, as well as a deeper understanding of fundamental theories such as relativistic quantum field theories (QFT). The role of quantum information theory in relativistic QFTs cannot be understated: a non-exhaustive list includes rigorous study of the entanglement structure, complexity, and measurement theory for QFTs \cite{verch2005distillability,hollands2017entanglement,hollands2023channel,sanders2023separable,summers1985bell,summers1987bell,casini2020entanglement,casini2021entropic,longo2018relative,fewster2020measurement,bostelmann2021impossible,polo2021detectorbased,jubb2022causal,pranzini2023detector,vanLuijk2023schmidt,vanLuijk2024embezzlement}; operational approach to particle production due to non-inertial motion and gravitational fields \cite{hawking1975particle,Crispino2008review,Unruh1979evaporation,DeWitt1979,Aubry2014derivative,aubry2018quantumBH,tjoa2022unruh,kaplanek2020hot,tjoa2023effective,moustos2017nonmarkov}; connections between relativistic causality, information-theoretic causality and reference frames \cite{vilasini2022embedding,Hardy2007towards,zych2019bell,Costa2016causal,paunkovic2020causal};  quantum information protocols mediated by the relativistic quantum field \cite{reznik2003entanglement,reznik2005violating,Valentini1991nonlocalcorr,pozas2015harvesting,pozas2016entanglement,Jonsson2014cavityQED,tjoa2021harvesting,perche2024relativisticEH,tjoa2022teleport,Landulfo2016communication,tjoa2022capacity,lapponi2023relativistic,kent1999bitcommit,lo1998quantum,adlam2015crypto,buhrman2014position,vilasini2019composable,Blasco2015Huygens,Blasco2016broadcast,simidzija2017cosmo,Jonsson2018qubit,Simidzija2020capacity,barcellos2024broadcast}; etc. 

This work is motivated by a somewhat vague, ambitious, and fundamental question: 
\begin{tcolorbox}
    \textbf{Problem:} \textit{What is the information-carrying capacity of quantum fields?}
\end{tcolorbox}
\noindent One can think of several ways to answer this question by making the problem statement more precise. For our purposes, we will adopt an operational approach, namely, whether two parties Alice and Bob, each in possession of some quantum-mechanical probes (atoms, qudits, oscillators, etc.) are able to transmit (quantum) information using the relativistic quantum field as a medium --- that is, as a \textit{relativistic quantum channel} mediating the communication between Alice and Bob. A simple model for of such a communication channel is through the \textit{Unruh-DeWitt} (UDW) \textit{particle detector model} and its generalizations or variants \cite{Unruh1979evaporation,DeWitt1979,Tales2020GRQO,tjoa2023qudit,Lopp2021deloc,tjoa2023nonperturbative,tjoa2024UDW} (see, e.g., \cite{mancini2014frw} for one different approach). Indeed, transmission of classical information has been analyzed by many authors both perturbatively and non-perturbatively and in flat and curved spacetimes using the UDW model \cite{Blasco2015Huygens,Blasco2016broadcast,simidzija2017cosmo,Landulfo2016communication,tjoa2022capacity,lapponi2023relativistic,barcellos2024broadcast,lapponi2024comm2}. Transmission of quantum information in the non-perturbative regime has been studied only in flat spacetimes in very specific examples \cite{Jonsson2018qubit,Simidzija2020capacity}. Note that in all these cases, the communication is encoded through the field \textit{amplitude} and not, say, through the polarization degrees of freedom of the electromagnetic field, since on its own the polarization is blind to the relativistic features and constraints. 

Here we present the most general formulation for the transmission of quantum information between Alice and Bob, each in possession of a two-level system (a `qubit detector'), by suitable coupling to a relativistic scalar field in arbitrary curved spacetime. Clearly, our work constitutes a natural extension of the flat spacetime result in \cite{Jonsson2018qubit,Simidzija2020capacity,Blasco2015Huygens,Blasco2016broadcast,simidzija2017cosmo}. However, we will do so by exploiting techniques from {algebraic quantum field theory} (AQFT) \cite{wald1994quantum,fewster2019algebraic,Khavkhine2015AQFT,KayWald1991theorems}. Crucially, what we gain from this will not be a simple rewriting of all the formulae in the literature: we will see that the quality of the relativistic quantum channel mediated by the quantum scalar field can be formulated purely in terms of smeared correlation functions of the field and the detector parameters. Our results make clear what features of the field-mediated communication channel are necessary for perfect transmission of quantum information, and they are by construction manifestly covariant, respect the causal structure of spacetime, and are independent of the details of the background geometry, topology, and (to some extent) the Hilbert space representations of the field. In particular, it automatically works for all quasifree representations that give rise to Fock representations \cite{wald1994quantum,Khavkhine2015AQFT,fewster2019algebraic,KayWald1991theorems}).

This paper is organized as follows. In Section~\ref{sec: AQFT} we briefly review the scalar field theory in curved spacetimes using the AQFT framework. In Section~\ref{sec: setup} we describe our UDW detector setup and review the quantum capacity of a quantum channel. In Section~\ref{sec: protocol} we give the protocol that allows the transmission of quantum information through the quantum field. In Section~\ref{sec: analysis} we discuss the validity of the channel and study the regimes in which it fails. In Section~\ref{sec: conclude} we conclude with some discussions of our results and outlook. We set $c=\hbar=1$ and adopt mostly-plus signature for the metric.

\section{Scalar field theory in curved spacetime}
\label{sec: AQFT}

In this section we review the scalar field theory in curved spacetimes within the AQFT framework to establish notation and convention, adapting the convention given in \cite{hack2015cosmological,tjoa2023nonperturbative} (for more details and reviews, see also \cite{tjoa2023nonperturbative,fewster2019algebraic,wald1994quantum,Khavkhine2015AQFT}). Readers who are familiar with the framework can skip to Section~\ref{sec: setup}.

\subsection{AQFT for scalar field}

Consider a Klein-Gordon (KG) scalar field $\phi$ in an $(n+1)$-dimensional globally hyperbolic Lorentzian spacetime $\mathcal{M}$ with metric tensor $g_{ab}$, whose equation of motion is given by the KG equation: 
\begin{align}
    (-\nabla_a\nabla^a + m^2 +\xi R)\phi = 0\,,
   \label{eq: KGE}
\end{align}
where $\nabla_a$ is the covariant derivative with respect to the Levi-Civita connection, $m$ is the mass parameter, $R$ is the Ricci scalar curvature, and $\xi\geq 0$ controls the non-minimal coupling of $\phi$ to $R$. Global hyperbolicity implies that we have a foliation of spacetime $\R \times \Sigma$ and the KG equation is well-posed with initial data given on Cauchy surfaces. 

The KG operator $\hat{P}\coloneqq -\nabla_a\nabla^a + m^2 +\xi R$ admits a unique \textit{advanced and retarded Green's functions} $G_{A}(\sx,\sy),G_{R}(\sx,\sy) = G_{A}(\sy,\sx)$ which satisfies
\begin{align}
    \hat{P}_\sx G_{R/A}(\sx,\sy) = \delta^{n+1}(\sx-\sy)\,,
\end{align}
where the subscript on $\hat{P}_\sx$ means the KG operator is w.r.t. the variable $\sx$. We can use these Green's functions to construct the antisymmetric \textit{causal propagator} (also known as the Pauli-Jordan distribution)
\begin{align}
    E(\sx,\sy) \coloneqq G_R(\sx,\sy) - G_{\textsc{a}}(\sx,\sy)\,.
\end{align}
Canonical quantization of the field $\phi$ gives rise to a field operator $\hat{\phi}$ obeying the canonical commutation relations (CCR) 
\begin{align}
    [\hat{\phi}(\sx),\hat{\phi}(\sy)] = i E(\sx,\sy)\openone\,.
    \label{eq: CCR}
\end{align}
As $\hat{\phi}$ is an operator-valued distribution, in the algebraic approach to QFT we consider \textit{smeared field operators}, i.e., by viewing the field operator as an $\R$-linear map from the test space of smooth compactly supported functions $\CS$ to a $*$-algebra $\A(\M)$, i.e., as a map $\hat{\phi}: C^\infty_0 (\M)\to \mathcal{A}(\M)$ with
\begin{align}
    \hat{\phi}(f)\coloneqq \int\dd V \,f(\sx)\hat{\phi}(\sx)\,.
    \label{eq: ordinary smearing}
\end{align}
The algebra of observables $\A(\M)$ is defined to be the collection of sums of products of $\hat{\phi}(f)$ over all functions in $\CS$ obeying the smeared CCR relation \eqref{eq: CCR}, i.e., a \textit{CCR algebra} \cite{Bratteli1972afalgebra}. The KG equation is encoded in the fact that we identify $\hat{\phi}(f) \equiv 0$ if $f=\hat{P}h$ for some $h\in\CS$.

The dynamical content of the field theory is reflected in terms of the solution space of the KG equation. The solution space $\Sol_\R(\M)$ can be equipped with a symplectic form $\sigma:\Sol_\R(\M)\times\Sol_\R(\M)\to \R$, defined as
\begin{align}
    \sigma(\phi_1,\phi_2) \coloneqq \int_{\Sigma_t}\!\! {\dd\Sigma^a}\,\Bigr[\phi_{{1}}\nabla_a\phi_{{2}} - \phi_{{2}}\nabla_a\phi_{{1}}\Bigr]\,,
    \label{eq: symplectic form}
\end{align}
where $\dd \Sigma^a = -t^a \dd\Sigma$, $-t^a$ is the inward-directed unit normal to the Cauchy surface $\Sigma_t$, and $\dd\Sigma = \sqrt{h}\,\dd^n\bx$ is the induced volume form on $\Sigma_t$ \cite{poisson2009toolkit,wald2010general}. As is well-known, this definition is independent of the choice of Cauchy surface. Writing $Ef\coloneqq\int\dd V'E(\sx,\sx')f(\sx')$, the field operator $\hat\phi(f)$ can be expressed as \textit{symplectically smeared field operator}  \cite{wald1994quantum} 
\begin{align}
    \label{eq: symplectic smearing}
    {\hat\phi(f) \equiv \sigma(Ef,\hat\phi)\,,}
\end{align}
and the CCR algebra can be written as 
\begin{align}
    {[\sigma(Ef,\hat\phi),\sigma(Eg,\hat\phi)] = i\sigma(Ef,Eg)\openone = i E(f,g)\openone \,,}
\end{align}
where $\sigma(Ef,Eg) = E(f,g)$ in the second equality follows from Eq.~\eqref{eq: ordinary smearing} and \eqref{eq: symplectic smearing}. 

Since the smeared field operators are unbounded operators, for free fields it is more convenient technically to work the \textit{Weyl algebra} $\W(\M)$, whose elements are bounded operators constructed by ``exponentiating'' the field operators. The Weyl algebra $\W(\M)$ is a unital $C^*$-algebra generated by elements that formally take the form 
\begin{align}
    W(Ef) \equiv 
    {e^{i\hat\phi(f)}}\,,\quad f\in \CS\,.
    \label{eq: Weyl-generator}
\end{align}
These elements satisfy \textit{Weyl relations}:
\begin{equation}
    \begin{aligned}
    W(Ef)^\dagger &= W(-Ef)\,, \notag \\
    \quad 
    W(E (\hat Pf) ) &= \openone\,,\\
    W(Ef)W(Eg) &= e^{-\frac{i}{2}E(f,g)} W(E(f+g))
    \end{aligned}
    \label{eq: Weyl-relations}
\end{equation}
where $f,g\in \CS$. The third relation enforces relativistic causality (or \textit{microcausality}).

In AQFT, a quantum state is defined as a $\C$-linear functional $\omega:\W(\M)\to \C$ (similarly for $\A(\M)$) such that 
\begin{align}
    \omega(\openone) = 1\,,\quad  \omega(A^\dagger A)\geq 0\quad \forall A\in \W(\M)\,.
    \label{eq: algebraic-state}
\end{align}
The state $\omega$ is pure if it cannot be written as $\omega= \alpha \omega_1 + (1-\alpha)\omega_2$ for any $\alpha\in (0,1)$ and any two algebraic states $\omega_1,\omega_2$; otherwise we say that the state is mixed.

The way to pass from AQFT to the Hilbert space approach is through \textit{Gelfand-Naimark-Segal (GNS) reconstruction theorem} \cite{wald1994quantum,Khavkhine2015AQFT,fewster2019algebraic}. The GNS theorem gives us, for a given algebra of observables $\W(\M)$ and a state $\omega$, a \textit{GNS triple} $(\mathcal{H}_\omega, \pi_\omega,{\ket{\Omega_\omega}})$, where $\pi_\omega: \mathcal{\W(\M)}\to {\mathcal{B}(\mathcal{H}_\omega)}$ is a Hilbert space representation with respect to state $\omega$. In its GNS representation, any state $\omega$ is realized as a {vector state} $\ket{\Omega_\omega}\in\mathcal{H}_\omega$ and  $A\in \W(\M)$ are represented as bounded operators $\hat A\coloneqq \pi_\omega(A)\in \mathcal{B}(\mathcal{H}_\omega)$, and we write $ \omega(A) = \braket{\Omega_\omega|\hat A|\Omega_\omega}$. Since there exists infinitely many unitarily inequivalent representations of the CCR algebra, the algebraic approach allows us work with all representations at once and only pick the physically relevant representation at the very end. 

For a fixed state $\omega$, we can construct the \textit{Wightman} $n$-\textit{point functions}, defined by
\begin{align}
    \mathsf{W}(f_1,...,f_n)\coloneqq \omega(\hat\phi(f_1)...\hat\phi(f_n))
    \label{eq: n-point-functions}
\end{align}
where $f_j\in \CS$ and the RHS is computed within some GNS representation of $\A(\M)$. The GNS representation of the Weyl algebra $\W(\M)$ allows us to calculate Eq.~\eqref{eq: n-point-functions} by differentiation\footnote{Strictly speaking, the state on $\A(\M)$ and on $\W(\M)$ are different, but they are in one-to-one correspondence whenever they are related by Eq.~\eqref{eq: Wightman-formal-bulk} (see \cite{ruep2021weakly} for good exposition on this). }: for example,  the smeared Wightman two-point function reads
\begin{align}\label{eq: Wightman-formal-bulk}
    &\mathsf{W}(f,g) \equiv -\frac{\partial^2}{\partial s\partial t}\Bigg|_{s,t=0}\!\!\!\!\!\!\!\!\omega(e^{i\hat\phi(sf)}e^{i\hat\phi(tg)})
\end{align}
where the RHS is calculated in the GNS representation of $\W(\M)$  \cite{fewster2019algebraic}. 

The consensus within the AQFT community is that physically reasonable states should be \textit{Hadamard states} \cite{KayWald1991theorems,Radzikowski1996microlocal}. Very roughly speaking, these states respect local flatness and the expectation values of all observables (in particular, the renormalized stress-energy tensor) are finite \cite{KayWald1991theorems}. A particularly nice subclass of Hadamard states is the family of \textit{quasifree states}: for these states, all odd-point functions in the sense of \eqref{eq: n-point-functions} vanish and all higher even-point functions can be written as in terms of just two-point functions. In modern terminology, the term quasifree state is synonymous with \textit{Gaussian state}, where the one-point functions need not vanish and higher-point functions only depend on one- and two-point functions \cite{ruep2021weakly}.

The relevance of quasifree states lies in the fact that they are completely specified once we know the Wightman two-point functions. More precisely, any given quasifree state $\omega$ is associated with a \textit{real-bilinear inner product} $\mu:\Sol_\R(\M)\times\Sol_\R(\M)\to \R$ that satisfies \cite{KayWald1991theorems}
\begin{align}
    |\sigma(Ef,Eg)|^2 \leq \mu(Ef,Ef)\mu(Eg,Eg)\,, 
    \label{eq: real-bilinear IP}
\end{align}
where we recall that $Ef,Eg\in \Sol_\R(\M)$ for all $f,g\in\CS$. The inequality is saturated if $\omega_\mu$ is a \textit{pure} quasifree state. Any quasifree state can then \textit{defined} by those that satisfy
\begin{align}
    \omega(W(Ef)) \coloneqq e^{-\mu(Ef,Ef)/2}\,.
    \label{eq: quasifree-def}
\end{align}
However, this definition is not useful unless we can compute the induced norm $||Ef|| \coloneqq \sqrt{\mu(Ef, Ef)}$. 

It turns out that $||Ef||$ is related to the Wightman two-point function as follows \cite{KayWald1991theorems}: we complexify $\Sol_\R(\M)$ into the space of \textit{complex} solutions $\Sol_\C(\M)$ and define the \textit{KG bilinear product} by 
\begin{align}
    \braket{\varphi_1,\varphi_2}_{\textsf{KG}}\coloneqq i\sigma(\varphi_1^*,\varphi_2)
\end{align}
for $\varphi_1,\varphi_2\in \Sol_\C(\M)$. It can be shown that for any quasifree state, there exists a subspace $\mathcal{H}\subset\Sol_\C(\M)$ such that $(\mathcal{H},\braket{\cdot,\cdot}_{\textsf{KG}})$ is a Hilbert space and an $\R$-linear map $K:\Sol_\R(\M)\to\mathcal{H}$ such that for all $\varphi_1,\varphi_2\in\Sol_\R(\M)$
\begin{enumerate}[label=(\alph*)]
    \item $K\Sol_\R(\M)+i K\Sol_\R(\M)$ is dense in $\mathcal{H}$;
    \item $\mu(\varphi_1,\varphi_2) = \Re\braket{K\varphi_1,K\varphi_2}_{\textsf{KG}}$;
    \item $\sigma(\varphi_1,\varphi_2) = 2\Im\braket{K\varphi_1,K\varphi_2}_{\textsf{KG}}$;
    \item $\Sol_\C(\M)\cong \mathcal{H}\oplus\mathcal{\overline H}$, where $\mathcal{\overline{H}}$ is the complex-conjugate Hilbert space of $\mathcal{H}$ and $\braket{u,v}_{\textsf{KG}}=0$ for all $u\in\mathcal{H}$ and $v\in\mathcal{\overline H}$.
\end{enumerate}
In a more familiar language of canonical quantization, the map $K$ projects to the ``positive-frequency part'' of a real solution to the KG equation. The Wightman two-point function is then given by \cite{KayWald1991theorems}
\begin{align}
    \mathsf{W}(f,g) &= \braket{KEf,KEg}_{\textsf{KG}} \notag\\
    &= \mu(Ef,Eg) + \frac{i}{2}E(f,g)\,,
\end{align}
where we have used the fact that $\sigma(Ef,Eg) = E(f,g)$. Since $E(f,g)$ is antisymmetric, it follows that $||Ef||^2 = \mu(Ef,Ef) \equiv \mathsf{W}(f,f)$ and hence
\begin{align}
    \omega(W(Ef)) =   e^{-{\frac{1}{2}}\mathsf{W}(f,f)}\,.
    \label{eq: quasifree-definition}
\end{align}
Indeed, we may very well take Eq.~\eqref{eq: quasifree-definition} as the \textit{definition} of quasifree states.

The most important example of a quasifree state is the ``vacuum state'' $\omega$, where we can write the (unsmeared) vacuum Wightman function as 
\begin{align}
    \omega(\hat{\phi}(\sx)\hat{\phi}(\sy)) \equiv \mathsf{W}(\sx,\sy) &= \int \dd^n\bk\, u^{\phantom{*}}_\bk(\sx) u^*_\bk(\sy)\,,
    \label{eq: two-point-wightman}
\end{align}
where $u_\bk(\sx)$ are called the positive-frequency modes of KG operator with respect to the KG inner product
\begin{align}
    (\phi_1,\phi_2)_\textsc{kg}\coloneqq i\sigma(\phi_1^*,\phi_2)\,,
\end{align}
where $\phi_j\in \Sol_\C(\M)$ are complexified solutions to Eq.~\eqref{eq: KGE}. The vacuum state $\omega$ is not unique and there is a sense in which every GNS vector is a ``vacuum state'' in that particular representation. However, once a particular basis of modes $\{u_\bk\}$ are chosen, it fixes the vacuum state as one whose Wightman function satisfies \eqref{eq: two-point-wightman}. In flat spacetime, the standard plane-wave basis in the inertial coordinates defines the Minkowski vacuum, which is the only state that respects the full Poincaré symmetry and minimizes the expectation value of the (renormalized) stress-energy tensor \cite{fewster2019algebraic}. In curved spacetimes, there are multiple inequivalent states that qualify as a vacuum state \cite{wald1994quantum}.

\subsection{Relationship with canonical quantization} 

The usual Fock space in canonical quantization arises from the GNS reconstruction applied to $\A(\M)$ with respect to some vacuum state $\omega$. Technically speaking, any pure quasifree (Hadamard) state qualifies as a vacuum state, as their GNS representation gives rise to the Fock vacuum in the standard sense \cite{KayWald1991theorems} (see also \cite{fewster2019algebraic} for a different definition in Minkowski spacetime). This is the best one can do in generic curved spacetimes as there is no preferred vacua without further assumptions such as time translation symmetry. 

In more detail, the GNS Hilbert space $\mathcal{H}_\omega$ for the pure quasifree representation is the Fock space over the one-particle Hilbert space $\mathcal{H}$
\begin{align}
    \mathcal{H}_\omega \equiv \fock = \bigoplus_{n=0}^\infty \mathcal{H}^{\odot n}\,,
\end{align}
where $\mathcal{H}^{\odot 0}\cong \C$ and $\mathcal{H}^{\odot n}$ means symmetrized direct sum (the $n$-particle sector of the Fock space). In this representation, we can write
\begin{align}
    \pi_\omega(\hat\phi(f)) = \hat{a}(({KEf})^*) + \hat{a}^\dagger(KEf)\,.
\end{align}
In what follows we drop $\pi_\omega$ if it is clear from the context that we are using a Fock representation. Note that if we consider complex smearing functions $f:\M\to \C$, we can write
\begin{align}
    \hat{\phi}(f) \equiv \hat{\phi}(\Re f) + i \hat{\phi}(\Im f)\,.
\end{align}
The operators $\hat{a}(u^*),\hat{a}^\dagger(v)$ are \textit{smeared ladder operators}\footnote{There are various conventions on how to label the smeared ladder operators (see \cite{tjoa2022capacity}). We follow \cite{wald1994quantum} in that it views $\hat{a},\hat{a}^\dagger$ as being linear in their arguments. The convention in \cite{fewster2019algebraic} does not include complex conjugation in the argument for $\hat{a}(\cdot)$ and views $\hat{a}$ as an anti-linear map while $\hat{a}^\dagger$ is linear.} obeying the CCR
\begin{align}
    [\hat{a}(u^*),\hat{a}^\dagger(v)] = \braket{u,v}_{\textsf{KG}}\openone
\end{align}
on a suitable dense domain of the Fock space.

The standard canonical approach is recovered by working with the unsmeared field operator $\hat{\phi}(\sx)$ and considering the Fourier mode decomposition
\begin{align}
    \hat{\phi}(\sx) = \int\dd^n\bk\,\hat{a}_\bk u_\bk(\sx) + \hat{a}_\bk^\dagger u_\bk^*(\sx)\,,
\end{align}
where $u_\bk(\sx)$ and $u_\bk^*(\sx)$ are positive- and negative-frequency modes normalized to Dirac delta functions via the KG inner product:
\begin{subequations}
\begin{align}
    \braket{u_\bk,u_{\bk'}}_{\textsf{KG}} &= \delta^n(\bk-\bk')\,,\\
    \braket{u^*_\bk,u^*_{\bk'}}_{\textsf{KG}} &= -\delta^n(\bk-\bk')\,,\\
    \braket{u_\bk,u^*_{\bk'}}_{\textsf{KG}} &= 0\,.
\end{align}
\end{subequations}
These modes are not proper elements of the one-particle Hilbert space but they are convenient to work with. In particular, we see that $\hat{a}^{\phantom{*}}_\bk\equiv \hat{a}(u_\bk^*)$ and 
$\hat{a}_\bk^\dagger\equiv \hat{a}^\dagger(u_\bk)$. If we were to pick any other quasifree state, the resulting one-particle structure will differ and this can give rise to Hilbert space representations that are unitarily inequivalent, e.g., by considering Kubo-Martin-Schwinger (KMS) thermal states \cite{kubo1957statistical,martinSchwinger1959theory,KayWald1991theorems,fewster2019algebraic}. 

Using these modes, we can decompose the solution $Ef$ in the ``eigenmode basis'' $\{u_\bk,u_\bk^*\}$
\begin{align}
    Ef \equiv \int\dd^n\bk \,\braket{u_\bk,Ef}_{\textsf{KG}}u_\bk + \braket{u^*_\bk,Ef}_{\textsf{KG}}u_\bk^*\,. \label{eq: basis decomposition of E}
\end{align}
From the CCR and the definition of the Wightman function, we see that
\begin{align}
   {i}E(\sx,\sx') &= \int\dd^n\bk\,u_\bk(\sx)u_\bk^*(\sx') - u_\bk^*(\sx)u_\bk(\sx')\,, 
\end{align}
the positive-frequency part $KEf$ of $Ef$ reads
\begin{align}
    KEf &= \int\dd^n\bk \,\braket{u_\bk,Ef}_{\textsf{KG}}u_\bk \equiv \tcb{-i}\int\dd^n\bk\,{f_\bk} u_\bk  
\end{align}
where
\begin{align}
    {f}_\bk\coloneqq \int\dd V\,f(\sx)u_\bk^*(\sx)\,.
    \label{eq: f-fourier}
\end{align}
Furthermore, since $KEf $ lies in $\mathcal{H}$, it is convenient to use the notation $\ket{KEf}\in \mathcal{H}$ and $\ket{u_\bk}$ for the (improper) basis $u_\bk$. In this notation, we write
\begin{align}
    \ket{KEf} = {-i}\int\dd^n\bk \, {{f}_\bk} \ket{u_\bk}\,,
    \label{eq: KEf-notation}
\end{align}
where the global phase $-i$ is a matter of convention. This suggests that we can label the ladder operators using the ``momentum space smearing function'' ${f}_\bk$ rather than the ``real space positive-frequency solution'' $KEf$, i.e., we can write
\begin{align}
    \hat{\phi}(f) &= \hat{a}({f}_\bk^*) + \hat{a}^\dagger({f}_\bk) \equiv \int\dd^n\bk\,\rr{\hat{a}_\bk^{\phantom{\dagger}} {{f}^*_\bk} + \hat{a}^\dagger_\bk {{f}_\bk}}\,. 
    \label{eq: fock-representation}
\end{align}
Essentially, we are using the elements of the one-particle Hilbert space $\mathcal{H}$ to label the field observables rather than using the spacetime smearing functions $\CS$. This notation is particularly convenient if the focus is more about which modes of the field that the detectors are coupled to \cite{spohn1989spinboson,hasler2021existence,fannes1988equilibrium,tjoa2024UDW}.

\section{Setup}
\label{sec: setup}

In this section, we introduce the UDW detector model that we will use to couple Alice's and Bob's qubits to the field and briefly review the notion of quantum capacity for a quantum channel and coherent information. These are relevant measures of how much quantum information can be transmitted (on average) through the quantum channel between Alice and Bob.

\subsection{UDW detector model}
\label{subsec: UDW}

In the UDW framework, Alice and Bob carry their own qubit (`detector') with a free Hamiltonian $\hat H_\nu = \frac{1}{2}\Omega_\nu \hat \sigma_z$, where $\nu \in \{A, B\}$, $\Omega_\nu$ the energy gap and $\hat \sigma_z$ the Pauli-$Z$ operator. We denote the eigenstates of $\hat \sigma_s$ by $\ket{\pm_s}$ with $s \in \{x,y,z\}$, where we do not label the Pauli matrices with $\nu$ to reduce notational clutter and the context will make it clear.  The excited and ground states of the free Hamiltonian $\hat H_\nu$ are given by $\ket{\pm_z}$ with eigenvalues $\pm \Omega_\nu/2$ respectively.

The UDW interaction between Alice's and Bob's detectors and the field is assumed to be linear in the field as it represents a simplified model of dipole interaction $\hat{\mathbf{d}}\cdot\hat{\mathbf{E}}$ in quantum optics, i.e., as a spin-boson type interaction \cite{tjoa2024UDW}. The general form of such an interaction is given by a linear combination of the form\footnote{This is a slight modification of the usual spin-boson model since typically spin-boson models are assumed to have time-independent interactions \cite{tjoa2024UDW} and to account for the covariant formulation of the UDW model \cite{Tales2020GRQO}.}
\begin{align}
    \hat h_{I,\nu}(\sx) \coloneqq \sum_{i}\lambda_{\nu,i} f_{\nu,i}(\sx) \hat m_{\nu,i}(\tau(\sx)) \otimes \hat \phi(\sx)\,,
    \label{eq: hamiltonian-density}
\end{align}
where $\lambda_{\nu,i}$ is the coupling strength, $\hat {m}_{\nu,i}(\tau)$ is a Hermitian operator acting on the qubit, and $f_{\nu,i}(\sx)$ is a spacetime smearing function defining the duration of the interaction and the spatial profile of the detector. The most commonly used UDW model involves only one term in $\hat{h}_{I,\nu}(\sx) = \lambda_\nu f_\nu(\sx)\hat{m}_\nu(\tau(\sx))\otimes\hat{\phi}(\sx)$ \cite{Tales2020GRQO}, but we will need the more general version in \eqref{eq: hamiltonian-density} for our purposes. We drop the tensor product symbol if the context is clear.

In the interaction picture, the time evolution operator generated by the interaction Hamiltonian \eqref{eq: hamiltonian-density} reads \cite{Tales2020GRQO}
\begin{equation}
    \hat U = \mathcal{T} \exp\left[-i \int_\M \dd{V} \sum_\nu \hat h_{I,\nu}\right]
    \label{eq: unitary-general}
\end{equation}
where $\dd{V} \coloneqq \dd^n{\sx} \sqrt{-g}$ is the invariant volume element and $\mathcal{T}$ is the time-ordering prescription\footnote{The time-ordering is assumed to be given by some global time function using some spacetime foliation, see \cite{Bruno2020time-ordering} for discussions about some subtleties of time-ordering for the spatially smeared UDW model.}. In general, it is not possible to solve the time evolution analytically due to the time-dependence of the interaction, thus many analyses involving the UDW detectors consider perturbative methods in the weak coupling regimes. However, our goal in this paper is to establish a perfect quantum channel between Alice and Bob and as was explained in \cite{Simidzija2020capacity}, this requires our model to be in the non-perturbative regime.

In order to construct the perfect quantum channel that can transmit quantum information perfectly, we need to be able to calculate the interaction unitary $\hat{U}$ non-perturbatively (in the sense of not truncating the Dyson series in the weak coupling regime). For this purpose, we will need to use the so-called delta-coupling regime \cite{tjoa2023nonperturbative}, where the interaction timescale is assumed to be much faster than all the relevant timescales of the problem so that the interaction can be taken to occur at a single instant in time. Delta-coupled detectors have been used in many contexts such as entanglement extraction, relativistic communication, work extraction, and many others (see, e.g., \cite{Simidzija2018no-go,Simidzija2017coherent,lapponi2023relativistic,Landulfo2016communication,Simidzija2020capacity,tjoa2022capacity,gallock2024relativistic,kollas2024engine,polo2024sequence}). As each detector has its own rest frame, we need to consider the Fermi normal coordinates (FNC) associated with each observer  \cite{poisson2011motion}, whose center-of-mass (COM) trajectory is parametrized by their respective proper times \cite{Tales2020GRQO}.

The physically reasonable assumption we make is that in the FNC, $\bar{\sx} \equiv (\tau,\bar{\bx})$, the spacetime smearing functions can be decomposed as
\begin{align}
    f(\sx(\bar{\sx})) = \chi(\tau)F(\bar{\bx})\,,
\end{align}
where $(\tau,\bar{\bx}=\mathbf{0})$ labels the COM trajectory. In other words, Alice (and Bob) can in their rest frames distinguish the spatial profile (which is assumed to satisfy `{Born rigidity}', i.e., time-independent with respect to $\tau$) from the switching function that controls the duration and strength of the interactions. The Hamiltonian density in this case reads
\begin{align}
    \hat h_{I,\nu}(\bar{\sx}) \coloneqq \sum_{i}\lambda_{\nu,i} \chi_{\nu,i}(\tau)F_{\nu,i}(\bar{\bx}) \hat m_{\nu,i}(\tau ) \hat \phi(\sx(\bar{\sx}))
    \label{eq: hamiltonian-density-FNC}
\end{align}
and the delta coupling interaction corresponds to the case when
\begin{align}
    \chi_{\nu,i}(\tau) = \delta(\tau-\tau_{\nu,i})
\end{align}
for some constants $\tau_{\nu,i}\in \R$. 

For our purposes, we consider the scenario where Alice and Bob each interact with the field via delta coupling twice, so that we need to consider multiple interaction times $\tau_{\nu,i}$  with $\tau_{\nu,1} < \tau_{\nu,2}$. More generally, if Alice (or Bob) performs a sequence of $N$ delta interactions with $\tau_{\nu,i}<\tau_{\nu,i+1}$, the corresponding unitary time evolution $\hat{U}_\nu$ can be decomposed into a sequence of \textit{simple-generated unitaries} \cite{Simidzija2018no-go,tjoa2023nonperturbative} (see also \cite{polo2024sequence})
\begin{align}
    \hat U_\nu &= \prod_{j=1}^N\hat{U}_{\nu,j}\,,\,\,\hat U_j = \exp[-i \lambda_\nu \hat m_\nu(\tau_{\nu,j})  \hat \phi(F_{\nu,j})].
\end{align}
The time-ordering $\mathcal{T}$ is taken care of since the unitaries $\hat U_i$ are written in time-increasing order. When two or more observers are involved, we can use the global time coordinates to order $\hat{U}_{\textsc{a}}$ and $\hat{U}_{\textsc{b}}$: for example, if Alice and Bob are spacelike separated then the ordering does not matter, while Alice unitary acts before Bob $\hat{U}_{\textsc{b}}$ if $\supp(f_{\textsc{b}})$ is in the causal future of $\supp(f_{\textsc{a}})$. The unitary $\hat{U}_\nu$ can be computed non-perturbatively since $\hat{U}_j$ is a \textit{controlled unitary}: for any two Hermitian operators $\hat{A}, \hat{B}$ such that $\hat{A}=\sum_{i=1}^ra_j \ket{a_j}\!\bra{a_j}$ is finite-dimensional operator in its spectral decomposition, we can write
\begin{align}
    \exp\left[{i \hat{A}\otimes \hat B}\right] = \sum_{j=1}^r\ket{a_j}\!\bra{a_j}\otimes \exp\left[{ia_j \hat{B}}\right]\,.
\end{align}
As we will see, for our setup the full unitary time evolution $\hat{U}$  will consist of products of Alice's two simple-generated unitaries and also Bob's two simple-generated unitaries.

Last but not least, let us recall the unitaries found in \cite{Simidzija2020capacity} for Minkowski spacetime. There Alice performs two simple-generated unitaries, but one of them involves coupling to a conjugate momentum
\begin{align}
    \hat{U}_{\textsc{a}} = e^{i\hat{\sigma}_x \hat{\pi}(F_{\textsc{a}})}e^{i\hat{\sigma}_z \hat{\phi}(F_{\textsc{a}})}\,,
    \label{eq: flat-space-Alice-unitary}
\end{align}
where the \textit{spatially smeared operators} are
\begin{subequations}
\begin{align}
    \hat{\phi}(F_{\textsc{a}}) &= \lambda_1\int\dd^3\bx\,F_{\textsc{a}}(\bx)\hat{\phi}(t_{\textsc{a}},\bx)\,,\\
    \hat{\pi}(F_{\textsc{a}}) &= \lambda_2\int\dd^3\bx\,F_{\textsc{a}}(\bx)\partial_t\hat{\phi}(t_{\textsc{a}},\bx)\,.
\end{align}
\end{subequations}
Note that here they take $t_{\textsc{a}}\coloneqq t_{A,1}$ and $t_{A,2} = t_{\textsc{a}}+\epsilon$ and consider the limit as $\epsilon\to 0^+$ so that the unitary generated by the conjugate momentum acts later. Bob's unitary is then constructed as a ``suitable inverse'' of Alice's unitary to faithfully extract Alice's qubit.

Our construction requires us to make two modifications to the above unitaries. First, we can generalize the conjugate momentum operator in curved spacetime, i.e., 
\begin{equation}
    \hat \pi(\sx) \coloneqq \partial_t\hat{\phi}(\sx)\equiv \sqrt{h} t^a\nabla_a \hat \phi(\sx)\,,
\end{equation}
where $t^a$ is the forward pointing unit normal, orthogonal to $\Sigma_t$ with an induced metric $h$. Second, it is straightforward to see from integration by parts that 
\begin{align}
    \hat{\pi}(f) &= \int\dd V\,f(\sx)\partial_t\hat{\phi}(\sx) = \hat{\phi}(-\partial_tf)\,.
\end{align}
Thus we can always view the smeared conjugate momentum operator as a smeared field operator with a different smearing function (namely, its time derivative). In Minkowski spacetime, using the spacetime smearing notation we can thus write
\begin{align}
    \hat{\pi}(F_{\textsc{a}})\equiv \hat{\pi}(\delta \cdot F_{\textsc{a}}) = \hat{\phi}(-\delta'\cdot F_{\textsc{a}}) \,,
\end{align}
where $\delta'(t)$ is the distributional derivative of the Dirac delta function $\delta(t)$. This second step might appear unconventional since $\hat{\pi}$ is not usually regarded as a spacetime-smeared operator, but this provides a simple way to recast all basic observables in terms of a single field operator $\hat{\phi}$. The connection between the modern AQFT approach based on spacetime-smeared operators (\textit{cf.} Section~\ref{sec: AQFT}) and the ``old'' (non-covariant) AQFT formulations based on two spatially smeared operators $\{\hat{\phi}|_\Sigma,\hat{\pi}|_\Sigma\}$ at a fixed Cauchy slice $\Sigma$ with equal-time CCR $[\hat\phi(t,\bx),\hat{\pi}(t,\by)]=i\delta^{(n)}(\bx-\by)$ can be made more precise and rigorous  (see, e.g., \cite{dimock1980algebras,kay1978linear}, and also \cite{hack2015cosmological,wald1994quantum} for details). 

The two modifications above suggest that what we really need for the generalization to curved spacetime is that Alice couples to two smeared field operators $\hat{\phi}(f_1),\hat{\phi}(f_2)$ that do not commute, i.e., $E(f_1,f_2)\neq 0$, and not necessarily to the field and its conjugate momentum. Indeed, in the modern AQFT framework the coupling to momentum is somewhat awkward because the perfect quantum channel in \cite{Simidzija2020capacity} requires that Alice couples to $\hat{\phi}(F_{\textsc{a}})$ more strongly than $\hat{\pi}(F_{\textsc{a}})$, but the coupling strengths $\lambda_1,\lambda_2$ would need to have different units, so one needs an additional length scale to compare their strengths. By considering different smeared field operators $\hat{\phi}(f_i)$, we can have $\lambda_i$ with the same units and we could use other length scales (e.g., from the effective size of the spatial smearing) if we insist on interpreting one of them as conjugate momentum. In what follows we assume that $f_i$ (and hence $\lambda_i$) have the same units for simplicity.

\subsection{Quantum capacity and coherent information}

Here we briefly review the notion of quantum capacity of a quantum channel and quantum coherent information, following closely the exposition in \cite{Simidzija2020capacity,Landulfo2016communication} (see \cite{khatri2020principles,wilde2011classical,gyongyosi2018surveycapacity,Holevo2020capacities} for more in-depth details). We denote by $\mathscr{D}(\mathcal{H})$ the space of density operators associated with some Hilbert space $\mathcal{H}$. Let $\Phi:\mathscr{D}(\mathcal{H}_\textsc{a})\to \mathscr{D}(\mathcal{H}_\textsc{b})$ be a quantum channel --- a completely-positive trace-preserving (CPTP) map --- between Alice's and Bob's density operators. 

An ideal quantum communication channel is one that is able to send any quantum state from Alice to Bob, or equivalently, one that preserves entanglement reliably since Alice's state can always be viewed as an entangled state with some environment that purifies it.  Schematically the idea is the following. First, Alice prepares a purification of the state $\rho_{\textsc{a}}\in \mathscr{D}(\mathcal{H}_\textsc{a})$, namely, $\rho_{\textsc{ea}}\coloneqq \ket{\psi_{\textsc{ea}}}\!\bra{\psi_{\textsc{ea}}}\in \mathcal{H}_\textsc{e}\otimes{\mathcal{H}_\textsc{a}}$ with $\mathcal{H}_\textsc{a}\cong\mathcal{H}_\textsc{e}$. Then Alice applies an encoding channel  $\mathcal{E}:\mathscr{D}(\mathcal{H}_\textsc{e}\otimes \mathcal{H}_\textsc{a})\to\mathscr{D}(\mathcal{H}_\textsc{e}\otimes \mathcal{H}_\textsc{a}^{\otimes N})$ to map her share of pure state into $N$ quantum systems, i.e.,
\begin{align}
    \rho_{\textsc{ea},N}\coloneqq (\openone_\textsc{e}\otimes \mathcal{E})(\rho_\textsc{ea})\,.
\end{align}
After encoding, Alice sends each quantum system through $N$ independent uses of the communication channel $\Phi$, i.e.,
\begin{align}
    \rho_{\textsc{eb},N}\coloneqq \Phi^{\otimes N}(\rho_{\textsc{ea},N})\in\mathscr{D}(\mathcal{H}_\textsc{e}\otimes\mathcal{H}_\textsc{b}^{\otimes N})\,.
\end{align}
Bob's task is to use a decoding channel $\mathcal{D}:\mathscr{D}(\mathcal{H}_\textsc{e}\otimes \mathcal{H}_\textsc{b}^{\otimes N})\to \mathscr{D}(\mathcal{H}_\textsc{e}\otimes \mathcal{H}_\textsc{b})$, such that
\begin{align}
    \rho_{\textsc{eb}} \coloneqq (\openone_{\textsc{e}}\otimes \mathcal{D})(\rho_{\textsc{eb},N})\,.
\end{align}
This is a good encoding-decoding procedure if for a given $\epsilon$ we can have
\begin{align}
    ||\tilde{\rho}_\textsc{eb}-\rho_\textsc{eb}||_1\leq \epsilon
\end{align}
where $||\cdot||_1$ is the trace norm, $\tilde{\rho}_\textsc{eb}$ is the state obtained if $\Phi$ is replaced with a \textit{noiseless channel} $\openone_{\textsc{a}\to\textsc{b}}:(\rho\in \mathscr{D}(\mathcal{H}_\textsc{a})) \mapsto(\rho\in \mathscr{D}(\mathcal{H}_\textsc{b}))$. 

The rate of the communication channel is the number of qubits transmitted per use of the channel \cite{wilde2011classical}, i.e.,
\begin{align}
    R \coloneqq \frac{1}{n}\log_2(\dim\mathcal{H}_\textsc{a})
\end{align}
\green{where $\log_2(\dim\mathcal{H}_\textsc{a})$ is the number of qubits in $\mathcal{H}_\textsc{a}$ and $n$ is the (possibly infinite) number of channel uses.} This rate is said to be achievable if for any $\delta,\epsilon>0$, there exists an $(n,R-\delta,\epsilon )$ for sufficiently large $n$. The \textit{quantum capacity} $\mathcal{Q}[\Phi]$ of the channel $\Phi$ is defined by the supremum of all achievable rates (see \cite{khatri2020principles,wilde2011classical} for more details). 

For practical calculations that we need, we consider a useful form of the quantum capacity in terms of an information-theoretic measure called \textit{coherent information}. Given a quantum channel $\Phi:\mathscr{D}(\mathcal{H}_\textsc{a})\to \mathscr{D}(\mathcal{H}_\textsc{b})$, let $\ket{\psi_{\textsc{ea}}}$ be a purification of Alice's input state $\rho_{\textsc{a}}$ and define
\begin{align}
    \rho_{\textsc{eb}} &\coloneqq (\openone\otimes\Phi)(\ket{\psi_\textsc{ea}}\!\bra{\psi_\textsc{ea}})\,.
    \label{eq: state-bob-environment}
\end{align}
Using $\rho_{\textsc{b}} \coloneqq \mathrm{Tr}_\textsc{e}(\rho_{\textsc{eb}}) =  \Phi(\rho_\textsc{a})$, the coherent information $I_c(\rho_{\textsc{a}},\Phi)$ is defined to be \cite{schumacher2002relative, lloyd1997capacity}
\begin{align}
    I_c(\rho_\textsc{a},\Phi) &\coloneqq S(\Phi(\rho_\textsc{a})) - S(\rho_{\textsc{eb}})\,.
\end{align}
We may sometimes write $I_c(\psi_{\textsc{ea}},\Phi)$ to make clear the joint state on $EA$. The coherent information depends on the input state $\rho_{\textsc{a}}$ and the channel $\Phi$ but is independent of its purification. From this, we can then define the coherent information of a quantum channel $I_c(\Phi)$ by maximizing over all input states:
\begin{align}
    I_c(\Phi)&\coloneqq \max_{\rho_{\textsc{a}}}I_c(\rho_\textsc{a},\Phi)\,.
    \label{eq: coherent information def}
\end{align}
Note that for qubit channels we have that $-1\leq I_c(\rho,\Phi)\leq 1$ and $0\leq I_c(\Phi)\leq 1$ \cite{wilde2011classical}. We can then rewrite the quantum capacity using the Lloyd-Shor-Devetak (LSD) formula \cite{lloyd1997capacity,shor2002capacity,devetak2003capacity}, given by \cite{khatri2020principles,wilde2011classical,gyongyosi2018surveycapacity,Holevo2020capacities}
\begin{align}
    \mathcal{Q}[\Phi]\coloneqq \lim_{n\to\infty}\frac{1}{n}I_{c}(\Phi^{\otimes n})\,,
    \label{eq: quantum-capacity-LSD}
\end{align}
where the RHS is also known as regularized coherent information.  

Despite the somewhat physically intuitive expression for $\mathcal{Q}[\Phi]$, the asymptotic nature of Eq.~\eqref{eq: quantum-capacity-LSD} makes quantum capacity very difficult to compute in most cases unless the quantum channel has very specific properties (e.g., when the quantum channel is \textit{additive}, namely $I_c[\Phi^{\otimes n}] = n I_c[\Phi]$, which is false in general \cite{khatri2020principles,wilde2011classical}). However, for us we do not actually need this: we follow the strategy adopted in \cite{Simidzija2020capacity}, which is to construct a quantum channel for qubits that has $I_c(\Phi) \to 1$ by appropriately adjusting the setup parameters. This works because the quantum coherent information is \textit{superadditive} \cite{hastings2009superadditivity,smith2008superadditive}:
\begin{align}
    I_c(\Phi^{\otimes N}) \geq NI_c(\Phi)
\end{align}
which immediately implies that
\begin{align}
    \mathcal{Q}[\Phi] \geq I_c(\Phi) \geq I_c(\rho_{\textsc{a}},\Phi)\,.
    \label{eq: coherent information lower bound}
\end{align}
For qubit channels, we have $0\leq \mathcal{Q}[\Phi]\leq 1$, thus we can construct an essentially perfect quantum channel if we can show that $I_c[\Phi]$ can be made arbitrarily close to 1. 

The interpretation of the above analysis is that in the limit where $\mathcal{Q}[\Phi]\to 1$, the quantum channel can effectively transmit quantum information (qubits) reliably in \textit{one use of the channel}. It is perhaps not surprising that this is possible if and only if the whole procedure is essentially a ``swapping operation'', i.e., Alice swaps her qubit `into the field', and then Bob swaps out the (approximate) embedded qubit in the field state to recover Alice's state. In effect, this suggests that the core of the protocol comprises of two components: (1) Bob being positioned correctly in spacetime in order to recover Alice's qubit, and (2) Alice and Bob being able to implement as high quality `SWAP gates' as possible\footnote{ In \cite{Simidzija2020capacity}, Alice's unitary is called ENCODING gate, and Bob's unitary is called DECODING gate. While the terminology has physical justification, it does not coincide with the encoding and decoding channels $\mathcal{E},\mathcal{D}$ even in the case of a single use of the communication channel $\Phi$ ($N=1$), as they are not part of $\Phi$ whose capacity is being computed. The ENCODING and DECODING gates were part of the definition of the communication channel.}.

As a final remark, we would like to point out that the existence of many superadditive quantum channels are in some sense a `curse' to finding a good characterization of quantum capacity --- we refer the readers to \cite{leditzky2023nonadditive} for one of the more recent summary and progress on this front. From Eq.~\eqref{eq: coherent information lower bound} we see that we need $0<I_c(\Phi) \leq 1$ for the channel to be able to transmit some quantum information and our construction capitalizes on this. However, the case for $I_c(\Phi) = 0$ is much more subtle, as many quantum channels exhibit strict superadditivity in both weak and strong sense \cite{leditzky2023nonadditive}. Weak superadditivity means that there exists $N$ such that $I_c(\Phi^{\otimes N}) > N I_c(\Phi)$, and strong superadditivity means $I_c(\Phi_1\otimes\Phi_2) > I_c(\Phi_1)+I_c(\Phi_2)$ for two channels $\Phi_1,\Phi_2$. In both cases, the `single-letter' formula for coherent information $I_c(\Phi)$ is not sufficient to calculate the quantum capacity,  although some exceptions are known. For instance, all \textit{degradable channels} (such as some regimes of the dephasing channels) are known to be weakly additive, in which case $\mathcal{Q}[\Phi] = I_c(\Phi)$ and hence non-positive channel coherent information would imply zero quantum capacity \cite{cubitt2008structure,khatri2020principles,wilde2011classical}. All \textit{anti-degradable channels} (such as all entanglement-breaking channels and certain regimes of quantum erasure channels) are known to have zero quantum capacity, though the argument does not rely on additivity and more towards no-cloning type arguments \cite{ruskai2003entanglementbreaking,cubitt2008structure}. Our construction (similarly for \cite{Simidzija2020capacity}) exploits the structure of the detector-field coupling to make the channel coherent information tractable.

\section{Protocol for the perfect relativistic quantum channel}
\label{sec: protocol}

In this section, we present the protocol for Alice and Bob to implement a perfect quantum channel through the field. While our protocol closely follows \cite{Simidzija2020capacity}, our approach naturally extends to arbitrary background spacetimes and choices of vacuum states of the field by recasting the protocol in the AQFT framework. 

A brief visual of the quantum channel is depicted in Fig.~\ref{fig: circuit diagram}. Mathematically, the quantum channel from Alice to Bob generically takes the form
\begin{equation}
    \Phi(\rho_{\textsc{a}}) = \mathrm{Tr}_{\textsc{a}, \phi} \left[\hat U_{\textsc{b}}^{\phantom{\dagger}} \hat U_{\textsc{a}}^{\phantom{\dagger}} \left( \rho_{\textsc{a}} \otimes \ket{0}\!\bra{0} \otimes \rho_{\textsc{b}} \right) \hat U_{\textsc{a}}^\dagger \hat U_{\textsc{b}}^\dagger\right] \,,
    \label{eq: quantum channel}
\end{equation}
where $\ket{0}$ is the field vacuum\footnote{The state $\ket{0}$ can be replaced with any GNS vector state $\ket{\Omega_\omega}$ associated with any of the quasifree representations, which accounts for the existence of many unitarily inequivalent ground states in QFT \cite{Khavkhine2015AQFT,fewster2019algebraic,ruep2021weakly,wald1994quantum,KayWald1991theorems}.}, $\rho_{\textsc{b}}$ is Bob's initial state, and $\hat U_{\textsc{a}}$ and $\hat U_{\textsc{b}}$ are the unitaries implemented by Alice and Bob through the UDW interaction with the field. The task is to construct a perfect quantum channel that can faithfully transmit quantum information to arbitrary accuracy.

Essentially, the basic idea is for Alice to (approximately) perform a `SWAP' gate\footnote{Technically speaking, the qubit and the quantum field have different Hilbert space dimensions but we will call them `SWAP' gates anyway since it would be the usual SWAP gate if we were to replace the quantum field with a third qubit. The fact that we are ``embedding'' the qubit state into the (much higher-dimensional) field can be viewed as a form of encoding operation, which justifies calling Alice's unitary ENCODING gate in \cite{Simidzija2020capacity}, and similarly Bob's unitary can be viewed as decoding operation from the field to his qubit system.} on her qubit state to the field, effectively encoding it into some superposition of approximately orthogonal coherent states, so that Bob can extract the qubit state by performing another swap operation. The entire problem thus reduces to analyzing under what circumstances the UDW interaction allows for nearly perfect `SWAP' operations on both parties which will depend on both the coupling strengths and the supports of the interactions. In particular, we will see that Bob can only perform a perfect `SWAP' gate if he is causally connected to Alice and the details of the interactions determine the quality of the swap operation. 

More precisely, Alice starts with an arbitrary state $\ket{\psi_{\textsc{a}}}\coloneqq c_1 \ket{+_z} + c_2 \ket{-_z}$, equivalently $\rho_{\textsc{a}}=\ket{\psi_{\textsc{a}}}\!\bra{\psi_{\textsc{a}}}$, and encodes that state into the field using a `SWAP' gate. Bob should then be able to decode the state from the quantum field into his own Hilbert space. If the `SWAP' gates $\hat U_{\textsc{a}}$ and $\hat U_{\textsc{b}}$ are implemented correctly, then in the idealized limit we should have a perfect quantum channel $\Phi(\rho_{\textsc{a}}) = \openone_{\textsc{a}\to\textsc{b}}(\rho_{\textsc{a}}) \equiv \rho_\textsc{a}\in\mathcal{H}_\textsc{b}$, noting that this channel is not trivial because it corresponds to sending the same state $\rho_\textsc{a}$ from Alice's Hilbert space to Bob's Hilbert space. We start by discussing how Alice encodes her state into the field:

\begin{figure}
    \centering
    \includegraphics[width=0.9\linewidth]{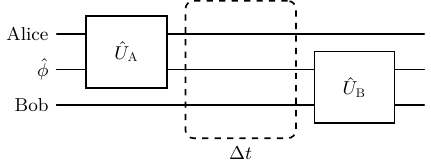}
    \caption{A perfect quantum channel from Alice to Bob through a quantum field $\hat \phi$ consists of two unitary operators: Alice encoding her state into the field and then Bob decoding the state. The time in $\Delta t$ between the two gates is a general $t$ that foliates the spacetime.  }
    \label{fig: circuit diagram}
\end{figure}

\begin{proposition}
    \label{thm: encoding}
    Let  $\ket{\psi_{\textsc{a}}} = c_1 \ket{+_z} + c_2 \ket{-_z}$ be Alice's initial state and let $\ket{0}$ be the initial vacuum state of the field. Consider Alice's unitary operation at time $t=t_{\textsc{a}}$ given by\footnote{The usual convention would be $\hat{U}_{\textsc{a}} =  e^{-i \hat \sigma_x \hat \phi(f_2)} e^{-i \hat \sigma_z \hat \phi(f_1)}$, but we opt for the positive sign for the exponent for convenience as many expressions appear more symmetric (this was also implicitly done in \cite{Simidzija2020capacity}). This amounts to the replacement $\hat{h}_{I}\to -\hat{h}_{I}$, or equivalently $f_i\to -f_i$. }
    \begin{equation}
        \hat U_{\textsc{a}} = e^{i \hat \sigma_x \hat \phi(f_2)} e^{i \hat \sigma_z \hat \phi(f_1)}\,.
        \label{eq: Alice unitary}
    \end{equation}
    Then Alice can approximately swap the state $\ket{\psi_{\textsc{a}}}$ into the field provided the following conditions are satisfied 
    \begin{subequations}\label{eq: Alice conditions}
    \begin{align}
        E(f_1, f_2) &= \frac{\pi}{4} \textnormal{ mod } 2\pi 
        \,,\label{eq: fine-tuning condition}\\
        E(f_1, f_2)^2 &\gg \mathsf{W}(f_2, f_2)\,. 
        \label{eq: strong-coupling condition}
    \end{align}
    \end{subequations}
    The resulting state after Alice's operation reads 
    \begin{align}
        \hat U_{\textsc{a}} (c_1 \ket{+_z} &+ c_2 \ket{-_z}) \ket{0}\notag\\ 
        &= \ket{+_y} (c_1 \ket{+\alpha} - i c_2 \ket{-\alpha})\,,
    \end{align}
    where $\ket{\pm \alpha}$ are coherent states of the quantum field defined by
    \begin{equation}
        \ket{\pm \alpha} \coloneqq e^{\pm i \hat \phi(f_1)} \ket{0}.
        \label{eq: coherent state def}
    \end{equation}
\end{proposition}
\begin{proof}
    In order to prove this statement, we will explicitly apply Eq.~\eqref{eq: Alice unitary} to Alice's qubit and the field and show that the qubit state is encoded under the conditions in Eqs.~\eqref{eq: fine-tuning condition} and \eqref{eq: strong-coupling condition}.

    The first unitary in Eq.~\eqref{eq: Alice unitary} gives 
    \begin{align}
        e^{i \hat \sigma_z \hat \phi(f_1)}\ket{\psi_{\textsc{a}}}\ket{0} 
        &= c_1 \ket{+_z} \ket{+\alpha}  + c_2 \ket{-_z} \ket{-\alpha}.
    \end{align}
    Next, we need to apply the second unitary $e^{i \hat \sigma_x \hat \phi(f_2)}$ but we first note that 
    \begin{equation}
        \hat \phi(f_2) \ket{\pm \alpha} = \pm E(f_1, f_2) \ket{\pm \alpha} + e^{\pm i \hat \phi(f_1)} \hat \phi(f_2) \ket{0}\,,
    \end{equation}
    which follows directly from the Baker-Campbell-Haussdorff (BCH) formula and the CCR. Now we see that if 
    \begin{equation}
        E(f_1, f_2)^2 \gg \bra{0}{\hat \phi(f_2) \hat \phi(f_2)}\ket{0} = \mathsf{W}(f_2, f_2)
    \end{equation}
    then $\ket{\pm \alpha}$ become approximate eigenvectors of $\hat{\phi}(f_2)$
    \begin{equation}
        \hat{\phi}(f_2) \ket{\pm \alpha} \approx \pm E(f_1, f_2) \ket{\pm \alpha}. \label{eq: approx eigenvector}
    \end{equation}
    If we pick the eigenvalues to be $E(f_1, f_2) = \frac{\pi}{4}\text{ mod }2\pi$, then the action of Alice's full unitary $\hat{U}_{\textsc{a}}$ gives
    \begin{align}
        \hat U_{\textsc{a}} & (c_1 \ket{+_z} + c_2 \ket{-_z}) \ket{0}\notag\\
        &= e^{i \hat \sigma_x \hat \phi(f_2)}(c_1 \ket{+_z} \ket{+\alpha}  + c_2 \ket{-_z} \ket{-\alpha})\notag\\
        &\approx c_1 e^{i \frac{\pi}{4} \hat\sigma_x } \ket{+_z} \ket{+\alpha}  
        + c_2 e^{-i \frac{\pi}{4} \hat \sigma_x } \ket{-_z} \ket{-\alpha} \notag\\
        &= \ket{+_y} \left(c_1 \ket{+\alpha} - i c_2 \ket{-\alpha}\right)\,,
        \label{eq: logical-qubit}
    \end{align}
    where we used $e^{i \frac{\pi}{4}\hat \sigma_x } \ket{+_z} = \ket{+_y}$ and $e^{-i \hat \sigma_x \frac{\pi}{4}} \ket{-_z} = -i\ket{+_y}$ in the last equality.

\end{proof}

The second of Alice's conditions in Eq.~\eqref{eq: strong-coupling condition} can be interpreted as a ``strong coupling condition'' that requires Alice's coupling through the field via the interaction profile $f_1$ be much stronger than through $f_2$. In \cite{Simidzija2020capacity}, this corresponds to choosing $f_2$ such that $\hat{\phi}_2$ corresponds to conjugate momentum $\hat{\pi}(F_{\textsc{a}})$ and the condition translates to Alice coupling to the field $\hat{\phi}(F_{\textsc{a}})$ much more strongly than to $\hat{\pi}(F_{\textsc{a}})$ (\textit{cf.} Eq.~\eqref{eq: flat-space-Alice-unitary}). \green{More intuitively, the two coherent states $\ket{\pm \alpha}$ are never orthogonal, $\langle+\alpha|\!-\!\alpha\rangle \neq 0$, but condition \eqref{eq: strong-coupling condition} ensures that the overlap is sufficiently small that the two states are close to orthogonal.} We can interpret Eq.~\eqref{eq: fine-tuning condition} as a ``fine-tuning condition'' that performs near exact $\frac{\pi}{4}$ rotations in the Bloch sphere so that Alice's second unitary decouples her detector from the quantum field. \green{$E(f_1, f_2) = \frac{\pi}{4}\textrm{ mod }2\pi$ are the unique values such that $e^{i E(f_1, f_2) \hat \sigma_x } \ket{+_z} \propto e^{-i  E(f_1, f_2) \hat \sigma_x} \ket{-_z}$ and the qubit decouples from the field \footnote{Alternative values of $E(f_1, f_2)$ could be used for other protocols, such at $E(f_1, f_2) = 0\textrm{ or }\frac{\pi}{2}$ leaving Alice's qubit entangled with the coherent states of the field.}}


\green{We remark that our procedure can be viewed as a minimal, simple example where the interaction unitary can be computed non-perturbatively without truncating any Dyson series expansion, as is usually the case for generic choice of switching functions. It is known that if we resort to leading-order weak coupling regime, then the quantum capacity is at most $\mathcal{O}(\lambda_A\lambda_B)$ \cite{Simidzija2020capacity}, thus any attempt in obtaining close-to-perfect quantum channel capacity requires us to go beyond perturbation theory. }

\green{We emphasize the fact that Alice can encode her qubit into the field using any spatial profiles $F_A$ as long as conditions \eqref{eq: Alice conditions} are satisfied. Therefore, in principle there is a fair amount of freedom in specifying the details of the interactions and coupling strengths provided they satisfy \eqref{eq: Alice conditions}. However, as will be discussed in Section \ref{subsec: intermediate scenarios}, it very difficult to analyze the capacity of quantum channels in the $0 < I_c(\Phi) < 1$ regime, thus in Section~\ref{sec: analysis} we will provide some analytic studies on certain limiting regimes, in particular the limit where $I_c\to 1$ where we have a perfect noiseless communication channel.  }

Now we move on to discussing how at a later time, Bob can decode Alice's state from the quantum field into his own qubit detector. 
Even in the idealized limit where Alice effectively implements a `SWAP' gate through unitary $\hat{U}_\textsc{a}$, Bob cannot simply perform the inverse unitary $\hat{U}_\textsc{a}^{-1}$ for two reasons. First, Bob is in the causal future of Alice, so the smearing functions of Bob's interactions cannot be at the same time slice as $f_i$ of Alice. More importantly, $\hat{U}_\textsc{a}^{-1}$ will simply undo what Alice does resulting in the trivial identity channel $ \openone_\textsc{a}(\rho_\textsc{a}) = \rho_\textsc{a}\in\mathcal{H}_\textsc{a}$ which corresponds no interactions with the field at all, instead of a noiseless communication channel $\openone_{\textsc{a}\to\textsc{b}}(\rho_\textsc{a})=\rho_\textsc{a}\in \mathcal{H}_\textsc{b}$ through the field. In other words, we want that Alice's operations defined through $\hat \phi(f_i)$ be defined at time $t=t_{\textsc{a}}$, while we want Bob to decode from the field using only field operators $\hat \phi(g_i)$ localized at time $t=t_{\textsc{b}}$. Bob's unitary will thus take the form
\begin{equation}
    \hat U_{\textsc{b}}^{\phantom{\dagger}} \coloneqq  e^{-i \hat \sigma_z^\textsc{B} \hat \phi(g_1)} e^{-i \hat \sigma_x^\textsc{B} \hat \phi(g_2)}\,,
    \label{eq: bobs unitary}
\end{equation}
where the superscript `B' emphasizes that Bob acts on $\mathcal{H}_\textsc{b}$ and  $g_i$ are chosen such that $\hat{U}_\textsc{b}$ implements a `SWAP' gate to decode the qubit from the field at time $t_\textsc{b}$. For this, 
we need a way to write the operators $\hat \phi(f_i)$ using different spacetime smearing functions $g_i$ defined at a different Cauchy slice. To achieve this we use the following theorem:
\begin{proposition}
    \label{thm: bob smearing}
    Suppose we have two smeared field operators $\hat \phi(f)$ and $\hat \phi(g)$, then the two operators are equal $\hat \phi(f) = \hat \phi(g)$ if the smearing functions satisfy
    \begin{equation}
        Ef = Eg. \label{eq: bobs smearing}
    \end{equation}
\end{proposition}
\begin{proof}
    The proof is almost immediate from the symplectic definition of smeared operators in Eq.~\eqref{eq: symplectic smearing}. If we assume $Ef = Eg$\footnote{The converse statement is false. From the axioms of AQFT, we also have that $\hat{\phi}(\hat{P}f)=0$ which is the KG equation. Thus $\hat{\phi}(g)$ is only unique up to homogeneous solution of the KG equation, i.e., $g\sim g+\hat{P}f$.}, then 
    \begin{equation}
        \hat \phi(f) = \sigma(Ef, \hat \phi) = \sigma(Eg, \hat \phi) = \hat \phi(g)\,.
    \end{equation}
\end{proof}
\noindent{Given that Alice and Bob's smearing functions $f_i$ and $g_i$ are related by Eq.~\eqref{eq: bobs smearing}  and given that Bob's qubit starts in the $\ket{+_y}$ state (the state that Alice's qubits is left as after $\hat U_\textsc{A}$), we have achieved our goal of creating a perfect quantum channel from Alice to Bob.}

Some remarks are in order in view of the above Proposition~\ref{thm: bob smearing}. First, intuitively Alice's state is encoded into the field amplitude that propagates into the causal future of Alice's interaction region $ \mathscr{J}^+(\supp f_i)$, where $\mathscr{J}^+(\supp(f_i))$ is the causal future of $\supp(f_i)$; for Bob to decode that state, his qubit has to be coupled everywhere where the initial data has propagated at a given time $t=t_{\textsc{b}}$, which in turn forces Bob's detector to interact with the field with the correct support and coupling strengths. In particular, $\supp(g_i)$ should essentially be $\Sigma_{t_{\textsc{b}}} \cap \supp (Ef_i)$, where $\supp(Ef_i)\subseteq \mathscr{J}^+(\supp f_i)$. For simplicity, we assume that $\supp(f_1)=\supp(f_2)$ and hence $\supp(g_1)=\supp(g_2)$, which in the delta-coupling scenario is easily enforced as we are fixing the spatial profile of the interaction on Alice's side --- this, in turn, fixes the needed spatial profile for Bob's coupling. 

Second, despite its simplicity and simple interpretation, Eq.~\eqref{eq: bobs smearing} is largely impractical to use for actually calculating Bob's spatial smearing functions $G_i(\bb x)$ for a given time $t = t_{\textsc{b}}$. Using the eigenmode basis provided in Eq.~\eqref{eq: basis decomposition of E} and the momentum space smearing functions in Eq.~\eqref{eq: f-fourier}, the condition in Eq.~\eqref{eq: bobs smearing} can be equivalently written in `momentum space' as
\begin{equation}
    f_{\bb k} = g_{\bb k}.
\end{equation}
These equations in momentum space are often easier to work with since in many simple cases (such as all FRW spacetimes) the spatial part of the mode function $u_{\bb k}(\sx) \propto e^{i \bb k \cdot \bb x}$ and we can use Fourier transforms, which is similar to what was originally done in \cite{Simidzija2020capacity} where Bob's spatial smearing functions were related to Alice's in momentum space. More generally, in generic curved spacetimes, we do not have a simple plane-wave basis, so it is essential that all the conditions we demand for the perfect quantum channel are given only in terms of the correlation functions of the field without reference to Fourier transforms. 

Last but not least, note that there is an asymmetry in how we quantify the quality of Alice's and Bob's respective `SWAP' gates, i.e., how close their unitaries are in achieving the ideal `SWAP' operation. On the one hand, Alice needs two smearing functions $f_1$ and $f_2$ that satisfy the conditions in Proposition~\ref{thm: encoding} to approximate the ideal `SWAP'. On the other hand, Bob needs two smearing functions $g_1$ and $g_2$ that are related to Alice's via Proposition~\ref{thm: bob smearing} and hence automatically satisfy the conditions in Proposition~\ref{thm: encoding}. However, Proposition~\ref{thm: encoding} alone is not enough for the perfect quantum channel: if  $Ef_i \neq Eg_i$, Bob will \textit{not} be able to decode Alice's state even though by construction Bob's smearing functions $g_1$ and $g_2$ satisfy Eqs.~\eqref{eq: fine-tuning condition} and \eqref{eq: strong-coupling condition}. In a sense, to decode the state using the `SWAP' operation with the field Bob needs an extra condition, namely that Bob needs to interact ``at the correct spacetime regions''. The reason for this asymmetry is that the Hilbert space dimension of the field $\hat \phi$ is (much) larger than that of the detectors: it is much simpler for Alice to encode her state in a larger system but Bob will need additional information to decode it. 

In summary, we have constructed a perfect relativistic quantum channel that allows Alice to send her qubit state to Bob through the quantum field. The quantum channel consists of two steps:
\begin{enumerate}[leftmargin=*,label=(\arabic*)]
    \item At time $\tau= \tau_{\textsc{a}}$, Alice encodes her qubit into the field using the unitary $\hat U_{\textsc{a}}$ such that the conditions in Proposition~\ref{thm: encoding} are satisfied.
    
    \item At a later time $\tau = \tau_{\textsc{b}}$, Bob decodes the state from the field using the unitary $\hat U_{\textsc{b}}$ defined in Eq.~\eqref{eq: bobs unitary}. However, to implement this protocol, Bob's detector must interact across a spacetime region that satisfies the conditions specified in Proposition~$\ref{thm: bob smearing}$.
\end{enumerate}
As the proper times are not the same in both trajectories, in practice one can use the ambient global coordinate systems and foliation to help with the time-ordering of these operations. Once times $\tau_{\textsc{a}}$ and $\tau_{\textsc{b}}$ are specified, we demand that Alice's spatial smearing functions $F_i(\bb x)$ fulfill conditions \eqref{eq: fine-tuning condition} and \eqref{eq: strong-coupling condition}, and Bob's functions $G_i(\bb x)$ fulfill the condition \eqref{eq: bobs smearing}\footnote{In principle, Alice or Bob could each need up to four spatial smearing functions since $f_i$ or $g_i$ could include both $\delta(t)$ and $\delta'(t)$ terms. For example, in flat spacetime, we can have $f_1 = \delta(t) F_{11}(\bb x) + \delta'(t)F_{12}(\bb x)$ (similarly for $f_2$), but for simplicity these can be chosen to be equal. Both $g_1$ and $g_2$ are fixed by $f_1$ and $f_2$ in the ideal case.}. 

 In order to simplify our subsequent discussions, we make a physically reasonable assumption that Alice's interaction profiles $f_i$ are completely fixed except their coupling constants, i.e., in FNC we have
\begin{align}
    f_i = \lambda_i \chi_i(\tau)F(\bar{\bx})
\end{align}
where $\chi_i$ are either the Dirac delta function $\delta$ or its distributional derivative $\delta'$ and the spatial profile is the same for both $f_1,f_2$ (hence they have the same supports $\supp(f_i) = \supp (F)\subseteq \M$). In other words, Alice only varies her interaction profile $f_i$ by tuning the coupling strengths $\lambda_i$ without changing the `shape' of the interaction profile, and this is sufficient to satisfy the perfect `SWAP' conditions \eqref{eq: strong-coupling condition}-\eqref{eq: fine-tuning condition} in Proposition~\ref{thm: encoding}.

\section{Analytic Tests of the Quantum Channel}
\label{sec: analysis}

In this section, we will prove that the protocol given in Section~\ref{sec: protocol} produces a perfect quantum channel in the correct limit without any approximations like Eq.~\eqref{eq: approx eigenvector}. This proves that entanglement can always be perfectly transmitted between two localized qubits through a quantum field regardless of the background spacetime geometry. We also discuss the case when Alice and Bob can in principle perfectly implement the SWAP gates with the field but are spacelike separated from each other. For this channel, the quantum and classical channel capacity are both zero, as we expect with relativistic causality.

\subsection{Ideal case: perfect quantum channel}
\label{ssec: perfect conditions}

Recall that our task is to show that the channel coherent information $I_c(\Phi)$ can be made arbitrarily close to 1, thus proving that the (relativistic) quantum channel can transmit quantum information perfectly. This means computing the coherent information $I_c(\rho_{\textsc{a},0},\Phi)$ and maximizing over all possible input states on Alice's side. Using the definition of coherent information \eqref{eq: coherent information def}, it turns out that it is sufficient to consider the case when Alice's initial state is the maximally mixed state $\rho_{A, 0} = \frac12 \openone$ and Bob initial state is $\ket{+_y}$ state in order to show that we can attain maximum channel coherent information $I_c(\Phi)\to 1$ and hence maximal quantum capacity. We stress that this does not mean Alice should send a maximally mixed state to Bob, which is ``garbage'' as a message and is even classical, since in practice Alice wants to be able to send pure states properly. Also, in general, it is not true that $I_c(\Phi)$ is maximized by the maximally mixed state: for example, the completely depolarizing channel $\Phi(\rho) = \openone/2$ has $I_c(\openone/2,\Phi) = -1$ but the channel coherent information $I_c(\Phi)=0$ is attained by any pure state.

In order to show that $I_c(\Phi)\to 1$ is attainable, we need to introduce a qubit of the environment $E$ that purifies Alice's maximally mixed state $\openone/2$, such that the joint state of Alice and the environment is the maximally entangled state $\ket{\psi_\textsc{ea}} = \frac{1}{\sqrt{2}} \left(\ket{-_z} \ket{+_z} + \ket{+_z} \ket{-_z} \right)_\textsc{ea}$. As discussed in Section~\ref{sec: setup}, a perfect quantum channel will preserve all the entanglement between Alice and the environment $EA$ to Bob and the same environment $EB$. To compute $I_c (\openone/2, \Phi)$, we must evaluate the state $\rho_\textsc{eb}$ in Eq.~\eqref{eq: state-bob-environment}. 

Writing out the expression for $\rho_{\textsc{eb}}$ in terms of the unitaries $\hat U_{\textsc{a}}$ and $\hat U_{\textsc{b}}$, we obtain:
\begin{equation}
    \rho_\textsc{eb} \!=\! \mathrm{Tr}_{A\phi} [\hat U_{\textsc{b}} \hat U_{\textsc{a}} \left( \ket{\psi_\textsc{ea}}\!\bra{\psi_\textsc{ea}} \!\otimes\! \ket{0}\!\bra{0} \!\otimes\! \ket{+_y}\!\bra{+_y} \right) \hat U_{\textsc{a}}^\dagger \hat U_{\textsc{b}}^\dagger].
\end{equation}
The simplest way of continuing is to decompose the unitaries using projections onto the eigenvectors of the qubit operators. That is, for $\hat U_{\textsc{a}}$ we write
\begin{align}
    \hat U_{\textsc{a}} &= \sum_{x,z \in \{\pm\} } \hat P_x \hat P_z \otimes e^{i x \hat \phi(f_2)} e^{i z \hat \phi(f_1)} \nonumber \\
    &= \sum_{x,z \in \{\pm\} } \hat P_x \hat P_z \otimes W(xEf_2) W(zEf_1)
\end{align}
where $\hat P_i$ are the projectors onto the eigenstates of $\hat \sigma_i$ and in the last line we used the Weyl algebra notation introduced in Eq.~\eqref{eq: Weyl-generator}. For Bob's unitary, we write
\begin{equation}
    \hat U_{\textsc{b}} = \sum_{x,z \in \{\pm\} } \hat P_z \hat P_x \otimes W(-zEf_1) W(-xEf_2)
\end{equation}
where the projection operators are understood to act on Bob's system. We have written Bob's unitary as in Eq.~\eqref{eq: bobs unitary}, the imperfect case is discussed in the following sections. We also expand the Bell state as $\ket{\psi_\textsc{ea}} = \frac{1}{\sqrt{2}} \sum_j \ket{-j_z} \ket{j_z}$ with $j \in \{ \pm \}$ to get the expression
\begin{align}
    \rho_\textsc{eb} &= \frac12 \sum_{j,k,x_i, z_i}  \omega(\hat{O}) \bra{k_z}{ \hat P_{z_1} \hat P_{x_1} \hat P_{x_4} \hat P_{z_4}}\ket{j_z} \nonumber \\
    &\times \ket{-j_z}_{\textsc{e}}\bra{-k_z} \otimes \hat P_{z_3}\hat P_{x_3} \ket{+_y}_{\textsc{b}}\bra{+_y} \hat P_{x_2}\hat P_{z_2}
\end{align}
where we write $ \omega(\hat{O}) \equiv \braket{0|\hat{O}|0}$ and
\begin{align}
    \hat O &= W(-z_1Ef_1) W(-x_1Ef_2) W(x_2Ef_2) W(z_2Ef_1) \nonumber\\
    & \times W(-z_3Ef_1) W(-x_3Ef_2) W(x_4Ef_2) W(z_4Ef_1),
\end{align}
$x_i$ stands for $x_1, x_2, x_3, x_4$, similarly for $z_i$, and all sums are taken over the set $\{ \pm \}$. Provided that we can obtain an expression for $ \omega(\hat{O}) $, a computer algebra system can straightforwardly sum over the $2^{10}$ terms in the sum. We start by redefining summation indices $x_1 \mapsto -x_1$, $z_1 \mapsto -z_1$, $x_3 \mapsto -x_3$, and $z_3 \mapsto -z_3$ such that  
\begin{align}
    \rho_\textsc{eb} &= \frac12 \sum_{j,k,x_i, z_i} \omega(\hat{O}) \bra{k_z}{ \hat P_{-z_1} \hat P_{-x_1} \hat P_{x_4} \hat P_{z_4}}\ket{j_z} \nonumber \\
    &\times \ket{-j_z}_{\textsc{e}}\bra{-k_z} \otimes \hat P_{-z_3}\hat P_{-x_3} \ket{+_y}_{\textsc{b}}\bra{+_y} \hat P_{x_2}\hat P_{z_2}
\end{align}
and $\hat O$ only contains terms without minus signs
\begin{align}
    \hat O &= W(z_1Ef_1) W(x_1Ef_2) W(x_2Ef_2) W(z_2Ef_1) \nonumber\\
    & \times W(z_3Ef_1) W(x_3Ef_2) W(x_4Ef_2) W(z_4Ef_1).
\end{align}
In \cite{Simidzija2020capacity}, $\bra{0}{\hat O}\ket{0}$ is evaluated using the Baker-Campbell-Hausdorff (BCH) formula and Wick's theorem. In this paper, we adopt a much simpler approach using the AQFT framework. It is straightforward (albeit somewhat tedious) to apply the third Weyl condition $W(Ef) W(Eg) = e^{- \frac{i}{2}E(f,g)} W(E(f+g))$ seven times to obtain
\begin{align}
    \hat{O} &= W(E((z_1\!+\!z_2\!+\!z_3\!+\!z_4)f_1+(x_1\!+\!x_2\!+\!x_3\!+\!x_4)f_2)) \nonumber\\
    &\times e^{-\frac{i}{2} (x_1+x_2)(z_1-z_2-z_3-z_4) E(f_1, f_2)} \nonumber \\
    &\times e^{-\frac{i}{2} (x_3+x_4)(z_1+z_2+z_3-z_4) E(f_1, f_2)}
\end{align}
where only one Weyl operator remains multiplied by a phase factor. 

Before we evaluate the expectation value of the remaining Weyl operator, it will be convenient to consider the decomposition of a Wightman two-point function into its real (symmetric) and imaginary (anti-symmetric) parts, i.e.,
\begin{align}
    \mathsf{W}(f,g) &= \frac{1}{2}\omega(\{\hat{\phi}(f),\hat{\phi}(g)\})+\frac{1}{2}\omega([\hat{\phi}(f),\hat{\phi}(g)])\notag\\
    &= \frac{1}{2}\mathsf{H}(f,g) + \frac{i}{2}E(f,g)\,,
\end{align}
where we used the CCR to obtain the second term and $\mathsf{H}(f,g) \coloneqq \omega(\{\hat{\phi}(f),\hat{\phi}(g)\}) \equiv \frac{1}{2}\mu(Ef,Eg)$ is the smeared Hadamard two-point distribution of the field (\textit{cf.} Eq.~\eqref{eq: real-bilinear IP}). This decomposition is important because only the symmetric part $\mathsf{H}(f,g)$ is \textit{state-dependent}, while $E(f, g)$ is independent of the state. This is useful as it tells us which parts of our analysis depend on the field state. Furthermore, while we have assumed that the field starts in some vacuum state for simplicity, but the protocol in Section~\ref{sec: protocol} automatically extends to all quasi-free states such as thermal states \cite{KayWald1991theorems}, squeezed vacuum state \cite{tjoa2023nonperturbative}, and with slight modifications to more general Gaussian states such as coherent states (see, e.g., \cite{ruep2021weakly}). 

We can now use the definition of quasifree state $\omega$ in Eq.~\eqref{eq: quasifree-definition} for the vacuum state to evaluate the expectation value $\omega(\hat{O})$ in terms of Wightman and Hadamard functions giving us
\begin{align}
    \omega(\hat{O}) &= e^{-\frac12 (z_1+z_2+z_3+z_4)^2 \mathsf{W}(f_1, f_1)-\frac12 (x_1+x_2+x_3+x_4)^2 \mathsf{W}(f_2, f_2)}  \nonumber\\
    &\times e^{-\frac12 (x_1+x_2+x_3+x_4)(z_1+z_2+z_3+z_4)\mathsf{H}(f_1, f_2)} \nonumber \\
    &\times e^{-\frac{i}{2} (x_1+x_2)(z_1-z_2-z_3-z_4) E(f_1, f_2)} \nonumber \\
    &\times e^{-\frac{i}{2} (x_3+x_4)(z_1+z_2+z_3-z_4) E(f_1, f_2)}\,.
\end{align}
Note that this expression is valid automatically for any quasifree states such as squeezed vacuum states and (squeezed) thermal states \cite{tjoa2023nonperturbative}.  The density matrix $\rho_{\textsc{eb}}$ in the ordered basis $\set{\ket{+_z+_z}_{\textsc{be}}, \ket{+_z-_z}_{\textsc{be}}, \ket{-_z+_z}_{\textsc{be}}, \ket{-_z-_z}_{\textsc{be}}}$ now reads
\begin{equation}
    \rho_{\textsc{eb}} = \begin{pmatrix}
        P_- & 0 & A & C\\ 
        0 & P_+ & X & B\\
        A^* & X^* & P_+ & 0\\
        C^* & B^* & 0 & P_-
    \end{pmatrix}
\end{equation}
where
{\allowdisplaybreaks
\begin{align*}
    P_\pm &= \frac{1}{4}\left(1 \pm e^{-2 \mathsf{W}(f_1, f_1) } \sin(2E(f_1, f_2)) \right)\\
    X &= \frac{1}{4}  \sin(2E(f_1, f_2)) \left( e^{- 2 \mathsf{W}(f_2, f_2)} +  \sin(2E(f_1, f_2)) \right)\\
    A &= - \frac{i}{4} e^{ -2 \mathsf{W}(f_1, f_1) } \cos(2E(f_1, f_2))\\
    &\times \left( \sin(2E(f_1, f_2)) -e^{ -2W(f_2, f_2)} \sinh(4\mathsf{W}(f_2, f_1)) \right)\\
    B &= -\frac{i}{4} e^{ - 2 \mathsf{W}(f_1, f_1) } \cos(2E(f_1, f_2)) \\
    &\times \left(\sin(2E(f_1, f_2)) + e^{-2\mathsf{W}(f_2, f_2)} \cosh(4 \mathsf{W}(f_1, f_2)) \right)\\
    C &= \frac{1}{4} e^{ - 8 \mathsf{W}(f_1, f_1)} \sin(2E(f_1, f_2)) \\
    &\times \left( \sin(2E(f_1, f_2)) - e^{-2 \mathsf{W}(f_2, f_2)}\cosh(4\mathsf{H}(f_1, f_2)) \right)\,.
\end{align*}
}

Now that we have an analytic expression for $\rho_{\textsc{eb}}$, we can use our two conditions in Proposition~\ref{thm: encoding} and verify that the coherent information can be made arbitrarily close to 1. First we use the fine-tuning condition $E(f_1, f_2) = \frac{\pi}{4}$ which simplifies $\rho_{\textsc{eb}}$  considerably to
\begin{equation}
    \rho_{\textsc{eb}} = \begin{pmatrix}
        \tilde P_- & 0 & 0 & \tilde C\\ 
        0 & \tilde P_+ & \tilde P_+ & 0\\ 
        0 & \tilde P_+ & \tilde P_+ & 0\\
        \tilde C & 0 &0 & \tilde P_-
    \end{pmatrix}
\end{equation}
where
\begin{align*}
    \tilde P_\pm &= \frac{1}{4} \left(1 \pm e^{-2\mathsf{W}(f_2, f_2)}\right)\\
    \tilde C &= \frac14 e^{-8 \mathsf{W}(f_1, f_1)} \left(e^{-2\mathsf{W}(f_2, f_2)} \cosh(4\mathsf{H}(f_1, f_2)) +1\right).
\end{align*}

Next, we use the strong coupling condition $E(f_1, f_2)^2 \gg \mathsf{W}(f_2, f_2)$. Recall that we are assuming for simplicity that on Alice's side, her interaction profile will be kept fixed except for the coupling strengths, thus only $\lambda_1,\lambda_2$ are adjustable parameters. In this case, the strong-coupling condition essentially means that $\lambda_1\gg\lambda_2$ and hence we have the hierarchy of scales
\begin{subequations}
        \begin{align}
        \mathsf{W}(f_1,f_1) \sim\mathcal{O}(\lambda^2_1) &\gg E(f_1,f_2)^2\sim \mathcal{O}(\lambda^2_1\lambda^2_2) \\
        E(f_1,f_2)^2\sim \mathcal{O}(\lambda^2_1\lambda^2_2) &\gg \mathsf{W}(f_2,f_2)\sim \mathcal{O}(\lambda^2_2)\\
     \mathsf{W}(f_1,f_1)\sim\mathcal{O}(\lambda^2_1) &\gg \mathsf{H}(f_1,f_2) \sim \mathcal{O}(\lambda_1\lambda_2)
    \end{align}
\end{subequations}
With this hierarchy, we can express $\lambda_1$ in units of $\lambda_2$ so that  $\lambda_1\coloneqq c\lambda_2$ for some $c\geq 0$, which is also convenient since in general $\lambda_i$ has physical units\footnote{The unit of $\lambda_i$ is $[\text{Length}]^{\frac{n-3}{2}}$, and we have assumed for simplicity that $\lambda_i$ have the same units since $f_i$ are chosen to be (\textit{cf.} the end of Section~\ref{subsec: UDW})}. We can thus write
\begin{subequations}
\begin{align}
    \mathsf{W}(f_1,f_1) &= c^2\lambda_2^2 \mathsf{w}_1\,,\quad \mathsf{W}(f_2,f_2) = \lambda_2^2 \mathsf{w}_2\,,\\
    E(f_1,f_2) &= c\lambda^2_2\mathsf{e}\,,\qquad \mathsf{H}(f_1,f_2) = c\lambda_2^2\mathsf{h}\,, 
\end{align}
\end{subequations}
where $\mathsf{w}_1,\mathsf{w}_2,\mathsf{e},\mathsf{h}$ are fixed by the spacetime smearings. Proposition~\ref{thm: encoding} then translates to
\begin{align}
    \lambda_2^2(c^2\mathsf{e}^2)\gg \mathsf{w}_2\,,\quad \lambda_2^2(c\mathsf{e}) = \frac{\pi}{4}\mod 2\pi\,.
    \label{eq: dimless-proposition-1}
\end{align}
Note that the fine-tuning condition can still be satisfied: for example, for a fixed coupling $\lambda_2$ we find $c$ such that 
\begin{align}
    \lambda_2^2(c\mathsf{e}) = \frac{\pi}{4}\mod 2\pi\,,
\end{align}
and we can always ensure that $c$ is sufficiently large so that
\begin{align}
    \lambda^2_2(c^2\mathsf{e}^2) = c\cdot\mathsf{e}\rr{\frac{\pi}{4}+2\pi n} \gg \mathsf{w}_2
\end{align}
in order to satisfy both conditions in \eqref{eq: dimless-proposition-1}. 

In terms of the rescaled coupling constants, the strong-coupling condition implies that for fixed $\lambda_2$ and $c\gg 1$, we get
\begin{equation}
    \rho_{\textsc{eb}} \xrightarrow{c\gg 1} 
    \begin{pmatrix}
        \tilde P_- & 0 & 0 & 0\\ 
        0 & \tilde P_+ & \tilde P_+ & 0\\ 
        0 & \tilde P_+ & \tilde P_+ & 0\\
        0 & 0 &0 & \tilde P_-
    \end{pmatrix}
     \,.
\end{equation}
Using negativity $\mathcal{N}$ as an entanglement measure \cite{vidal2002negativity}, 
\begin{align}
    \mathcal{N}[\rho_{\textsc{eb}}] = \frac{||\rho_\textsc{eb}^{\Gamma_\textsc{e}}||_1-1}{2} = \frac{1}{2}e^{-2\lambda_2^2\mathsf{w}_2}
\end{align}
where $\rho_{\textsc{eb}}^{\Gamma_\textsc{e}}$ is the partial transpose of $\rho_\textsc{eb}$ and $||\cdot||_1$ is the trace norm. This means that $\rho_{\textsc{eb}}$ is always an entangled state for any finite coupling $\lambda_2\ll \lambda_1$ and the channel is able to preserve some quantum entanglement. 

Finally, we obtain the perfect quantum channel in the idealized limit where $c\to\infty$ and $\lambda_2\to 0$ such that $\lambda_1$ is finite and still satisfies \eqref{eq: dimless-proposition-1}, i.e., 
\begin{equation}
    \rho_{\textsc{eb}} = \frac12 \begin{pmatrix}
    0 & 0 & 0 & 0\\ 
    0 & 1 & 1 & 0\\ 
    0 & 1 & 1 & 0\\
    0 & 0 &0 & 0 
    \end{pmatrix}\,,
\end{equation}
i.e., a Bell state with maximum negativity $\mathcal{N}[\rho_\textsc{eb}] = 1/2$. Therefore, in this limit the coherent information $I_c(\Phi)\to 1$, which corresponds to a perfect quantum channel that can transmit a maximally entangled state perfectly.

\subsection{Alice and Bob are spacelike separated}

The preceding analysis shows that the ability to implement the approximate `SWAP' gate depends on how closely one is able to satisfy the fine-tuning and strong-coupling conditions in Proposition~\ref{thm: encoding}. While this is straightforward for Alice, Proposition~\ref{thm: bob smearing} tells us that for Bob to emulate $\hat\phi(f_i)$ defined at some (global coordinate) time $t(\tau_\textsc{a})$,  Bob's spacetime smearing functions $g_i$ at time $t(\tau_\textsc{b})$ has to be somewhat delocalized over some region of spacetime to catch all the signals from Alice. Since the quantum channel $\Phi$ is built out of a relativistic quantum field, we ought to be able to show that there can be no transmission of both classical and quantum information through $\Phi$.

Formally, we start from the CCR which tells us that 
\begin{align}
    E(f_i,g_j) = 0\qquad \forall i,j=1,2
    \label{eq: spacelike-separation-Alice-Bob}
\end{align}
when Alice and Bob are spacelike-separated, i.e., $\supp(g_i)\cap \mathscr{J}^\pm(\supp(f_i)) = \emptyset$. Note that the Hadamard function $\mathsf{H}(f_i,g_j)$ is generically nonzero for spacelike-separated regions due to entanglement contained in the field state \cite{summers1985bell,summers1987bell,tjoa2021harvesting}. The calculations proceeds similarly as in the previous section, except that Bob will have his own arbitrary smearing functions and their respective UDW interaction unitaries are
{\allowdisplaybreaks[4]
\begin{subequations}
\begin{align}
    \hat U_\textsc{A} &= \sum_{x,z \in \{\pm\} } \hat P_x \hat P_z \otimes W(xEf_2) W(zEf_1)\\
    \hat U_\textsc{b} &= \sum_{x,z \in \{\pm\} } \hat P_z \hat P_x \otimes W(-zEg_1) W(-xEg_2).
\end{align}
\end{subequations}}
The condition for spacelike separation in Eq.~\eqref{eq: spacelike-separation-Alice-Bob} will be enforced later in the calculation once everything is written in terms of two-point functions. 

We can analyse this scenario very generally by allowing Alice and the environment to start in any arbitrary state $\rho_\textsc{ea}$ instead of a maximally entangled state or purifications of any initial state $\rho_{\textsc{a},0}$. The expression for the density matrix $\rho_{\textsc{eb}}$ now reads

\begin{align}
    \rho_\textsc{eb} &= \sum_{x_i, z_i} \omega(\hat{O})  \mathrm{Tr}_\textsc{a} \left[ \hat P_{x_4} \hat P_{z_4} \rho_\textsc{ae} \hat P_{-z_1} \hat P_{-x_1} \right]\notag\\
    & \otimes \hat P_{-z_3}\hat P_{-x_3} \ket{+_y}\!\bra{+_y}_\textsc{b} \hat P_{x_2}\hat P_{z_2}
\end{align}
where the projection operators act on Alice or Bob and $\hat O$ is now defined to be
\begin{align}
    \hat O &= W(z_1Ef_1) W(x_1Ef_2) W(x_2Eg_2) W(z_2Eg_1)  \notag\\
    & \times W(z_3Eg_1) W(x_3Eg_2) W(x_4Ef_2) W(z_4Ef_1).
\end{align}
Applying the third Weyl relation and using Eq.~\eqref{eq: quasifree-definition}, this evaluates to 
\begin{widetext}
    \begin{align}
    \omega(\hat{O}) 
    &= 
    e^{-\frac{i}{2} (z_1 - z_4)(x_1 + x_4) E(f_{1}, f_{2})}  
    e^{-\frac{i}{2} (z_2 + z_3)(x_3 - x_2) E(g_{1}, g_{2})}
    e^{-\frac{i}{2} (z_2 + z_3)(z_1 - z_4) E(f_{1}, g_{1})} \nonumber \\
    &\times 
    e^{-\frac{i}{2} (z_1 - z_4)(x_2 + x_3) E(f_{1}, g_{2})}
    e^{-\frac{i}{2} (z_2 + z_3)(x_1 - x_4) E(f_{2}, g_{1})}
    e^{-\frac{i}{2} (x_2 + x_3)(x_1 - x_4) E(f_{2}, g_{2})} \nonumber \\
    &\times  e^{-\frac12 \left(
      (z_1+z_4) (z_1+z_4)  \mathsf{W}(f_{1}, f_{1})
    + (x_1+x_4) (x_1+x_4)  \mathsf{W}(f_{2}, f_{2})
    + (z_2+z_3) (z_2+z_3)  \mathsf{W}(g_{1}, g_{1})
    + (x_2+x_3) (x_2+x_3)  \mathsf{W}(g_{2}, g_{2}) \right)} \nonumber \\
    &\times e^{-\frac12 \left(
      (x_1+x_4) (z_1+z_4)  \mathsf{H}(f_{1}, f_{2})
    + (z_2+z_3) (z_1+z_4)  \mathsf{H}(f_{1}, g_{1})
    + (x_2+x_3) (z_1+z_4)  \mathsf{H}(f_{1}, g_{2}) \right)} \nonumber \\
    &\times e^{-\frac12 \left(
      (z_2+z_3) (x_1+x_4)  \mathsf{H}(f_{2}, g_{1})
    + (x_2+x_3) (x_1+x_4)  \mathsf{H}(f_{2}, g_{2})
    + (x_2+x_3) (z_2+z_3)  \mathsf{H}(g_{1}, g_{2}) \right)}\,.
\end{align}
This is the most general form that $\omega(\hat O)$ can take in terms of Alice and Bob's two-point functions. We now enforce spacelike-separation using Eq.~\eqref{eq: spacelike-separation-Alice-Bob}, giving us
\begin{align}
    \omega(\hat{O}) 
    &= 
    e^{-\frac{i}{2} (z_1 - z_4)(x_1 + x_4) E(f_{1}, f_{2})}  
    e^{-\frac{i}{2} (z_2 + z_3)(x_3 - x_2) E(g_{1}, g_{2})}
    \nonumber \\
    &\times  e^{-\frac12 \left(
      (z_1+z_4) (z_1+z_4)  \mathsf{W}(f_{1}, f_{1})
    + (x_1+x_4) (x_1+x_4)  \mathsf{W}(f_{2}, f_{2})
    + (z_2+z_3) (z_2+z_3)  \mathsf{W}(g_{1}, g_{1})
    + (x_2+x_3) (x_2+x_3)  \mathsf{W}(g_{2}, g_{2}) \right)} \nonumber \\
    &\times e^{-\frac12 \left(
      (x_1+x_4) (z_1+z_4)  \mathsf{H}(f_{1}, f_{2})
    + (z_2+z_3) (z_1+z_4)  \mathsf{H}(f_{1}, g_{1})
    + (x_2+x_3) (z_1+z_4)  \mathsf{H}(f_{1}, g_{2}) \right)} \nonumber \\
    &\times e^{-\frac12 \left(
      (z_2+z_3) (x_1+x_4)  \mathsf{H}(f_{2}, g_{1})
    + (x_2+x_3) (x_1+x_4)  \mathsf{H}(f_{2}, g_{2})
    + (x_2+x_3) (z_2+z_3)  \mathsf{H}(g_{1}, g_{2}) \right)}\,.
\end{align}
\end{widetext} 
The density matrix $\rho_\textsc{eb}$ can now be calculated 
\begin{equation}
    \rho_\textsc{eb} = \rho_{\textsc{e}} \otimes \begin{pmatrix}
        \frac12 & X\\ X^* & \frac12
    \end{pmatrix}
    \label{eq: rhoEB-spacelike-separated}
\end{equation}
with $\rho_\textsc{e} = \mathrm{Tr}_\textsc{a} [\rho_{\textsc{ea}}]$ and 
\begin{align*}
    X &= -\frac{i}{2} e^{-2 \mathsf{W}(g_1, g_1)} \nonumber\\
    & \left[ e^{-2 \mathsf{W}(g_2, g_2)}\cosh(2\mathsf{H}(g_1, g_2)) + \sin(2E(g_1, g_2)) \right].
\end{align*}

Observe that the density matrix $\rho_\textsc{eb}$ depends only on the initial condition of the environment and Bob's local two-point functions, i.e. not involving Alice's smearing functions at all, and is clearly separable across $B$ and $E$. 
Consequently, the quantum channel preserves no entanglement at all from the joint system $EA$ (viewing it as $\openone\otimes\Phi$). In fact, since we start from any arbitrary state $\rho_\textsc{ea}$, we have effectively shown that $\Phi$ is a PPT channel \cite{graeme2012detect,singh2022detecting}, and hence also an \textit{entanglement-breaking channel} (because PPT is equivalent to separability for two qubits), and these are known to fall under the class of anti-degradable channels with zero quantum capacity $\mathcal{Q}[\Phi] = 0$  \cite{ruskai2003entanglementbreaking}. 

From an information-theoretic viewpoint, zero quantum capacity says nothing about the ability of the quantum channel to transmit classical information, since the classical channel capacity $\mathcal{C}[\Phi]$ is an upper bound for quantum capacity, i.e., $\mathcal{C}[\Phi] \geq \mathcal{Q}[\Phi]$. Quantum channels can also transmit classical bits by, for instance, considering only diagonal density matrix
\begin{equation}
    \rho_{\textsc{a},0} = \begin{pmatrix}
        \alpha & 0\\ 0 & 1-\alpha
    \end{pmatrix}\,,\quad \alpha\in [0,1]\,.
\end{equation}
and indeed it has been shown that for certain delta-coupled UDW models,  the classical capacity $\mathcal{C}[\Phi]$ can be computed exactly and is zero if Alice and Bob are spacelike \cite{tjoa2022capacity,Landulfo2016communication}. Our quantum channel is more complicated as Alice and Bob each perform a sequence of two simple-generated unitaries, but it is not necessary for us to find the explicit formula for $\mathcal{C}[\Phi]$ to show that it is zero. This follows directly from the fact that the density matrix in Eq.~\eqref{eq: rhoEB-spacelike-separated} does not depend on $\rho_{\textsc{a},0}$ at all, therefore the quantum channel carries no information whatsoever about Alice's initial state, hence $\mathcal{C}[\Phi] = 0$. 

We stress that one of the main features of a relativistic quantum channel is precisely that one should be able to show that it respects the causal structure of spacetime by construction. This requires one to be able to show, for instance, that no information transmission is allowed at spacelike separation, or more generally between causally disconnected observers. In fact we can say more: our formalism does not require us to specify anything about the (violation of) strong Huygens' principle \cite{McLenaghan1974huygens,Sonego1992huygenscurved,Faraoni2019huygens,tjoa2021harvesting,Causality2015Eduardo,Casals2020commBH}, since everything is encoded in the causal propagator. For example, in $(3+1)$-dimensional Minkowski spacetime, Bob can be timelike-separated from Alice but for a massless scalar field, the causal propagator has \textit{zero support} in the interior of the light cone (see, e.g., \cite{tjoa2021harvesting,Causality2015Eduardo}). In this case, if Bob's spacetime smearing functions to the field  $g_i$ are strictly in the timelike interior and  $\supp(g_i)\cap \mathscr{J}^+(\supp f_j) = \emptyset$ for any $i,j=1,2$, then $\mathcal{Q}[\Phi]=\mathcal{C}[\Phi]=0$.

\subsection{Intermediate scenarios}
\label{subsec: intermediate scenarios}

In the two preceding subsections, we have analyzed two opposite regimes: the ideal case with $Ef_i = Eg_i$ and the worst case with Alice and Bob spacelike-separated. Let us make some comments on intermediate regimes. 

When the coherent information of a quantum channel is not close to 1, e.g., when $0 \leq I_c(\Phi) \ll 1$, it is in general very difficult to make quantitative claims about the quantum channel capacity for the reasons like superadditivity described in Section~\ref{sec: setup}. In fact, even for quantum channels with $I_c(\Phi) = 0$, only small subsets of these are known to have zero quantum capacity, namely the anti-degradable channels \cite{khatri2020principles,wolf2007degradable,ruskai2003entanglementbreaking} and the PPT channels \cite{graeme2012detect} (see also \cite{singh2022detecting}). For intermediate scenarios, such as when Bob is only partially covering  $\mathscr{J}^\pm(\supp(f_i))$ or implementing faulty SWAP gates we are not able to give generic quantitative statements about $\mathcal{Q}[\Phi]$, as is the case for generic quantum communication channels. 

However, in some situations we may try to use no-cloning type argument to conclude that $\mathcal{Q}[\Phi]=0$ in some scenarios as was done in \cite{Simidzija2020capacity,Jonsson2018qubit} (see also \cite{graeme2012detect,bruss1998cloneapprox,bennett1997erasure}). Suppose Bob is allowed to couple only to at most half of Alice's signal at some future time $t$ that defines the Cauchy slice for Bob's interaction. Consider another observer Charlie who now chooses to couple to the other half of the region at the same time slice so that jointly the interaction support is the full region needed for the perfect quantum channel. Since Charlie's interaction is by construction spacelike separated from Bob, Charlie's interaction is invisible to Bob (and vice versa). If Bob and Charlie were able to individually decode Alice's qubit, this would violate the no-cloning theorem, so we can conclude that the quantum capacity must be zero if the UDW interactions are set up this way\footnote{It is worth emphasizing that the protocol we have (and also in \cite{Simidzija2020capacity}) are generically not ideal, since for any finite interaction strengths the `SWAP' gates are never perfect and the `logical qubit' in the field is only approximate --- consequently one has to be careful about using exact cloning argument. In practice, it would be better to think of the scenario in terms of approximate cloning. Indeed, in Figure~8 of \cite{Simidzija2020capacity} the coherent information vanishes \textit{smoothly} as one varies Bob's interaction support when analyzing the cloning-type argument.}. 

The above discussion is not necessarily the only (or natural) way to view this. One way to think about this from the perspective of the quantum field is this: if Bob manages to decode Alice's state at time $t$, the `logical qubit' state $c_1\ket{+\alpha}-ic_2\ket{-\alpha}$ in Eq.~\eqref{eq: logical-qubit} would have been taken out by Bob, and the field would now change into a new state that has no logical qubit in it (essentially because Bob swapped out the state into his physical detector). Charlie, who couples to the field at any time $t+\epsilon$ for any arbitrary $\epsilon>0$, would no longer have any logical qubit to extract from the field, so Charlie cannot extract another qubit for free. Yet another way to see this is that Bob's interaction requires one to find $g_i$ in $\mathscr{J}^+(\supp(f_j))$ that can `undo' Alice's operation. Therefore, if we choose to pick $g_i$ that does not cover the full region needed for the decoding, then Bob \textit{fails} to implement his ideal $\hat U_{\textsc{B}}$ and instead performs a different unitary operation. In this interpretation, we would say that Bob is unable to perform the required `SWAP' gate to decode Alice's state. Showing that Bob has zero quantum capacity would mean showing that Bob's operations somehow completely depolarize the qubit state and lose all the quantum information in it --- we will not pursue this further in this paper and leave this for future work.

\section{Discussion and outlook}
\label{sec: conclude}

We motivated our work through the question: \textit{What is the information-carrying capacity of quantum fields?} The UDW framework allows us to formulate a precise statement of this question, in terms of quantum capacity of the relativistic quantum channel between Alice and Bob who try to transmit quantum information between their qubit detectors through the quantum field.  In this work, we have successfully generalized the results of \cite{Jonsson2018qubit,Simidzija2020capacity,Blasco2015Huygens,Blasco2016broadcast,simidzija2017cosmo}, in particular \cite{Simidzija2020capacity}, where we construct a quantum channel that can transmit quantum information arbitrarily well through a relativistic quantum scalar field in arbitrary curved spacetimes.

The features of our construction is that the resulting quantum channel, and hence the quantum capacity, are manifestly covariant, respect the causal structure of spacetime, and are independent of the details of the background geometry, topology, and the choice of Hilbert space (quasifree) representations of the CCR algebra. Furthermore, in this generalization we no longer need to deal with (violation of) the strong Huygens' principle that may occur in curved spacetimes for non-conformally invariant fields (see, e.g.,\cite{Blasco2015Huygens,Blasco2016broadcast,McLenaghan1974huygens,Faraoni2019huygens,Sonego1992huygenscurved}), since this is fully encoded in the causal propagator. We also showed explicitly that the quantum and classical capacity is zero for spacelike separated observers, as we expect from quantum communication channel built from a relativistic quantum field.

Crucially, in this generalized formulation it becomes clear that there are two essential features of a near-perfect quantum channel that can transmit quantum information arbitrarily well: (1) Alice and Bob's interactions need to satisfy the so-called fine-tuning and strong coupling conditions in Proposition~\ref{thm: encoding} in order to approximate an ideal `SWAP' gates; (2) Bob needs to couple to the field at appropriate spacetime regions, fixed by Proposition~\ref{thm: bob smearing}, so that he could execute the approximately ideal `SWAP' gates. Indeed, where Bob should be is fixed by the causal propagator that determines where Alice's information propagates in spacetime. The asymmetry between Alice and Bob can be viewed as the consequence of the asymmetry between encoding the qubit state into the field, with the logical basis being the coherent states $\ket{\pm \alpha}$, and decoding the state while taking into account relativistic causality. In this sense, our construction constitutes the most general relativistic qubit channel within the UDW framework that can transmit quantum information arbitrarily well in appropriate regimes and we also provided analyses when the quantum channel is sub-optimal.

There are several future directions that merit investigation, at least from a more fundamental standpoint:
\begin{itemize}[leftmargin=*]
    \item Is it possible to construct a perfect relativistic quantum channel using higher-dimensional detector models (using qudit detector models \cite{lima2023unruh,verdon2016quditQET} rather than qubits)? \textit{A priori} it is not obvious how to construct logical qudits in the field using only controlled displacement operators and what kind of interaction Hamiltonians are needed to do this, if it exists at all. It is also natural to consider several emitters and receivers, exploiting the fact that one can interpret displacement operations can be viewed as Weyl generators with respect to some spacetime smearing functions \cite{tjoa2023nonperturbative}.

    \item As mentioned in \cite{Simidzija2020capacity}, one may be interested in more directional coupling, which can be envisaged in (1+1)D models (see, e.g., \cite{Jonsson2018qubit,Jonsson2014cavityQED}). Our protocol works automatically simply because only the two-point functions change to reflect that the left- and right- movers of the field modes are independent. A more interesting setup would be to study how realistic settings involving waveguides and non-trivial dispersion relations can be adapted for this purpose. 

    \item Is the information-carrying capacity different between different fields with different spins and statistics? For example, how does one compare the scalar field with a fermionic field? The issue one has to deal with is that we have to separate the effects coming from the model dependence of the coupling (e.g., we cannot couple $\hat{\sigma}_x\otimes \hat{\psi}$ where $\psi$ is the fermionic field) from truly field-theoretic physics that contribute to the quantum capacity, if we wish to do it in the same way as we have done in this paper (see, e.g., some proposed fermionic couplings in RQI literature \cite{louko2016fermionflat,bruno2020neutrino}).

    \item Our construction relies on the fact that the quantum field is non-interacting (Gaussian). There are very few studies involving the UDW model with interacting QFTs such as the $\phi^4$ theory, even perturbatively and for a single qubit (see, e.g., \cite{kaplanek2020hot} for one example). The effects of interactions are unclear for (long-distance) quantum communication, both perturbatively and non-perturbatively. A possibly relevant testbed would be to consider known rigorous constructions in (1+1)D interacting QFTs (see, e.g., \cite{glimm1968lambda,glimm1970lambda,glimm1970lambda2,schrader1972yukawa,summers2012perspective}). 

    \item Is it possible to systematically study quantum capacity in the non-relativistic quantum many-body settings? Somewhat related studies to our setup are, for example, \cite{giovanetti2005capacity-spin,bose2003spin-chain,bayat2008spin-chain}. It would be interesting to see the similarities and differences when they are compared with setups involving relativistic quantum fields.
\end{itemize}
We leave these problems open for the future.

\section*{Acknowledgment}

M. K. acknowledges support from the Princeton International Internship Program (IIP). E. T. acknowledges funding from the Munich Center for Quantum Science and Technology (MCQST), funded by the Deutsche Forschungsgemeinschaft (DFG) under Germany’s Excellence Strategy (EXC2111 - 390814868). 

\bibliography{RQC}

\begin{thebibliography}{121}%
\makeatletter
\providecommand \@ifxundefined [1]{%
 \@ifx{#1\undefined}
}%
\providecommand \@ifnum [1]{%
 \ifnum #1\expandafter \@firstoftwo
 \else \expandafter \@secondoftwo
 \fi
}%
\providecommand \@ifx [1]{%
 \ifx #1\expandafter \@firstoftwo
 \else \expandafter \@secondoftwo
 \fi
}%
\providecommand \natexlab [1]{#1}%
\providecommand \enquote  [1]{``#1''}%
\providecommand \bibnamefont  [1]{#1}%
\providecommand \bibfnamefont [1]{#1}%
\providecommand \citenamefont [1]{#1}%
\providecommand \href@noop [0]{\@secondoftwo}%
\providecommand \href [0]{\begingroup \@sanitize@url \@href}%
\providecommand \@href[1]{\@@startlink{#1}\@@href}%
\providecommand \@@href[1]{\endgroup#1\@@endlink}%
\providecommand \@sanitize@url [0]{\catcode `\\12\catcode `\$12\catcode `\&12\catcode `\#12\catcode `\^12\catcode `\_12\catcode `\%12\relax}%
\providecommand \@@startlink[1]{}%
\providecommand \@@endlink[0]{}%
\providecommand \url  [0]{\begingroup\@sanitize@url \@url }%
\providecommand \@url [1]{\endgroup\@href {#1}{\urlprefix }}%
\providecommand \urlprefix  [0]{URL }%
\providecommand \Eprint [0]{\href }%
\providecommand \doibase [0]{http://dx.doi.org/}%
\providecommand \selectlanguage [0]{\@gobble}%
\providecommand \bibinfo  [0]{\@secondoftwo}%
\providecommand \bibfield  [0]{\@secondoftwo}%
\providecommand \translation [1]{[#1]}%
\providecommand \BibitemOpen [0]{}%
\providecommand \bibitemStop [0]{}%
\providecommand \bibitemNoStop [0]{.\EOS\space}%
\providecommand \EOS [0]{\spacefactor3000\relax}%
\providecommand \BibitemShut  [1]{\csname bibitem#1\endcsname}%
\let\auto@bib@innerbib\@empty
\bibitem [{\citenamefont {Verch}\ and\ \citenamefont {Werner}(2005)}]{verch2005distillability}%
  \BibitemOpen
  \bibfield  {author} {\bibinfo {author} {\bibfnamefont {R.}~\bibnamefont {Verch}}\ and\ \bibinfo {author} {\bibfnamefont {R.~F.}\ \bibnamefont {Werner}},\ }\href@noop {} {\bibfield  {journal} {\bibinfo  {journal} {Reviews in Mathematical Physics}\ }\textbf {\bibinfo {volume} {17}},\ \bibinfo {pages} {545} (\bibinfo {year} {2005})}\BibitemShut {NoStop}%
\bibitem [{\citenamefont {Hollands}\ and\ \citenamefont {Sanders}(2018)}]{hollands2017entanglement}%
  \BibitemOpen
  \bibfield  {author} {\bibinfo {author} {\bibfnamefont {S.}~\bibnamefont {Hollands}}\ and\ \bibinfo {author} {\bibfnamefont {K.}~\bibnamefont {Sanders}},\ }\href {https://link.springer.com/book/10.1007/978-3-319-94902-4} {\emph {\bibinfo {title} {Entanglement measures and their properties in quantum field theory}}},\ Vol.~\bibinfo {volume} {34}\ (\bibinfo  {publisher} {Springer},\ \bibinfo {year} {2018})\BibitemShut {NoStop}%
\bibitem [{\citenamefont {Hollands}\ and\ \citenamefont {Ranallo}(2023)}]{hollands2023channel}%
  \BibitemOpen
  \bibfield  {author} {\bibinfo {author} {\bibfnamefont {S.}~\bibnamefont {Hollands}}\ and\ \bibinfo {author} {\bibfnamefont {A.}~\bibnamefont {Ranallo}},\ }\href {\doibase https://doi.org/10.1007/s00220-023-04855-x} {\bibfield  {journal} {\bibinfo  {journal} {Communications in Mathematical Physics}\ }\textbf {\bibinfo {volume} {404}},\ \bibinfo {pages} {927} (\bibinfo {year} {2023})}\BibitemShut {NoStop}%
\bibitem [{\citenamefont {Sanders}(2023)}]{sanders2023separable}%
  \BibitemOpen
  \bibfield  {author} {\bibinfo {author} {\bibfnamefont {K.}~\bibnamefont {Sanders}},\ }\href {\doibase 10.1088/1751-8121/ad0bca} {\bibfield  {journal} {\bibinfo  {journal} {Journal of Physics A: Mathematical and Theoretical}\ }\textbf {\bibinfo {volume} {56}},\ \bibinfo {pages} {505201} (\bibinfo {year} {2023})}\BibitemShut {NoStop}%
\bibitem [{\citenamefont {Summers}\ and\ \citenamefont {Werner}(1985)}]{summers1985bell}%
  \BibitemOpen
  \bibfield  {author} {\bibinfo {author} {\bibfnamefont {S.~J.}\ \bibnamefont {Summers}}\ and\ \bibinfo {author} {\bibfnamefont {R.}~\bibnamefont {Werner}},\ }\href {\doibase https://doi.org/10.1016/0375-9601(85)90093-3} {\bibfield  {journal} {\bibinfo  {journal} {Phys. Lett. A}\ }\textbf {\bibinfo {volume} {110}},\ \bibinfo {pages} {257 } (\bibinfo {year} {1985})}\BibitemShut {NoStop}%
\bibitem [{\citenamefont {Summers}\ and\ \citenamefont {Werner}(1987)}]{summers1987bell}%
  \BibitemOpen
  \bibfield  {author} {\bibinfo {author} {\bibfnamefont {S.~J.}\ \bibnamefont {Summers}}\ and\ \bibinfo {author} {\bibfnamefont {R.}~\bibnamefont {Werner}},\ }\href {\doibase 10.1063/1.527733} {\bibfield  {journal} {\bibinfo  {journal} {J. Math. Phys.}\ }\textbf {\bibinfo {volume} {28}},\ \bibinfo {pages} {2440} (\bibinfo {year} {1987})}\BibitemShut {NoStop}%
\bibitem [{\citenamefont {Casini}\ \emph {et~al.}(2020)\citenamefont {Casini}, \citenamefont {Huerta}, \citenamefont {Mag{\'a}n},\ and\ \citenamefont {Pontello}}]{casini2020entanglement}%
  \BibitemOpen
  \bibfield  {author} {\bibinfo {author} {\bibfnamefont {H.}~\bibnamefont {Casini}}, \bibinfo {author} {\bibfnamefont {M.}~\bibnamefont {Huerta}}, \bibinfo {author} {\bibfnamefont {J.~M.}\ \bibnamefont {Mag{\'a}n}}, \ and\ \bibinfo {author} {\bibfnamefont {D.}~\bibnamefont {Pontello}},\ }\href {\doibase https://doi.org/10.1007/JHEP02(2020)014} {\bibfield  {journal} {\bibinfo  {journal} {Journal of High Energy Physics}\ }\textbf {\bibinfo {volume} {2020}},\ \bibinfo {pages} {1} (\bibinfo {year} {2020})}\BibitemShut {NoStop}%
\bibitem [{\citenamefont {Casini}\ \emph {et~al.}(2021)\citenamefont {Casini}, \citenamefont {Huerta}, \citenamefont {Mag{\'a}n},\ and\ \citenamefont {Pontello}}]{casini2021entropic}%
  \BibitemOpen
  \bibfield  {author} {\bibinfo {author} {\bibfnamefont {H.}~\bibnamefont {Casini}}, \bibinfo {author} {\bibfnamefont {M.}~\bibnamefont {Huerta}}, \bibinfo {author} {\bibfnamefont {J.~M.}\ \bibnamefont {Mag{\'a}n}}, \ and\ \bibinfo {author} {\bibfnamefont {D.}~\bibnamefont {Pontello}},\ }\href {\doibase https://doi.org/10.1007/JHEP04(2021)277} {\bibfield  {journal} {\bibinfo  {journal} {Journal of High Energy Physics}\ }\textbf {\bibinfo {volume} {2021}},\ \bibinfo {pages} {1} (\bibinfo {year} {2021})}\BibitemShut {NoStop}%
\bibitem [{\citenamefont {Longo}\ and\ \citenamefont {Xu}(2018)}]{longo2018relative}%
  \BibitemOpen
  \bibfield  {author} {\bibinfo {author} {\bibfnamefont {R.}~\bibnamefont {Longo}}\ and\ \bibinfo {author} {\bibfnamefont {F.}~\bibnamefont {Xu}},\ }\href {\doibase https://doi.org/10.1016/j.aim.2018.08.015} {\bibfield  {journal} {\bibinfo  {journal} {Advances in Mathematics}\ }\textbf {\bibinfo {volume} {337}},\ \bibinfo {pages} {139} (\bibinfo {year} {2018})}\BibitemShut {NoStop}%
\bibitem [{\citenamefont {Fewster}\ and\ \citenamefont {Verch}(2020)}]{fewster2020measurement}%
  \BibitemOpen
  \bibfield  {author} {\bibinfo {author} {\bibfnamefont {C.~J.}\ \bibnamefont {Fewster}}\ and\ \bibinfo {author} {\bibfnamefont {R.}~\bibnamefont {Verch}},\ }\href {\doibase https://doi.org/10.1007/s00220-020-03800-6} {\bibfield  {journal} {\bibinfo  {journal} {Communications in Mathematical physics}\ }\textbf {\bibinfo {volume} {378}},\ \bibinfo {pages} {851} (\bibinfo {year} {2020})}\BibitemShut {NoStop}%
\bibitem [{\citenamefont {Bostelmann}\ \emph {et~al.}(2021)\citenamefont {Bostelmann}, \citenamefont {Fewster},\ and\ \citenamefont {Ruep}}]{bostelmann2021impossible}%
  \BibitemOpen
  \bibfield  {author} {\bibinfo {author} {\bibfnamefont {H.}~\bibnamefont {Bostelmann}}, \bibinfo {author} {\bibfnamefont {C.~J.}\ \bibnamefont {Fewster}}, \ and\ \bibinfo {author} {\bibfnamefont {M.~H.}\ \bibnamefont {Ruep}},\ }\href {\doibase 10.1103/PhysRevD.103.025017} {\bibfield  {journal} {\bibinfo  {journal} {Phys. Rev. D}\ }\textbf {\bibinfo {volume} {103}},\ \bibinfo {pages} {025017} (\bibinfo {year} {2021})}\BibitemShut {NoStop}%
\bibitem [{\citenamefont {Polo-G\'omez}\ \emph {et~al.}(2022)\citenamefont {Polo-G\'omez}, \citenamefont {Garay},\ and\ \citenamefont {Mart\'{\i}n-Mart\'{\i}nez}}]{polo2021detectorbased}%
  \BibitemOpen
  \bibfield  {author} {\bibinfo {author} {\bibfnamefont {J.}~\bibnamefont {Polo-G\'omez}}, \bibinfo {author} {\bibfnamefont {L.~J.}\ \bibnamefont {Garay}}, \ and\ \bibinfo {author} {\bibfnamefont {E.}~\bibnamefont {Mart\'{\i}n-Mart\'{\i}nez}},\ }\href {\doibase 10.1103/PhysRevD.105.065003} {\bibfield  {journal} {\bibinfo  {journal} {Phys. Rev. D}\ }\textbf {\bibinfo {volume} {105}},\ \bibinfo {pages} {065003} (\bibinfo {year} {2022})}\BibitemShut {NoStop}%
\bibitem [{\citenamefont {Jubb}(2022)}]{jubb2022causal}%
  \BibitemOpen
  \bibfield  {author} {\bibinfo {author} {\bibfnamefont {I.}~\bibnamefont {Jubb}},\ }\href {\doibase 10.1103/PhysRevD.105.025003} {\bibfield  {journal} {\bibinfo  {journal} {Phys. Rev. D}\ }\textbf {\bibinfo {volume} {105}},\ \bibinfo {pages} {025003} (\bibinfo {year} {2022})}\BibitemShut {NoStop}%
\bibitem [{\citenamefont {Pranzini}\ and\ \citenamefont {Keski-Vakkuri}(2023)}]{pranzini2023detector}%
  \BibitemOpen
  \bibfield  {author} {\bibinfo {author} {\bibfnamefont {N.}~\bibnamefont {Pranzini}}\ and\ \bibinfo {author} {\bibfnamefont {E.}~\bibnamefont {Keski-Vakkuri}},\ }\href {https://arxiv.org/abs/2310.06596} {\bibfield  {journal} {\bibinfo  {journal} {arXiv preprint arXiv:2310.06596}\ } (\bibinfo {year} {2023})}\BibitemShut {NoStop}%
\bibitem [{\citenamefont {van Luijk}\ \emph {et~al.}(2024{\natexlab{a}})\citenamefont {van Luijk}, \citenamefont {Schwonnek}, \citenamefont {Stottmeister},\ and\ \citenamefont {Werner}}]{vanLuijk2023schmidt}%
  \BibitemOpen
  \bibfield  {author} {\bibinfo {author} {\bibfnamefont {L.}~\bibnamefont {van Luijk}}, \bibinfo {author} {\bibfnamefont {R.}~\bibnamefont {Schwonnek}}, \bibinfo {author} {\bibfnamefont {A.}~\bibnamefont {Stottmeister}}, \ and\ \bibinfo {author} {\bibfnamefont {R.~F.}\ \bibnamefont {Werner}},\ }\href {\doibase https://doi.org/10.1007/s00220-024-05011-9} {\bibfield  {journal} {\bibinfo  {journal} {Communications in Mathematical Physics}\ }\textbf {\bibinfo {volume} {405}},\ \bibinfo {pages} {152} (\bibinfo {year} {2024}{\natexlab{a}})}\BibitemShut {NoStop}%
\bibitem [{\citenamefont {van Luijk}\ \emph {et~al.}(2024{\natexlab{b}})\citenamefont {van Luijk}, \citenamefont {Stottmeister}, \citenamefont {Werner},\ and\ \citenamefont {Wilming}}]{vanLuijk2024embezzlement}%
  \BibitemOpen
  \bibfield  {author} {\bibinfo {author} {\bibfnamefont {L.}~\bibnamefont {van Luijk}}, \bibinfo {author} {\bibfnamefont {A.}~\bibnamefont {Stottmeister}}, \bibinfo {author} {\bibfnamefont {R.~F.}\ \bibnamefont {Werner}}, \ and\ \bibinfo {author} {\bibfnamefont {H.}~\bibnamefont {Wilming}},\ }\href {https://arxiv.org/abs/2401.07299} {\bibfield  {journal} {\bibinfo  {journal} {arXiv preprint arXiv:2401.07299}\ } (\bibinfo {year} {2024}{\natexlab{b}})}\BibitemShut {NoStop}%
\bibitem [{\citenamefont {Hawking}(1975)}]{hawking1975particle}%
  \BibitemOpen
  \bibfield  {author} {\bibinfo {author} {\bibfnamefont {S.~W.}\ \bibnamefont {Hawking}},\ }\href {\doibase https://link.springer.com/article/10.1007/BF01608497} {\bibfield  {journal} {\bibinfo  {journal} {Comm. Math. Phys.}\ }\textbf {\bibinfo {volume} {43}},\ \bibinfo {pages} {199} (\bibinfo {year} {1975})}\BibitemShut {NoStop}%
\bibitem [{\citenamefont {Crispino}\ \emph {et~al.}(2008)\citenamefont {Crispino}, \citenamefont {Higuchi},\ and\ \citenamefont {Matsas}}]{Crispino2008review}%
  \BibitemOpen
  \bibfield  {author} {\bibinfo {author} {\bibfnamefont {L.~C.~B.}\ \bibnamefont {Crispino}}, \bibinfo {author} {\bibfnamefont {A.}~\bibnamefont {Higuchi}}, \ and\ \bibinfo {author} {\bibfnamefont {G.~E.~A.}\ \bibnamefont {Matsas}},\ }\href {\doibase 10.1103/RevModPhys.80.787} {\bibfield  {journal} {\bibinfo  {journal} {Rev. Mod. Phys.}\ }\textbf {\bibinfo {volume} {80}},\ \bibinfo {pages} {787} (\bibinfo {year} {2008})}\BibitemShut {NoStop}%
\bibitem [{\citenamefont {Unruh}(1976)}]{Unruh1979evaporation}%
  \BibitemOpen
  \bibfield  {author} {\bibinfo {author} {\bibfnamefont {W.~G.}\ \bibnamefont {Unruh}},\ }\href {\doibase 10.1103/PhysRevD.14.870} {\bibfield  {journal} {\bibinfo  {journal} {Phys. Rev. D}\ }\textbf {\bibinfo {volume} {14}},\ \bibinfo {pages} {870} (\bibinfo {year} {1976})}\BibitemShut {NoStop}%
\bibitem [{\citenamefont {{DeWitt}}(1979)}]{DeWitt1979}%
  \BibitemOpen
  \bibfield  {author} {\bibinfo {author} {\bibfnamefont {B.~S.}\ \bibnamefont {{DeWitt}}},\ }in\ \href@noop {} {\emph {\bibinfo {booktitle} {General Relativity: An Einstein centenary survey}}},\ \bibinfo {editor} {edited by\ \bibinfo {editor} {\bibfnamefont {S.~W.}\ \bibnamefont {{Hawking}}}\ and\ \bibinfo {editor} {\bibfnamefont {W.}~\bibnamefont {{Israel}}}}\ (\bibinfo  {publisher} {Cambridge University Press, Cambridge},\ \bibinfo {year} {1979})\ pp.\ \bibinfo {pages} {680--745}\BibitemShut {NoStop}%
\bibitem [{\citenamefont {Ju{\'{a}}rez-Aubry}\ and\ \citenamefont {Louko}(2014)}]{Aubry2014derivative}%
  \BibitemOpen
  \bibfield  {author} {\bibinfo {author} {\bibfnamefont {B.~A.}\ \bibnamefont {Ju{\'{a}}rez-Aubry}}\ and\ \bibinfo {author} {\bibfnamefont {J.}~\bibnamefont {Louko}},\ }\href {\doibase 10.1088/0264-9381/31/24/245007} {\bibfield  {journal} {\bibinfo  {journal} {Class. and Quantum Gravity}\ }\textbf {\bibinfo {volume} {31}},\ \bibinfo {pages} {245007} (\bibinfo {year} {2014})}\BibitemShut {NoStop}%
\bibitem [{\citenamefont {Ju{\'a}rez-Aubry}\ and\ \citenamefont {Louko}(2018)}]{aubry2018quantumBH}%
  \BibitemOpen
  \bibfield  {author} {\bibinfo {author} {\bibfnamefont {B.~A.}\ \bibnamefont {Ju{\'a}rez-Aubry}}\ and\ \bibinfo {author} {\bibfnamefont {J.}~\bibnamefont {Louko}},\ }\href {\doibase https://doi.org/10.1007/JHEP05(2018)140} {\bibfield  {journal} {\bibinfo  {journal} {Journal of High Energy Physics}\ }\textbf {\bibinfo {volume} {2018}},\ \bibinfo {pages} {1} (\bibinfo {year} {2018})}\BibitemShut {NoStop}%
\bibitem [{\citenamefont {Tjoa}\ and\ \citenamefont {Mann}(2022)}]{tjoa2022unruh}%
  \BibitemOpen
  \bibfield  {author} {\bibinfo {author} {\bibfnamefont {E.}~\bibnamefont {Tjoa}}\ and\ \bibinfo {author} {\bibfnamefont {R.~B.}\ \bibnamefont {Mann}},\ }\href {\doibase https://doi.org/10.1007/JHEP03(2022)014} {\bibfield  {journal} {\bibinfo  {journal} {Journal of High Energy Physics}\ }\textbf {\bibinfo {volume} {2022}},\ \bibinfo {pages} {1} (\bibinfo {year} {2022})}\BibitemShut {NoStop}%
\bibitem [{\citenamefont {Kaplanek}\ and\ \citenamefont {Burgess}(2020)}]{kaplanek2020hot}%
  \BibitemOpen
  \bibfield  {author} {\bibinfo {author} {\bibfnamefont {G.}~\bibnamefont {Kaplanek}}\ and\ \bibinfo {author} {\bibfnamefont {C.}~\bibnamefont {Burgess}},\ }\href {\doibase https://doi.org/10.1007/JHEP03(2020)008} {\bibfield  {journal} {\bibinfo  {journal} {Journal of High Energy Physics}\ }\textbf {\bibinfo {volume} {2020}},\ \bibinfo {pages} {1} (\bibinfo {year} {2020})}\BibitemShut {NoStop}%
\bibitem [{\citenamefont {Kaplanek}\ and\ \citenamefont {Tjoa}(2023)}]{tjoa2023effective}%
  \BibitemOpen
  \bibfield  {author} {\bibinfo {author} {\bibfnamefont {G.}~\bibnamefont {Kaplanek}}\ and\ \bibinfo {author} {\bibfnamefont {E.}~\bibnamefont {Tjoa}},\ }\href {\doibase 10.1103/PhysRevA.107.012208} {\bibfield  {journal} {\bibinfo  {journal} {Phys. Rev. A}\ }\textbf {\bibinfo {volume} {107}},\ \bibinfo {pages} {012208} (\bibinfo {year} {2023})}\BibitemShut {NoStop}%
\bibitem [{\citenamefont {Moustos}\ and\ \citenamefont {Anastopoulos}(2017)}]{moustos2017nonmarkov}%
  \BibitemOpen
  \bibfield  {author} {\bibinfo {author} {\bibfnamefont {D.}~\bibnamefont {Moustos}}\ and\ \bibinfo {author} {\bibfnamefont {C.}~\bibnamefont {Anastopoulos}},\ }\href {\doibase 10.1103/PhysRevD.95.025020} {\bibfield  {journal} {\bibinfo  {journal} {Phys. Rev. D}\ }\textbf {\bibinfo {volume} {95}},\ \bibinfo {pages} {025020} (\bibinfo {year} {2017})}\BibitemShut {NoStop}%
\bibitem [{\citenamefont {Vilasini}\ and\ \citenamefont {Renner}(2022)}]{vilasini2022embedding}%
  \BibitemOpen
  \bibfield  {author} {\bibinfo {author} {\bibfnamefont {V.}~\bibnamefont {Vilasini}}\ and\ \bibinfo {author} {\bibfnamefont {R.}~\bibnamefont {Renner}},\ }\href {https://arxiv.org/abs/2203.11245} {\bibfield  {journal} {\bibinfo  {journal} {arXiv preprint arXiv:2203.11245}\ } (\bibinfo {year} {2022})}\BibitemShut {NoStop}%
\bibitem [{\citenamefont {Hardy}(2007)}]{Hardy2007towards}%
  \BibitemOpen
  \bibfield  {author} {\bibinfo {author} {\bibfnamefont {L.}~\bibnamefont {Hardy}},\ }\href {\doibase 10.1088/1751-8113/40/12/s12} {\bibfield  {journal} {\bibinfo  {journal} {Journal of Physics A: Mathematical and Theoretical}\ }\textbf {\bibinfo {volume} {40}},\ \bibinfo {pages} {3081} (\bibinfo {year} {2007})}\BibitemShut {NoStop}%
\bibitem [{\citenamefont {Zych}\ \emph {et~al.}(2019)\citenamefont {Zych}, \citenamefont {Costa}, \citenamefont {Pikovski},\ and\ \citenamefont {Brukner}}]{zych2019bell}%
  \BibitemOpen
  \bibfield  {author} {\bibinfo {author} {\bibfnamefont {M.}~\bibnamefont {Zych}}, \bibinfo {author} {\bibfnamefont {F.}~\bibnamefont {Costa}}, \bibinfo {author} {\bibfnamefont {I.}~\bibnamefont {Pikovski}}, \ and\ \bibinfo {author} {\bibfnamefont {{\v{C}}.}~\bibnamefont {Brukner}},\ }\href {\doibase https://doi.org/10.1038/s41467-019-11579-x} {\bibfield  {journal} {\bibinfo  {journal} {Nature communications}\ }\textbf {\bibinfo {volume} {10}},\ \bibinfo {pages} {1} (\bibinfo {year} {2019})}\BibitemShut {NoStop}%
\bibitem [{\citenamefont {Costa}\ and\ \citenamefont {Shrapnel}(2016)}]{Costa2016causal}%
  \BibitemOpen
  \bibfield  {author} {\bibinfo {author} {\bibfnamefont {F.}~\bibnamefont {Costa}}\ and\ \bibinfo {author} {\bibfnamefont {S.}~\bibnamefont {Shrapnel}},\ }\href {\doibase 10.1088/1367-2630/18/6/063032} {\bibfield  {journal} {\bibinfo  {journal} {New Journal of Physics}\ }\textbf {\bibinfo {volume} {18}},\ \bibinfo {pages} {063032} (\bibinfo {year} {2016})}\BibitemShut {NoStop}%
\bibitem [{\citenamefont {Paunkovi{\'{c}}}\ and\ \citenamefont {Vojinovi{\'{c}}}(2020)}]{paunkovic2020causal}%
  \BibitemOpen
  \bibfield  {author} {\bibinfo {author} {\bibfnamefont {N.}~\bibnamefont {Paunkovi{\'{c}}}}\ and\ \bibinfo {author} {\bibfnamefont {M.}~\bibnamefont {Vojinovi{\'{c}}}},\ }\href {\doibase 10.22331/q-2020-05-28-275} {\bibfield  {journal} {\bibinfo  {journal} {{Quantum}}\ }\textbf {\bibinfo {volume} {4}},\ \bibinfo {pages} {275} (\bibinfo {year} {2020})}\BibitemShut {NoStop}%
\bibitem [{\citenamefont {Reznik}(2003)}]{reznik2003entanglement}%
  \BibitemOpen
  \bibfield  {author} {\bibinfo {author} {\bibfnamefont {B.}~\bibnamefont {Reznik}},\ }\href {\doibase https://doi.org/10.1023/A:1022875910744} {\bibfield  {journal} {\bibinfo  {journal} {Found. Phys.}\ }\textbf {\bibinfo {volume} {33}},\ \bibinfo {pages} {167} (\bibinfo {year} {2003})}\BibitemShut {NoStop}%
\bibitem [{\citenamefont {Reznik}\ \emph {et~al.}(2005)\citenamefont {Reznik}, \citenamefont {Retzker},\ and\ \citenamefont {Silman}}]{reznik2005violating}%
  \BibitemOpen
  \bibfield  {author} {\bibinfo {author} {\bibfnamefont {B.}~\bibnamefont {Reznik}}, \bibinfo {author} {\bibfnamefont {A.}~\bibnamefont {Retzker}}, \ and\ \bibinfo {author} {\bibfnamefont {J.}~\bibnamefont {Silman}},\ }\href {\doibase 10.1103/PhysRevA.71.042104} {\bibfield  {journal} {\bibinfo  {journal} {Phys. Rev. A}\ }\textbf {\bibinfo {volume} {71}},\ \bibinfo {pages} {042104} (\bibinfo {year} {2005})}\BibitemShut {NoStop}%
\bibitem [{\citenamefont {Valentini}(1991)}]{Valentini1991nonlocalcorr}%
  \BibitemOpen
  \bibfield  {author} {\bibinfo {author} {\bibfnamefont {A.}~\bibnamefont {Valentini}},\ }\href {\doibase https://doi.org/10.1016/0375-9601(91)90952-5} {\bibfield  {journal} {\bibinfo  {journal} {Phys. Lett. A}\ }\textbf {\bibinfo {volume} {153}},\ \bibinfo {pages} {321 } (\bibinfo {year} {1991})}\BibitemShut {NoStop}%
\bibitem [{\citenamefont {Pozas-Kerstjens}\ and\ \citenamefont {Mart\'{\i}n-Mart\'{\i}nez}(2015)}]{pozas2015harvesting}%
  \BibitemOpen
  \bibfield  {author} {\bibinfo {author} {\bibfnamefont {A.}~\bibnamefont {Pozas-Kerstjens}}\ and\ \bibinfo {author} {\bibfnamefont {E.}~\bibnamefont {Mart\'{\i}n-Mart\'{\i}nez}},\ }\href {\doibase 10.1103/PhysRevD.92.064042} {\bibfield  {journal} {\bibinfo  {journal} {Phys. Rev. D}\ }\textbf {\bibinfo {volume} {92}},\ \bibinfo {pages} {064042} (\bibinfo {year} {2015})}\BibitemShut {NoStop}%
\bibitem [{\citenamefont {Pozas-Kerstjens}\ and\ \citenamefont {Mart\'{\i}n-Mart\'{\i}nez}(2016)}]{pozas2016entanglement}%
  \BibitemOpen
  \bibfield  {author} {\bibinfo {author} {\bibfnamefont {A.}~\bibnamefont {Pozas-Kerstjens}}\ and\ \bibinfo {author} {\bibfnamefont {E.}~\bibnamefont {Mart\'{\i}n-Mart\'{\i}nez}},\ }\href {\doibase 10.1103/PhysRevD.94.064074} {\bibfield  {journal} {\bibinfo  {journal} {Phys. Rev. D}\ }\textbf {\bibinfo {volume} {94}},\ \bibinfo {pages} {064074} (\bibinfo {year} {2016})}\BibitemShut {NoStop}%
\bibitem [{\citenamefont {Jonsson}\ \emph {et~al.}(2014)\citenamefont {Jonsson}, \citenamefont {Mart\'{\i}n-Mart\'{\i}nez},\ and\ \citenamefont {Kempf}}]{Jonsson2014cavityQED}%
  \BibitemOpen
  \bibfield  {author} {\bibinfo {author} {\bibfnamefont {R.~H.}\ \bibnamefont {Jonsson}}, \bibinfo {author} {\bibfnamefont {E.}~\bibnamefont {Mart\'{\i}n-Mart\'{\i}nez}}, \ and\ \bibinfo {author} {\bibfnamefont {A.}~\bibnamefont {Kempf}},\ }\href {\doibase 10.1103/PhysRevA.89.022330} {\bibfield  {journal} {\bibinfo  {journal} {Phys. Rev. A}\ }\textbf {\bibinfo {volume} {89}},\ \bibinfo {pages} {022330} (\bibinfo {year} {2014})}\BibitemShut {NoStop}%
\bibitem [{\citenamefont {Tjoa}\ and\ \citenamefont {Mart\'{\i}n-Mart\'{\i}nez}(2021)}]{tjoa2021harvesting}%
  \BibitemOpen
  \bibfield  {author} {\bibinfo {author} {\bibfnamefont {E.}~\bibnamefont {Tjoa}}\ and\ \bibinfo {author} {\bibfnamefont {E.}~\bibnamefont {Mart\'{\i}n-Mart\'{\i}nez}},\ }\href {\doibase 10.1103/PhysRevD.104.125005} {\bibfield  {journal} {\bibinfo  {journal} {Phys. Rev. D}\ }\textbf {\bibinfo {volume} {104}},\ \bibinfo {pages} {125005} (\bibinfo {year} {2021})}\BibitemShut {NoStop}%
\bibitem [{\citenamefont {Perche}\ \emph {et~al.}(2024)\citenamefont {Perche}, \citenamefont {Polo-G\'omez}, \citenamefont {Torres},\ and\ \citenamefont {Mart\'{\i}n-Mart\'{\i}nez}}]{perche2024relativisticEH}%
  \BibitemOpen
  \bibfield  {author} {\bibinfo {author} {\bibfnamefont {T.~R.}\ \bibnamefont {Perche}}, \bibinfo {author} {\bibfnamefont {J.}~\bibnamefont {Polo-G\'omez}}, \bibinfo {author} {\bibfnamefont {B.~d. S.~L.}\ \bibnamefont {Torres}}, \ and\ \bibinfo {author} {\bibfnamefont {E.}~\bibnamefont {Mart\'{\i}n-Mart\'{\i}nez}},\ }\href {\doibase 10.1103/PhysRevD.109.045018} {\bibfield  {journal} {\bibinfo  {journal} {Phys. Rev. D}\ }\textbf {\bibinfo {volume} {109}},\ \bibinfo {pages} {045018} (\bibinfo {year} {2024})}\BibitemShut {NoStop}%
\bibitem [{\citenamefont {Tjoa}(2022)}]{tjoa2022teleport}%
  \BibitemOpen
  \bibfield  {author} {\bibinfo {author} {\bibfnamefont {E.}~\bibnamefont {Tjoa}},\ }\href {\doibase 10.1103/PhysRevA.106.032432} {\bibfield  {journal} {\bibinfo  {journal} {Phys. Rev. A}\ }\textbf {\bibinfo {volume} {106}},\ \bibinfo {pages} {032432} (\bibinfo {year} {2022})}\BibitemShut {NoStop}%
\bibitem [{\citenamefont {Landulfo}(2016)}]{Landulfo2016communication}%
  \BibitemOpen
  \bibfield  {author} {\bibinfo {author} {\bibfnamefont {A.~G.~S.}\ \bibnamefont {Landulfo}},\ }\href {\doibase 10.1103/PhysRevD.93.104019} {\bibfield  {journal} {\bibinfo  {journal} {Phys. Rev. D}\ }\textbf {\bibinfo {volume} {93}},\ \bibinfo {pages} {104019} (\bibinfo {year} {2016})}\BibitemShut {NoStop}%
\bibitem [{\citenamefont {Tjoa}\ and\ \citenamefont {Gallock-Yoshimura}(2022)}]{tjoa2022capacity}%
  \BibitemOpen
  \bibfield  {author} {\bibinfo {author} {\bibfnamefont {E.}~\bibnamefont {Tjoa}}\ and\ \bibinfo {author} {\bibfnamefont {K.}~\bibnamefont {Gallock-Yoshimura}},\ }\href {\doibase 10.1103/PhysRevD.105.085011} {\bibfield  {journal} {\bibinfo  {journal} {Phys. Rev. D}\ }\textbf {\bibinfo {volume} {105}},\ \bibinfo {pages} {085011} (\bibinfo {year} {2022})}\BibitemShut {NoStop}%
\bibitem [{\citenamefont {Lapponi}\ \emph {et~al.}(2023)\citenamefont {Lapponi}, \citenamefont {Moustos}, \citenamefont {Bruschi},\ and\ \citenamefont {Mancini}}]{lapponi2023relativistic}%
  \BibitemOpen
  \bibfield  {author} {\bibinfo {author} {\bibfnamefont {A.}~\bibnamefont {Lapponi}}, \bibinfo {author} {\bibfnamefont {D.}~\bibnamefont {Moustos}}, \bibinfo {author} {\bibfnamefont {D.~E.}\ \bibnamefont {Bruschi}}, \ and\ \bibinfo {author} {\bibfnamefont {S.}~\bibnamefont {Mancini}},\ }\href {\doibase 10.1103/PhysRevD.107.125010} {\bibfield  {journal} {\bibinfo  {journal} {Phys. Rev. D}\ }\textbf {\bibinfo {volume} {107}},\ \bibinfo {pages} {125010} (\bibinfo {year} {2023})}\BibitemShut {NoStop}%
\bibitem [{\citenamefont {Kent}(1999)}]{kent1999bitcommit}%
  \BibitemOpen
  \bibfield  {author} {\bibinfo {author} {\bibfnamefont {A.}~\bibnamefont {Kent}},\ }\href {\doibase 10.1103/PhysRevLett.83.1447} {\bibfield  {journal} {\bibinfo  {journal} {Phys. Rev. Lett.}\ }\textbf {\bibinfo {volume} {83}},\ \bibinfo {pages} {1447} (\bibinfo {year} {1999})}\BibitemShut {NoStop}%
\bibitem [{\citenamefont {Lo}\ and\ \citenamefont {Chau}(1998)}]{lo1998quantum}%
  \BibitemOpen
  \bibfield  {author} {\bibinfo {author} {\bibfnamefont {H.-K.}\ \bibnamefont {Lo}}\ and\ \bibinfo {author} {\bibfnamefont {H.}~\bibnamefont {Chau}},\ }\href {\doibase https://doi.org/10.1016/S0167-2789(98)00053-0} {\bibfield  {journal} {\bibinfo  {journal} {Physica D: Nonlinear Phenomena}\ }\textbf {\bibinfo {volume} {120}},\ \bibinfo {pages} {177} (\bibinfo {year} {1998})},\ \bibinfo {note} {proceedings of the Fourth Workshop on Physics and Consumption}\BibitemShut {NoStop}%
\bibitem [{\citenamefont {Adlam}\ and\ \citenamefont {Kent}(2015)}]{adlam2015crypto}%
  \BibitemOpen
  \bibfield  {author} {\bibinfo {author} {\bibfnamefont {E.}~\bibnamefont {Adlam}}\ and\ \bibinfo {author} {\bibfnamefont {A.}~\bibnamefont {Kent}},\ }\href {\doibase 10.1103/PhysRevA.92.022315} {\bibfield  {journal} {\bibinfo  {journal} {Phys. Rev. A}\ }\textbf {\bibinfo {volume} {92}},\ \bibinfo {pages} {022315} (\bibinfo {year} {2015})}\BibitemShut {NoStop}%
\bibitem [{\citenamefont {Buhrman}\ \emph {et~al.}(2014)\citenamefont {Buhrman}, \citenamefont {Chandran}, \citenamefont {Fehr}, \citenamefont {Gelles}, \citenamefont {Goyal}, \citenamefont {Ostrovsky},\ and\ \citenamefont {Schaffner}}]{buhrman2014position}%
  \BibitemOpen
  \bibfield  {author} {\bibinfo {author} {\bibfnamefont {H.}~\bibnamefont {Buhrman}}, \bibinfo {author} {\bibfnamefont {N.}~\bibnamefont {Chandran}}, \bibinfo {author} {\bibfnamefont {S.}~\bibnamefont {Fehr}}, \bibinfo {author} {\bibfnamefont {R.}~\bibnamefont {Gelles}}, \bibinfo {author} {\bibfnamefont {V.}~\bibnamefont {Goyal}}, \bibinfo {author} {\bibfnamefont {R.}~\bibnamefont {Ostrovsky}}, \ and\ \bibinfo {author} {\bibfnamefont {C.}~\bibnamefont {Schaffner}},\ }\href {\doibase 10.1137/130913687} {\bibfield  {journal} {\bibinfo  {journal} {SIAM Journal on Computing}\ }\textbf {\bibinfo {volume} {43}},\ \bibinfo {pages} {150} (\bibinfo {year} {2014})}\BibitemShut {NoStop}%
\bibitem [{\citenamefont {Vilasini}\ \emph {et~al.}(2019)\citenamefont {Vilasini}, \citenamefont {Portmann},\ and\ \citenamefont {del Rio}}]{vilasini2019composable}%
  \BibitemOpen
  \bibfield  {author} {\bibinfo {author} {\bibfnamefont {V.}~\bibnamefont {Vilasini}}, \bibinfo {author} {\bibfnamefont {C.}~\bibnamefont {Portmann}}, \ and\ \bibinfo {author} {\bibfnamefont {L.}~\bibnamefont {del Rio}},\ }\href {\doibase 10.1088/1367-2630/ab0e3b} {\bibfield  {journal} {\bibinfo  {journal} {New Journal of Physics}\ }\textbf {\bibinfo {volume} {21}},\ \bibinfo {pages} {043057} (\bibinfo {year} {2019})}\BibitemShut {NoStop}%
\bibitem [{\citenamefont {Blasco}\ \emph {et~al.}(2015)\citenamefont {Blasco}, \citenamefont {Garay}, \citenamefont {Mart\'{\i}n-Benito},\ and\ \citenamefont {Mart\'{\i}n-Mart\'{\i}nez}}]{Blasco2015Huygens}%
  \BibitemOpen
  \bibfield  {author} {\bibinfo {author} {\bibfnamefont {A.}~\bibnamefont {Blasco}}, \bibinfo {author} {\bibfnamefont {L.~J.}\ \bibnamefont {Garay}}, \bibinfo {author} {\bibfnamefont {M.}~\bibnamefont {Mart\'{\i}n-Benito}}, \ and\ \bibinfo {author} {\bibfnamefont {E.}~\bibnamefont {Mart\'{\i}n-Mart\'{\i}nez}},\ }\href {\doibase 10.1103/PhysRevLett.114.141103} {\bibfield  {journal} {\bibinfo  {journal} {Phys. Rev. Lett.}\ }\textbf {\bibinfo {volume} {114}},\ \bibinfo {pages} {141103} (\bibinfo {year} {2015})}\BibitemShut {NoStop}%
\bibitem [{\citenamefont {Blasco}\ \emph {et~al.}(2016)\citenamefont {Blasco}, \citenamefont {Garay}, \citenamefont {Mart\'{\i}n-Benito},\ and\ \citenamefont {Mart\'{\i}n-Mart\'{\i}nez}}]{Blasco2016broadcast}%
  \BibitemOpen
  \bibfield  {author} {\bibinfo {author} {\bibfnamefont {A.}~\bibnamefont {Blasco}}, \bibinfo {author} {\bibfnamefont {L.~J.}\ \bibnamefont {Garay}}, \bibinfo {author} {\bibfnamefont {M.}~\bibnamefont {Mart\'{\i}n-Benito}}, \ and\ \bibinfo {author} {\bibfnamefont {E.}~\bibnamefont {Mart\'{\i}n-Mart\'{\i}nez}},\ }\href {\doibase 10.1103/PhysRevD.93.024055} {\bibfield  {journal} {\bibinfo  {journal} {Phys. Rev. D}\ }\textbf {\bibinfo {volume} {93}},\ \bibinfo {pages} {024055} (\bibinfo {year} {2016})}\BibitemShut {NoStop}%
\bibitem [{\citenamefont {Simidzija}\ and\ \citenamefont {Mart\'{\i}n-Mart\'{\i}nez}(2017{\natexlab{a}})}]{simidzija2017cosmo}%
  \BibitemOpen
  \bibfield  {author} {\bibinfo {author} {\bibfnamefont {P.}~\bibnamefont {Simidzija}}\ and\ \bibinfo {author} {\bibfnamefont {E.}~\bibnamefont {Mart\'{\i}n-Mart\'{\i}nez}},\ }\href {\doibase 10.1103/PhysRevD.95.025002} {\bibfield  {journal} {\bibinfo  {journal} {Phys. Rev. D}\ }\textbf {\bibinfo {volume} {95}},\ \bibinfo {pages} {025002} (\bibinfo {year} {2017}{\natexlab{a}})}\BibitemShut {NoStop}%
\bibitem [{\citenamefont {Jonsson}\ \emph {et~al.}(2018)\citenamefont {Jonsson}, \citenamefont {Ried}, \citenamefont {Martín-Martínez},\ and\ \citenamefont {Kempf}}]{Jonsson2018qubit}%
  \BibitemOpen
  \bibfield  {author} {\bibinfo {author} {\bibfnamefont {R.~H.}\ \bibnamefont {Jonsson}}, \bibinfo {author} {\bibfnamefont {K.}~\bibnamefont {Ried}}, \bibinfo {author} {\bibfnamefont {E.}~\bibnamefont {Martín-Martínez}}, \ and\ \bibinfo {author} {\bibfnamefont {A.}~\bibnamefont {Kempf}},\ }\href {\doibase 10.1088/1751-8121/aae78a} {\bibfield  {journal} {\bibinfo  {journal} {Journal of Physics A: Mathematical and Theoretical}\ }\textbf {\bibinfo {volume} {51}},\ \bibinfo {pages} {485301} (\bibinfo {year} {2018})}\BibitemShut {NoStop}%
\bibitem [{\citenamefont {Simidzija}\ \emph {et~al.}(2020)\citenamefont {Simidzija}, \citenamefont {Ahmadzadegan}, \citenamefont {Kempf},\ and\ \citenamefont {Mart\'{\i}n-Mart\'{\i}nez}}]{Simidzija2020capacity}%
  \BibitemOpen
  \bibfield  {author} {\bibinfo {author} {\bibfnamefont {P.}~\bibnamefont {Simidzija}}, \bibinfo {author} {\bibfnamefont {A.}~\bibnamefont {Ahmadzadegan}}, \bibinfo {author} {\bibfnamefont {A.}~\bibnamefont {Kempf}}, \ and\ \bibinfo {author} {\bibfnamefont {E.}~\bibnamefont {Mart\'{\i}n-Mart\'{\i}nez}},\ }\href {\doibase 10.1103/PhysRevD.101.036014} {\bibfield  {journal} {\bibinfo  {journal} {Phys. Rev. D}\ }\textbf {\bibinfo {volume} {101}},\ \bibinfo {pages} {036014} (\bibinfo {year} {2020})}\BibitemShut {NoStop}%
\bibitem [{\citenamefont {Barcellos}\ and\ \citenamefont {Landulfo}(2024)}]{barcellos2024broadcast}%
  \BibitemOpen
  \bibfield  {author} {\bibinfo {author} {\bibfnamefont {I.~B.}\ \bibnamefont {Barcellos}}\ and\ \bibinfo {author} {\bibfnamefont {A.~G.~S.}\ \bibnamefont {Landulfo}},\ }\href {\doibase 10.1103/PhysRevD.109.065020} {\bibfield  {journal} {\bibinfo  {journal} {Phys. Rev. D}\ }\textbf {\bibinfo {volume} {109}},\ \bibinfo {pages} {065020} (\bibinfo {year} {2024})}\BibitemShut {NoStop}%
\bibitem [{\citenamefont {Mart\'{\i}n-Mart\'{\i}nez}\ \emph {et~al.}(2020)\citenamefont {Mart\'{\i}n-Mart\'{\i}nez}, \citenamefont {Perche},\ and\ \citenamefont {de~S.~L.~Torres}}]{Tales2020GRQO}%
  \BibitemOpen
  \bibfield  {author} {\bibinfo {author} {\bibfnamefont {E.}~\bibnamefont {Mart\'{\i}n-Mart\'{\i}nez}}, \bibinfo {author} {\bibfnamefont {T.~R.}\ \bibnamefont {Perche}}, \ and\ \bibinfo {author} {\bibfnamefont {B.}~\bibnamefont {de~S.~L.~Torres}},\ }\href {\doibase 10.1103/PhysRevD.101.045017} {\bibfield  {journal} {\bibinfo  {journal} {Phys. Rev. D}\ }\textbf {\bibinfo {volume} {101}},\ \bibinfo {pages} {045017} (\bibinfo {year} {2020})}\BibitemShut {NoStop}%
\bibitem [{\citenamefont {Lima}\ \emph {et~al.}(2023{\natexlab{a}})\citenamefont {Lima}, \citenamefont {Patterson}, \citenamefont {Tjoa},\ and\ \citenamefont {Mann}}]{tjoa2023qudit}%
  \BibitemOpen
  \bibfield  {author} {\bibinfo {author} {\bibfnamefont {C.}~\bibnamefont {Lima}}, \bibinfo {author} {\bibfnamefont {E.}~\bibnamefont {Patterson}}, \bibinfo {author} {\bibfnamefont {E.}~\bibnamefont {Tjoa}}, \ and\ \bibinfo {author} {\bibfnamefont {R.~B.}\ \bibnamefont {Mann}},\ }\href {\doibase 10.1103/PhysRevD.108.105020} {\bibfield  {journal} {\bibinfo  {journal} {Phys. Rev. D}\ }\textbf {\bibinfo {volume} {108}},\ \bibinfo {pages} {105020} (\bibinfo {year} {2023}{\natexlab{a}})}\BibitemShut {NoStop}%
\bibitem [{\citenamefont {Lopp}\ and\ \citenamefont {Mart\'{\i}n-Mart\'{\i}nez}(2021)}]{Lopp2021deloc}%
  \BibitemOpen
  \bibfield  {author} {\bibinfo {author} {\bibfnamefont {R.}~\bibnamefont {Lopp}}\ and\ \bibinfo {author} {\bibfnamefont {E.}~\bibnamefont {Mart\'{\i}n-Mart\'{\i}nez}},\ }\href {\doibase 10.1103/PhysRevA.103.013703} {\bibfield  {journal} {\bibinfo  {journal} {Phys. Rev. A}\ }\textbf {\bibinfo {volume} {103}},\ \bibinfo {pages} {013703} (\bibinfo {year} {2021})}\BibitemShut {NoStop}%
\bibitem [{\citenamefont {Tjoa}(2023)}]{tjoa2023nonperturbative}%
  \BibitemOpen
  \bibfield  {author} {\bibinfo {author} {\bibfnamefont {E.}~\bibnamefont {Tjoa}},\ }\href {\doibase 10.1103/PhysRevD.108.045003} {\bibfield  {journal} {\bibinfo  {journal} {Phys. Rev. D}\ }\textbf {\bibinfo {volume} {108}},\ \bibinfo {pages} {045003} (\bibinfo {year} {2023})}\BibitemShut {NoStop}%
\bibitem [{\citenamefont {Tjoa}\ and\ \citenamefont {Gray}(2024)}]{tjoa2024UDW}%
  \BibitemOpen
  \bibfield  {author} {\bibinfo {author} {\bibfnamefont {E.}~\bibnamefont {Tjoa}}\ and\ \bibinfo {author} {\bibfnamefont {F.}~\bibnamefont {Gray}},\ }\href {\doibase 10.1088/1751-8121/ad6365} {\bibfield  {journal} {\bibinfo  {journal} {Journal of Physics A: Mathematical and Theoretical}\ }\textbf {\bibinfo {volume} {57}},\ \bibinfo {pages} {325301} (\bibinfo {year} {2024})}\BibitemShut {NoStop}%
\bibitem [{\citenamefont {Mancini}\ \emph {et~al.}(2014)\citenamefont {Mancini}, \citenamefont {Pierini},\ and\ \citenamefont {Wilde}}]{mancini2014frw}%
  \BibitemOpen
  \bibfield  {author} {\bibinfo {author} {\bibfnamefont {S.}~\bibnamefont {Mancini}}, \bibinfo {author} {\bibfnamefont {R.}~\bibnamefont {Pierini}}, \ and\ \bibinfo {author} {\bibfnamefont {M.~M.}\ \bibnamefont {Wilde}},\ }\href {\doibase 10.1088/1367-2630/16/12/123049} {\bibfield  {journal} {\bibinfo  {journal} {New Journal of Physics}\ }\textbf {\bibinfo {volume} {16}},\ \bibinfo {pages} {123049} (\bibinfo {year} {2014})}\BibitemShut {NoStop}%
\bibitem [{\citenamefont {Lapponi}\ \emph {et~al.}(2024)\citenamefont {Lapponi}, \citenamefont {Louko},\ and\ \citenamefont {Mancini}}]{lapponi2024comm2}%
  \BibitemOpen
  \bibfield  {author} {\bibinfo {author} {\bibfnamefont {A.}~\bibnamefont {Lapponi}}, \bibinfo {author} {\bibfnamefont {J.}~\bibnamefont {Louko}}, \ and\ \bibinfo {author} {\bibfnamefont {S.}~\bibnamefont {Mancini}},\ }\href {\doibase 10.1103/PhysRevD.110.025018} {\bibfield  {journal} {\bibinfo  {journal} {Phys. Rev. D}\ }\textbf {\bibinfo {volume} {110}},\ \bibinfo {pages} {025018} (\bibinfo {year} {2024})}\BibitemShut {NoStop}%
\bibitem [{\citenamefont {Wald}(1994)}]{wald1994quantum}%
  \BibitemOpen
  \bibfield  {author} {\bibinfo {author} {\bibfnamefont {R.~M.}\ \bibnamefont {Wald}},\ }\href {https://books.google.ca/books?id=Iud7eyDxT1AC} {\emph {\bibinfo {title} {Quantum Field Theory in Curved Spacetime and Black Hole Thermodynamics}}},\ Chicago Lectures in Physics\ (\bibinfo  {publisher} {University of Chicago Press},\ \bibinfo {year} {1994})\BibitemShut {NoStop}%
\bibitem [{\citenamefont {Fewster}\ and\ \citenamefont {Rejzner}(2019)}]{fewster2019algebraic}%
  \BibitemOpen
  \bibfield  {author} {\bibinfo {author} {\bibfnamefont {C.~J.}\ \bibnamefont {Fewster}}\ and\ \bibinfo {author} {\bibfnamefont {K.}~\bibnamefont {Rejzner}},\ }\href@noop {} {\enquote {\bibinfo {title} {Algebraic quantum field theory -- an introduction},}\ } (\bibinfo {year} {2019}),\ \Eprint {http://arxiv.org/abs/1904.04051} {arXiv:1904.04051 [hep-th]} \BibitemShut {NoStop}%
\bibitem [{\citenamefont {Khavkine}\ and\ \citenamefont {Moretti}(2015)}]{Khavkhine2015AQFT}%
  \BibitemOpen
  \bibfield  {author} {\bibinfo {author} {\bibfnamefont {I.}~\bibnamefont {Khavkine}}\ and\ \bibinfo {author} {\bibfnamefont {V.}~\bibnamefont {Moretti}},\ }\href {\doibase 10.1007/978-3-319-21353-8_5} {\bibfield  {journal} {\bibinfo  {journal} {Mathematical Physics Studies}\ ,\ \bibinfo {pages} {191–251}} (\bibinfo {year} {2015})}\BibitemShut {NoStop}%
\bibitem [{\citenamefont {Kay}\ and\ \citenamefont {Wald}(1991)}]{KayWald1991theorems}%
  \BibitemOpen
  \bibfield  {author} {\bibinfo {author} {\bibfnamefont {B.~S.}\ \bibnamefont {Kay}}\ and\ \bibinfo {author} {\bibfnamefont {R.~M.}\ \bibnamefont {Wald}},\ }\href {\doibase https://doi.org/10.1016/0370-1573(91)90015-E} {\bibfield  {journal} {\bibinfo  {journal} {Physics Reports}\ }\textbf {\bibinfo {volume} {207}},\ \bibinfo {pages} {49} (\bibinfo {year} {1991})}\BibitemShut {NoStop}%
\bibitem [{\citenamefont {Hack}(2015)}]{hack2015cosmological}%
  \BibitemOpen
  \bibfield  {author} {\bibinfo {author} {\bibfnamefont {T.-P.}\ \bibnamefont {Hack}},\ }\href {https://link.springer.com/book/10.1007/978-3-319-21894-6} {\emph {\bibinfo {title} {Cosmological applications of algebraic quantum field theory in curved spacetimes}}},\ Vol.~\bibinfo {volume} {6}\ (\bibinfo  {publisher} {Springer},\ \bibinfo {year} {2015})\BibitemShut {NoStop}%
\bibitem [{\citenamefont {Bratteli}(1972)}]{Bratteli1972afalgebra}%
  \BibitemOpen
  \bibfield  {author} {\bibinfo {author} {\bibfnamefont {O.}~\bibnamefont {Bratteli}},\ }\href {http://www.jstor.org/stable/1996380} {\bibfield  {journal} {\bibinfo  {journal} {Transactions of the American Mathematical Society}\ }\textbf {\bibinfo {volume} {171}},\ \bibinfo {pages} {195} (\bibinfo {year} {1972})}\BibitemShut {NoStop}%
\bibitem [{\citenamefont {Poisson}(2009)}]{poisson2009toolkit}%
  \BibitemOpen
  \bibfield  {author} {\bibinfo {author} {\bibfnamefont {E.}~\bibnamefont {Poisson}},\ }\href {\doibase 10.1017/CBO9780511606601} {\emph {\bibinfo {title} {{A Relativist's Toolkit: The Mathematics of Black-Hole Mechanics}}}}\ (\bibinfo  {publisher} {Cambridge University Press},\ \bibinfo {year} {2009})\BibitemShut {NoStop}%
\bibitem [{\citenamefont {Wald}(2010)}]{wald2010general}%
  \BibitemOpen
  \bibfield  {author} {\bibinfo {author} {\bibfnamefont {R.}~\bibnamefont {Wald}},\ }\href {https://books.google.ca/books?id=9S-hzg6-moYC} {\emph {\bibinfo {title} {General Relativity}}}\ (\bibinfo  {publisher} {University of Chicago Press},\ \bibinfo {year} {2010})\BibitemShut {NoStop}%
\bibitem [{\citenamefont {Ruep}(2021)}]{ruep2021weakly}%
  \BibitemOpen
  \bibfield  {author} {\bibinfo {author} {\bibfnamefont {M.~H.}\ \bibnamefont {Ruep}},\ }\href {\doibase 10.1088/1361-6382/ac1b08} {\bibfield  {journal} {\bibinfo  {journal} {Classical and Quantum Gravity}\ }\textbf {\bibinfo {volume} {38}},\ \bibinfo {pages} {195029} (\bibinfo {year} {2021})}\BibitemShut {NoStop}%
\bibitem [{\citenamefont {Radzikowski}(1996)}]{Radzikowski1996microlocal}%
  \BibitemOpen
  \bibfield  {author} {\bibinfo {author} {\bibfnamefont {M.~J.}\ \bibnamefont {Radzikowski}},\ }\href {\doibase cmp/1104287114} {\bibfield  {journal} {\bibinfo  {journal} {Communications in Mathematical Physics}\ }\textbf {\bibinfo {volume} {179}},\ \bibinfo {pages} {529 } (\bibinfo {year} {1996})}\BibitemShut {NoStop}%
\bibitem [{\citenamefont {Kubo}(1957)}]{kubo1957statistical}%
  \BibitemOpen
  \bibfield  {author} {\bibinfo {author} {\bibfnamefont {R.}~\bibnamefont {Kubo}},\ }\href {\doibase https://doi.org/10.1143/JPSJ.12.570} {\bibfield  {journal} {\bibinfo  {journal} {Journal of the physical society of Japan}\ }\textbf {\bibinfo {volume} {12}},\ \bibinfo {pages} {570} (\bibinfo {year} {1957})}\BibitemShut {NoStop}%
\bibitem [{\citenamefont {Martin}\ and\ \citenamefont {Schwinger}(1959)}]{martinSchwinger1959theory}%
  \BibitemOpen
  \bibfield  {author} {\bibinfo {author} {\bibfnamefont {P.~C.}\ \bibnamefont {Martin}}\ and\ \bibinfo {author} {\bibfnamefont {J.}~\bibnamefont {Schwinger}},\ }\href {\doibase 10.1103/PhysRev.115.1342} {\bibfield  {journal} {\bibinfo  {journal} {Phys. Rev.}\ }\textbf {\bibinfo {volume} {115}},\ \bibinfo {pages} {1342} (\bibinfo {year} {1959})}\BibitemShut {NoStop}%
\bibitem [{\citenamefont {Spohn}(1989)}]{spohn1989spinboson}%
  \BibitemOpen
  \bibfield  {author} {\bibinfo {author} {\bibfnamefont {H.}~\bibnamefont {Spohn}},\ }\href {\doibase https://doi.org/10.1007/BF01238859} {\bibfield  {journal} {\bibinfo  {journal} {Communications in Mathematical Physics}\ }\textbf {\bibinfo {volume} {123}},\ \bibinfo {pages} {277 } (\bibinfo {year} {1989})}\BibitemShut {NoStop}%
\bibitem [{\citenamefont {Hasler}\ \emph {et~al.}(2021)\citenamefont {Hasler}, \citenamefont {Hinrichs},\ and\ \citenamefont {Siebert}}]{hasler2021existence}%
  \BibitemOpen
  \bibfield  {author} {\bibinfo {author} {\bibfnamefont {D.}~\bibnamefont {Hasler}}, \bibinfo {author} {\bibfnamefont {B.}~\bibnamefont {Hinrichs}}, \ and\ \bibinfo {author} {\bibfnamefont {O.}~\bibnamefont {Siebert}},\ }\href {\doibase https://doi.org/10.1007/s00220-021-04185-w} {\bibfield  {journal} {\bibinfo  {journal} {Communications in Mathematical Physics}\ }\textbf {\bibinfo {volume} {388}},\ \bibinfo {pages} {419} (\bibinfo {year} {2021})}\BibitemShut {NoStop}%
\bibitem [{\citenamefont {Fannes}\ \emph {et~al.}(1988)\citenamefont {Fannes}, \citenamefont {Nachtergaele},\ and\ \citenamefont {Verbeure}}]{fannes1988equilibrium}%
  \BibitemOpen
  \bibfield  {author} {\bibinfo {author} {\bibfnamefont {M.}~\bibnamefont {Fannes}}, \bibinfo {author} {\bibfnamefont {B.}~\bibnamefont {Nachtergaele}}, \ and\ \bibinfo {author} {\bibfnamefont {A.}~\bibnamefont {Verbeure}},\ }\href {\doibase https://doi.org/10.1007/BF01229453} {\bibfield  {journal} {\bibinfo  {journal} {Communications in mathematical physics}\ }\textbf {\bibinfo {volume} {114}},\ \bibinfo {pages} {537} (\bibinfo {year} {1988})}\BibitemShut {NoStop}%
\bibitem [{\citenamefont {Mart\'{\i}n-Mart\'{\i}nez}\ \emph {et~al.}(2021)\citenamefont {Mart\'{\i}n-Mart\'{\i}nez}, \citenamefont {Perche},\ and\ \citenamefont {Torres}}]{Bruno2020time-ordering}%
  \BibitemOpen
  \bibfield  {author} {\bibinfo {author} {\bibfnamefont {E.}~\bibnamefont {Mart\'{\i}n-Mart\'{\i}nez}}, \bibinfo {author} {\bibfnamefont {T.~R.}\ \bibnamefont {Perche}}, \ and\ \bibinfo {author} {\bibfnamefont {B.~d. S.~L.}\ \bibnamefont {Torres}},\ }\href {\doibase 10.1103/PhysRevD.103.025007} {\bibfield  {journal} {\bibinfo  {journal} {Phys. Rev. D}\ }\textbf {\bibinfo {volume} {103}},\ \bibinfo {pages} {025007} (\bibinfo {year} {2021})}\BibitemShut {NoStop}%
\bibitem [{\citenamefont {Simidzija}\ \emph {et~al.}(2018)\citenamefont {Simidzija}, \citenamefont {Jonsson},\ and\ \citenamefont {Mart\'{\i}n-Mart\'{\i}nez}}]{Simidzija2018no-go}%
  \BibitemOpen
  \bibfield  {author} {\bibinfo {author} {\bibfnamefont {P.}~\bibnamefont {Simidzija}}, \bibinfo {author} {\bibfnamefont {R.~H.}\ \bibnamefont {Jonsson}}, \ and\ \bibinfo {author} {\bibfnamefont {E.}~\bibnamefont {Mart\'{\i}n-Mart\'{\i}nez}},\ }\href {\doibase 10.1103/PhysRevD.97.125002} {\bibfield  {journal} {\bibinfo  {journal} {Phys. Rev. D}\ }\textbf {\bibinfo {volume} {97}},\ \bibinfo {pages} {125002} (\bibinfo {year} {2018})}\BibitemShut {NoStop}%
\bibitem [{\citenamefont {Simidzija}\ and\ \citenamefont {Mart\'{\i}n-Mart\'{\i}nez}(2017{\natexlab{b}})}]{Simidzija2017coherent}%
  \BibitemOpen
  \bibfield  {author} {\bibinfo {author} {\bibfnamefont {P.}~\bibnamefont {Simidzija}}\ and\ \bibinfo {author} {\bibfnamefont {E.}~\bibnamefont {Mart\'{\i}n-Mart\'{\i}nez}},\ }\href {\doibase 10.1103/PhysRevD.96.025020} {\bibfield  {journal} {\bibinfo  {journal} {Phys. Rev. D}\ }\textbf {\bibinfo {volume} {96}},\ \bibinfo {pages} {025020} (\bibinfo {year} {2017}{\natexlab{b}})}\BibitemShut {NoStop}%
\bibitem [{\citenamefont {Gallock-Yoshimura}(2024)}]{gallock2024relativistic}%
  \BibitemOpen
  \bibfield  {author} {\bibinfo {author} {\bibfnamefont {K.}~\bibnamefont {Gallock-Yoshimura}},\ }\href {\doibase https://doi.org/10.1007/JHEP01(2024)198} {\bibfield  {journal} {\bibinfo  {journal} {Journal of High Energy Physics}\ }\textbf {\bibinfo {volume} {2024}},\ \bibinfo {pages} {1} (\bibinfo {year} {2024})}\BibitemShut {NoStop}%
\bibitem [{\citenamefont {Kollas}\ and\ \citenamefont {Moustos}(2024)}]{kollas2024engine}%
  \BibitemOpen
  \bibfield  {author} {\bibinfo {author} {\bibfnamefont {N.~K.}\ \bibnamefont {Kollas}}\ and\ \bibinfo {author} {\bibfnamefont {D.}~\bibnamefont {Moustos}},\ }\href {\doibase 10.1103/PhysRevD.109.065025} {\bibfield  {journal} {\bibinfo  {journal} {Phys. Rev. D}\ }\textbf {\bibinfo {volume} {109}},\ \bibinfo {pages} {065025} (\bibinfo {year} {2024})}\BibitemShut {NoStop}%
\bibitem [{\citenamefont {Polo-G\'omez}\ and\ \citenamefont {Mart\'{\i}n-Mart\'{\i}nez}(2024)}]{polo2024sequence}%
  \BibitemOpen
  \bibfield  {author} {\bibinfo {author} {\bibfnamefont {J.}~\bibnamefont {Polo-G\'omez}}\ and\ \bibinfo {author} {\bibfnamefont {E.}~\bibnamefont {Mart\'{\i}n-Mart\'{\i}nez}},\ }\href {\doibase 10.1103/PhysRevD.109.045014} {\bibfield  {journal} {\bibinfo  {journal} {Phys. Rev. D}\ }\textbf {\bibinfo {volume} {109}},\ \bibinfo {pages} {045014} (\bibinfo {year} {2024})}\BibitemShut {NoStop}%
\bibitem [{\citenamefont {Poisson}\ \emph {et~al.}(2011)\citenamefont {Poisson}, \citenamefont {Pound},\ and\ \citenamefont {Vega}}]{poisson2011motion}%
  \BibitemOpen
  \bibfield  {author} {\bibinfo {author} {\bibfnamefont {E.}~\bibnamefont {Poisson}}, \bibinfo {author} {\bibfnamefont {A.}~\bibnamefont {Pound}}, \ and\ \bibinfo {author} {\bibfnamefont {I.}~\bibnamefont {Vega}},\ }\href@noop {} {\bibfield  {journal} {\bibinfo  {journal} {Living Reviews in Relativity}\ }\textbf {\bibinfo {volume} {14}},\ \bibinfo {pages} {1} (\bibinfo {year} {2011})}\BibitemShut {NoStop}%
\bibitem [{\citenamefont {Dimock}(1980)}]{dimock1980algebras}%
  \BibitemOpen
  \bibfield  {author} {\bibinfo {author} {\bibfnamefont {J.}~\bibnamefont {Dimock}},\ }\href {\doibase 10.1007/BF01269921} {\bibfield  {journal} {\bibinfo  {journal} {Communications in Mathematical Physics}\ }\textbf {\bibinfo {volume} {77}},\ \bibinfo {pages} {219} (\bibinfo {year} {1980})}\BibitemShut {NoStop}%
\bibitem [{\citenamefont {Kay}(1978)}]{kay1978linear}%
  \BibitemOpen
  \bibfield  {author} {\bibinfo {author} {\bibfnamefont {B.~S.}\ \bibnamefont {Kay}},\ }\href {\doibase https://doi.org/10.1007/BF01230084} {\bibfield  {journal} {\bibinfo  {journal} {Commun. Math. Phys}\ }\textbf {\bibinfo {volume} {62}},\ \bibinfo {pages} {55} (\bibinfo {year} {1978})}\BibitemShut {NoStop}%
\bibitem [{\citenamefont {Khatri}\ and\ \citenamefont {Wilde}(2020)}]{khatri2020principles}%
  \BibitemOpen
  \bibfield  {author} {\bibinfo {author} {\bibfnamefont {S.}~\bibnamefont {Khatri}}\ and\ \bibinfo {author} {\bibfnamefont {M.~M.}\ \bibnamefont {Wilde}},\ }\href {https://arxiv.org/abs/2011.04672} {\bibfield  {journal} {\bibinfo  {journal} {arXiv preprint arXiv:2011.04672}\ } (\bibinfo {year} {2020})}\BibitemShut {NoStop}%
\bibitem [{\citenamefont {Wilde}(2017)}]{wilde2011classical}%
  \BibitemOpen
  \bibfield  {author} {\bibinfo {author} {\bibfnamefont {M.~M.}\ \bibnamefont {Wilde}},\ }\enquote {\bibinfo {title} {Preface to the second edition},}\ in\ \href@noop {} {\emph {\bibinfo {booktitle} {Quantum Information Theory}}}\ (\bibinfo  {publisher} {Cambridge University Press},\ \bibinfo {year} {2017})\ p.\ \bibinfo {pages} {xi–xii}\BibitemShut {NoStop}%
\bibitem [{\citenamefont {Gyongyosi}\ \emph {et~al.}(2018)\citenamefont {Gyongyosi}, \citenamefont {Imre},\ and\ \citenamefont {Nguyen}}]{gyongyosi2018surveycapacity}%
  \BibitemOpen
  \bibfield  {author} {\bibinfo {author} {\bibfnamefont {L.}~\bibnamefont {Gyongyosi}}, \bibinfo {author} {\bibfnamefont {S.}~\bibnamefont {Imre}}, \ and\ \bibinfo {author} {\bibfnamefont {H.~V.}\ \bibnamefont {Nguyen}},\ }\href {\doibase 10.1109/COMST.2017.2786748} {\bibfield  {journal} {\bibinfo  {journal} {IEEE Communications Surveys \& Tutorials}\ }\textbf {\bibinfo {volume} {20}},\ \bibinfo {pages} {1149} (\bibinfo {year} {2018})}\BibitemShut {NoStop}%
\bibitem [{\citenamefont {Holevo}(2020)}]{Holevo2020capacities}%
  \BibitemOpen
  \bibfield  {author} {\bibinfo {author} {\bibfnamefont {A.}~\bibnamefont {Holevo}},\ }\href {\doibase 10.1070/QEL17285} {\bibfield  {journal} {\bibinfo  {journal} {Quantum Electronics}\ }\textbf {\bibinfo {volume} {50}},\ \bibinfo {pages} {440} (\bibinfo {year} {2020})}\BibitemShut {NoStop}%
\bibitem [{\citenamefont {Schumacher}\ and\ \citenamefont {Westmoreland}(2002)}]{schumacher2002relative}%
  \BibitemOpen
  \bibfield  {author} {\bibinfo {author} {\bibfnamefont {B.}~\bibnamefont {Schumacher}}\ and\ \bibinfo {author} {\bibfnamefont {M.~D.}\ \bibnamefont {Westmoreland}},\ }\href@noop {} {\bibfield  {journal} {\bibinfo  {journal} {Contemporary Mathematics}\ }\textbf {\bibinfo {volume} {305}},\ \bibinfo {pages} {265} (\bibinfo {year} {2002})}\BibitemShut {NoStop}%
\bibitem [{\citenamefont {Lloyd}(1997)}]{lloyd1997capacity}%
  \BibitemOpen
  \bibfield  {author} {\bibinfo {author} {\bibfnamefont {S.}~\bibnamefont {Lloyd}},\ }\href {\doibase 10.1103/PhysRevA.55.1613} {\bibfield  {journal} {\bibinfo  {journal} {Phys. Rev. A}\ }\textbf {\bibinfo {volume} {55}},\ \bibinfo {pages} {1613} (\bibinfo {year} {1997})}\BibitemShut {NoStop}%
\bibitem [{\citenamefont {Shor}(2002)}]{shor2002capacity}%
  \BibitemOpen
  \bibfield  {author} {\bibinfo {author} {\bibfnamefont {P.~W.}\ \bibnamefont {Shor}},\ }in\ \href@noop {} {\emph {\bibinfo {booktitle} {lecture notes, MSRI Workshop on Quantum Computation}}}\ (\bibinfo {year} {2002})\BibitemShut {NoStop}%
\bibitem [{\citenamefont {Devetak}\ and\ \citenamefont {Winter}(2003)}]{devetak2003capacity}%
  \BibitemOpen
  \bibfield  {author} {\bibinfo {author} {\bibfnamefont {I.}~\bibnamefont {Devetak}}\ and\ \bibinfo {author} {\bibfnamefont {A.}~\bibnamefont {Winter}},\ }\href {\doibase 10.1103/PhysRevA.68.042301} {\bibfield  {journal} {\bibinfo  {journal} {Phys. Rev. A}\ }\textbf {\bibinfo {volume} {68}},\ \bibinfo {pages} {042301} (\bibinfo {year} {2003})}\BibitemShut {NoStop}%
\bibitem [{\citenamefont {Hastings}(2009)}]{hastings2009superadditivity}%
  \BibitemOpen
  \bibfield  {author} {\bibinfo {author} {\bibfnamefont {M.~B.}\ \bibnamefont {Hastings}},\ }\href {\doibase https://doi.org/10.1038/nphys1224} {\bibfield  {journal} {\bibinfo  {journal} {Nature Physics}\ }\textbf {\bibinfo {volume} {5}},\ \bibinfo {pages} {255} (\bibinfo {year} {2009})}\BibitemShut {NoStop}%
\bibitem [{\citenamefont {Smith}\ and\ \citenamefont {Yard}(2008)}]{smith2008superadditive}%
  \BibitemOpen
  \bibfield  {author} {\bibinfo {author} {\bibfnamefont {G.}~\bibnamefont {Smith}}\ and\ \bibinfo {author} {\bibfnamefont {J.}~\bibnamefont {Yard}},\ }\href {\doibase 10.1126/science.1162242} {\bibfield  {journal} {\bibinfo  {journal} {Science}\ }\textbf {\bibinfo {volume} {321}},\ \bibinfo {pages} {1812} (\bibinfo {year} {2008})}\BibitemShut {NoStop}%
\bibitem [{\citenamefont {Leditzky}\ \emph {et~al.}(2023)\citenamefont {Leditzky}, \citenamefont {Leung}, \citenamefont {Siddhu}, \citenamefont {Smith},\ and\ \citenamefont {Smolin}}]{leditzky2023nonadditive}%
  \BibitemOpen
  \bibfield  {author} {\bibinfo {author} {\bibfnamefont {F.}~\bibnamefont {Leditzky}}, \bibinfo {author} {\bibfnamefont {D.}~\bibnamefont {Leung}}, \bibinfo {author} {\bibfnamefont {V.}~\bibnamefont {Siddhu}}, \bibinfo {author} {\bibfnamefont {G.}~\bibnamefont {Smith}}, \ and\ \bibinfo {author} {\bibfnamefont {J.~A.}\ \bibnamefont {Smolin}},\ }\href {\doibase 10.1103/PhysRevLett.130.200801} {\bibfield  {journal} {\bibinfo  {journal} {Phys. Rev. Lett.}\ }\textbf {\bibinfo {volume} {130}},\ \bibinfo {pages} {200801} (\bibinfo {year} {2023})}\BibitemShut {NoStop}%
\bibitem [{\citenamefont {Cubitt}\ \emph {et~al.}(2008)\citenamefont {Cubitt}, \citenamefont {Ruskai},\ and\ \citenamefont {Smith}}]{cubitt2008structure}%
  \BibitemOpen
  \bibfield  {author} {\bibinfo {author} {\bibfnamefont {T.~S.}\ \bibnamefont {Cubitt}}, \bibinfo {author} {\bibfnamefont {M.~B.}\ \bibnamefont {Ruskai}}, \ and\ \bibinfo {author} {\bibfnamefont {G.}~\bibnamefont {Smith}},\ }\href {\doibase 10.1063/1.2953685} {\bibfield  {journal} {\bibinfo  {journal} {Journal of Mathematical Physics}\ }\textbf {\bibinfo {volume} {49}},\ \bibinfo {pages} {102104} (\bibinfo {year} {2008})}\BibitemShut {NoStop}%
\bibitem [{\citenamefont {Horodecki}\ \emph {et~al.}(2003)\citenamefont {Horodecki}, \citenamefont {Shor},\ and\ \citenamefont {Ruskai}}]{ruskai2003entanglementbreaking}%
  \BibitemOpen
  \bibfield  {author} {\bibinfo {author} {\bibfnamefont {M.}~\bibnamefont {Horodecki}}, \bibinfo {author} {\bibfnamefont {P.~W.}\ \bibnamefont {Shor}}, \ and\ \bibinfo {author} {\bibfnamefont {M.~B.}\ \bibnamefont {Ruskai}},\ }\href {\doibase 10.1142/S0129055X03001709} {\bibfield  {journal} {\bibinfo  {journal} {Reviews in Mathematical Physics}\ }\textbf {\bibinfo {volume} {15}},\ \bibinfo {pages} {629} (\bibinfo {year} {2003})}\BibitemShut {NoStop}%
\bibitem [{\citenamefont {Vidal}\ and\ \citenamefont {Werner}(2002)}]{vidal2002negativity}%
  \BibitemOpen
  \bibfield  {author} {\bibinfo {author} {\bibfnamefont {G.}~\bibnamefont {Vidal}}\ and\ \bibinfo {author} {\bibfnamefont {R.~F.}\ \bibnamefont {Werner}},\ }\href {\doibase 10.1103/PhysRevA.65.032314} {\bibfield  {journal} {\bibinfo  {journal} {Phys. Rev. A}\ }\textbf {\bibinfo {volume} {65}},\ \bibinfo {pages} {032314} (\bibinfo {year} {2002})}\BibitemShut {NoStop}%
\bibitem [{\citenamefont {Smith}\ and\ \citenamefont {Smolin}(2012)}]{graeme2012detect}%
  \BibitemOpen
  \bibfield  {author} {\bibinfo {author} {\bibfnamefont {G.}~\bibnamefont {Smith}}\ and\ \bibinfo {author} {\bibfnamefont {J.~A.}\ \bibnamefont {Smolin}},\ }\href {\doibase 10.1103/PhysRevLett.108.230507} {\bibfield  {journal} {\bibinfo  {journal} {Phys. Rev. Lett.}\ }\textbf {\bibinfo {volume} {108}},\ \bibinfo {pages} {230507} (\bibinfo {year} {2012})}\BibitemShut {NoStop}%
\bibitem [{\citenamefont {Singh}\ and\ \citenamefont {Datta}(2022)}]{singh2022detecting}%
  \BibitemOpen
  \bibfield  {author} {\bibinfo {author} {\bibfnamefont {S.}~\bibnamefont {Singh}}\ and\ \bibinfo {author} {\bibfnamefont {N.}~\bibnamefont {Datta}},\ }\href {\doibase https://doi.org/10.1038/s41534-022-00550-2} {\bibfield  {journal} {\bibinfo  {journal} {npj Quantum Information}\ }\textbf {\bibinfo {volume} {8}},\ \bibinfo {pages} {50} (\bibinfo {year} {2022})}\BibitemShut {NoStop}%
\bibitem [{\citenamefont {McLenaghan}(1974)}]{McLenaghan1974huygens}%
  \BibitemOpen
  \bibfield  {author} {\bibinfo {author} {\bibfnamefont {R.~G.}\ \bibnamefont {McLenaghan}},\ }\href {{http://www.numdam.org/item/AIHPA_1974__20_2_153_0/}} {\bibfield  {journal} {\bibinfo  {journal} {Annales de l'I.H.P. Physique th\'eorique}\ }\textbf {\bibinfo {volume} {20}},\ \bibinfo {pages} {153} (\bibinfo {year} {1974})}\BibitemShut {NoStop}%
\bibitem [{\citenamefont {Sonego}\ and\ \citenamefont {Faraoni}(1992)}]{Sonego1992huygenscurved}%
  \BibitemOpen
  \bibfield  {author} {\bibinfo {author} {\bibfnamefont {S.}~\bibnamefont {Sonego}}\ and\ \bibinfo {author} {\bibfnamefont {V.}~\bibnamefont {Faraoni}},\ }\href {\doibase 10.1063/1.529798} {\bibfield  {journal} {\bibinfo  {journal} {Journal of Mathematical Physics}\ }\textbf {\bibinfo {volume} {33}},\ \bibinfo {pages} {625} (\bibinfo {year} {1992})},\ \Eprint {http://arxiv.org/abs/https://doi.org/10.1063/1.529798} {https://doi.org/10.1063/1.529798} \BibitemShut {NoStop}%
\bibitem [{\citenamefont {Faraoni}(2019)}]{Faraoni2019huygens}%
  \BibitemOpen
  \bibfield  {author} {\bibinfo {author} {\bibfnamefont {V.}~\bibnamefont {Faraoni}},\ }\href {\doibase 10.3390/sym11010036} {\bibfield  {journal} {\bibinfo  {journal} {Symmetry}\ }\textbf {\bibinfo {volume} {11}} (\bibinfo {year} {2019}),\ 10.3390/sym11010036}\BibitemShut {NoStop}%
\bibitem [{\citenamefont {Mart\'{\i}n-Mart\'{\i}nez}(2015)}]{Causality2015Eduardo}%
  \BibitemOpen
  \bibfield  {author} {\bibinfo {author} {\bibfnamefont {E.}~\bibnamefont {Mart\'{\i}n-Mart\'{\i}nez}},\ }\href {\doibase 10.1103/PhysRevD.92.104019} {\bibfield  {journal} {\bibinfo  {journal} {Phys. Rev. D}\ }\textbf {\bibinfo {volume} {92}},\ \bibinfo {pages} {104019} (\bibinfo {year} {2015})}\BibitemShut {NoStop}%
\bibitem [{\citenamefont {Jonsson}\ \emph {et~al.}(2020)\citenamefont {Jonsson}, \citenamefont {Aruquipa}, \citenamefont {Casals}, \citenamefont {Kempf},\ and\ \citenamefont {Mart\'{\i}n-Mart\'{\i}nez}}]{Casals2020commBH}%
  \BibitemOpen
  \bibfield  {author} {\bibinfo {author} {\bibfnamefont {R.~H.}\ \bibnamefont {Jonsson}}, \bibinfo {author} {\bibfnamefont {D.~Q.}\ \bibnamefont {Aruquipa}}, \bibinfo {author} {\bibfnamefont {M.}~\bibnamefont {Casals}}, \bibinfo {author} {\bibfnamefont {A.}~\bibnamefont {Kempf}}, \ and\ \bibinfo {author} {\bibfnamefont {E.}~\bibnamefont {Mart\'{\i}n-Mart\'{\i}nez}},\ }\href {\doibase 10.1103/PhysRevD.101.125005} {\bibfield  {journal} {\bibinfo  {journal} {Phys. Rev. D}\ }\textbf {\bibinfo {volume} {101}},\ \bibinfo {pages} {125005} (\bibinfo {year} {2020})}\BibitemShut {NoStop}%
\bibitem [{\citenamefont {Wolf}\ and\ \citenamefont {P\'erez-Garc\'{\i}a}(2007)}]{wolf2007degradable}%
  \BibitemOpen
  \bibfield  {author} {\bibinfo {author} {\bibfnamefont {M.~M.}\ \bibnamefont {Wolf}}\ and\ \bibinfo {author} {\bibfnamefont {D.}~\bibnamefont {P\'erez-Garc\'{\i}a}},\ }\href {\doibase 10.1103/PhysRevA.75.012303} {\bibfield  {journal} {\bibinfo  {journal} {Phys. Rev. A}\ }\textbf {\bibinfo {volume} {75}},\ \bibinfo {pages} {012303} (\bibinfo {year} {2007})}\BibitemShut {NoStop}%
\bibitem [{\citenamefont {Bru\ss{}}\ \emph {et~al.}(1998)\citenamefont {Bru\ss{}}, \citenamefont {DiVincenzo}, \citenamefont {Ekert}, \citenamefont {Fuchs}, \citenamefont {Macchiavello},\ and\ \citenamefont {Smolin}}]{bruss1998cloneapprox}%
  \BibitemOpen
  \bibfield  {author} {\bibinfo {author} {\bibfnamefont {D.}~\bibnamefont {Bru\ss{}}}, \bibinfo {author} {\bibfnamefont {D.~P.}\ \bibnamefont {DiVincenzo}}, \bibinfo {author} {\bibfnamefont {A.}~\bibnamefont {Ekert}}, \bibinfo {author} {\bibfnamefont {C.~A.}\ \bibnamefont {Fuchs}}, \bibinfo {author} {\bibfnamefont {C.}~\bibnamefont {Macchiavello}}, \ and\ \bibinfo {author} {\bibfnamefont {J.~A.}\ \bibnamefont {Smolin}},\ }\href {\doibase 10.1103/PhysRevA.57.2368} {\bibfield  {journal} {\bibinfo  {journal} {Phys. Rev. A}\ }\textbf {\bibinfo {volume} {57}},\ \bibinfo {pages} {2368} (\bibinfo {year} {1998})}\BibitemShut {NoStop}%
\bibitem [{\citenamefont {Bennett}\ \emph {et~al.}(1997)\citenamefont {Bennett}, \citenamefont {DiVincenzo},\ and\ \citenamefont {Smolin}}]{bennett1997erasure}%
  \BibitemOpen
  \bibfield  {author} {\bibinfo {author} {\bibfnamefont {C.~H.}\ \bibnamefont {Bennett}}, \bibinfo {author} {\bibfnamefont {D.~P.}\ \bibnamefont {DiVincenzo}}, \ and\ \bibinfo {author} {\bibfnamefont {J.~A.}\ \bibnamefont {Smolin}},\ }\href {\doibase 10.1103/PhysRevLett.78.3217} {\bibfield  {journal} {\bibinfo  {journal} {Phys. Rev. Lett.}\ }\textbf {\bibinfo {volume} {78}},\ \bibinfo {pages} {3217} (\bibinfo {year} {1997})}\BibitemShut {NoStop}%
\bibitem [{\citenamefont {Lima}\ \emph {et~al.}(2023{\natexlab{b}})\citenamefont {Lima}, \citenamefont {Patterson}, \citenamefont {Tjoa},\ and\ \citenamefont {Mann}}]{lima2023unruh}%
  \BibitemOpen
  \bibfield  {author} {\bibinfo {author} {\bibfnamefont {C.}~\bibnamefont {Lima}}, \bibinfo {author} {\bibfnamefont {E.}~\bibnamefont {Patterson}}, \bibinfo {author} {\bibfnamefont {E.}~\bibnamefont {Tjoa}}, \ and\ \bibinfo {author} {\bibfnamefont {R.~B.}\ \bibnamefont {Mann}},\ }\href {\doibase 10.1103/PhysRevD.108.105020} {\bibfield  {journal} {\bibinfo  {journal} {Phys. Rev. D}\ }\textbf {\bibinfo {volume} {108}},\ \bibinfo {pages} {105020} (\bibinfo {year} {2023}{\natexlab{b}})}\BibitemShut {NoStop}%
\bibitem [{\citenamefont {Verdon-Akzam}\ \emph {et~al.}(2016)\citenamefont {Verdon-Akzam}, \citenamefont {Mart\'{\i}n-Mart\'{\i}nez},\ and\ \citenamefont {Kempf}}]{verdon2016quditQET}%
  \BibitemOpen
  \bibfield  {author} {\bibinfo {author} {\bibfnamefont {G.}~\bibnamefont {Verdon-Akzam}}, \bibinfo {author} {\bibfnamefont {E.}~\bibnamefont {Mart\'{\i}n-Mart\'{\i}nez}}, \ and\ \bibinfo {author} {\bibfnamefont {A.}~\bibnamefont {Kempf}},\ }\href {\doibase 10.1103/PhysRevA.93.022308} {\bibfield  {journal} {\bibinfo  {journal} {Phys. Rev. A}\ }\textbf {\bibinfo {volume} {93}},\ \bibinfo {pages} {022308} (\bibinfo {year} {2016})}\BibitemShut {NoStop}%
\bibitem [{\citenamefont {Louko}\ and\ \citenamefont {Toussaint}(2016)}]{louko2016fermionflat}%
  \BibitemOpen
  \bibfield  {author} {\bibinfo {author} {\bibfnamefont {J.}~\bibnamefont {Louko}}\ and\ \bibinfo {author} {\bibfnamefont {V.}~\bibnamefont {Toussaint}},\ }\href {\doibase 10.1103/PhysRevD.94.064027} {\bibfield  {journal} {\bibinfo  {journal} {Phys. Rev. D}\ }\textbf {\bibinfo {volume} {94}},\ \bibinfo {pages} {064027} (\bibinfo {year} {2016})}\BibitemShut {NoStop}%
\bibitem [{\citenamefont {Torres}\ \emph {et~al.}(2020)\citenamefont {Torres}, \citenamefont {Perche}, \citenamefont {Landulfo},\ and\ \citenamefont {Matsas}}]{bruno2020neutrino}%
  \BibitemOpen
  \bibfield  {author} {\bibinfo {author} {\bibfnamefont {B.~d. S.~L.}\ \bibnamefont {Torres}}, \bibinfo {author} {\bibfnamefont {T.~R.}\ \bibnamefont {Perche}}, \bibinfo {author} {\bibfnamefont {A.~G.~S.}\ \bibnamefont {Landulfo}}, \ and\ \bibinfo {author} {\bibfnamefont {G.~E.~A.}\ \bibnamefont {Matsas}},\ }\href {\doibase 10.1103/PhysRevD.102.093003} {\bibfield  {journal} {\bibinfo  {journal} {Phys. Rev. D}\ }\textbf {\bibinfo {volume} {102}},\ \bibinfo {pages} {093003} (\bibinfo {year} {2020})}\BibitemShut {NoStop}%
\bibitem [{\citenamefont {Glimm}\ and\ \citenamefont {Jaffe}(1968)}]{glimm1968lambda}%
  \BibitemOpen
  \bibfield  {author} {\bibinfo {author} {\bibfnamefont {J.}~\bibnamefont {Glimm}}\ and\ \bibinfo {author} {\bibfnamefont {A.}~\bibnamefont {Jaffe}},\ }\href {\doibase 10.1103/PhysRev.176.1945} {\bibfield  {journal} {\bibinfo  {journal} {Phys. Rev.}\ }\textbf {\bibinfo {volume} {176}},\ \bibinfo {pages} {1945} (\bibinfo {year} {1968})}\BibitemShut {NoStop}%
\bibitem [{\citenamefont {Glimm}\ and\ \citenamefont {Jaffe}(1970{\natexlab{a}})}]{glimm1970lambda}%
  \BibitemOpen
  \bibfield  {author} {\bibinfo {author} {\bibfnamefont {J.}~\bibnamefont {Glimm}}\ and\ \bibinfo {author} {\bibfnamefont {A.}~\bibnamefont {Jaffe}},\ }\href {http://www.jstor.org/stable/1970582} {\bibfield  {journal} {\bibinfo  {journal} {Annals of Mathematics}\ }\textbf {\bibinfo {volume} {91}},\ \bibinfo {pages} {362} (\bibinfo {year} {1970}{\natexlab{a}})}\BibitemShut {NoStop}%
\bibitem [{\citenamefont {Glimm}\ and\ \citenamefont {Jaffe}(1970{\natexlab{b}})}]{glimm1970lambda2}%
  \BibitemOpen
  \bibfield  {author} {\bibinfo {author} {\bibfnamefont {J.}~\bibnamefont {Glimm}}\ and\ \bibinfo {author} {\bibfnamefont {A.}~\bibnamefont {Jaffe}},\ }\href {\doibase 10.1007/BF02392335} {\bibfield  {journal} {\bibinfo  {journal} {Acta Mathematica}\ }\textbf {\bibinfo {volume} {125}},\ \bibinfo {pages} {203 } (\bibinfo {year} {1970}{\natexlab{b}})}\BibitemShut {NoStop}%
\bibitem [{\citenamefont {Schrader}(1972)}]{schrader1972yukawa}%
  \BibitemOpen
  \bibfield  {author} {\bibinfo {author} {\bibfnamefont {R.}~\bibnamefont {Schrader}},\ }\href {\doibase https://doi.org/10.1016/0003-4916(72)90274-6} {\bibfield  {journal} {\bibinfo  {journal} {Annals of Physics}\ }\textbf {\bibinfo {volume} {70}},\ \bibinfo {pages} {412} (\bibinfo {year} {1972})}\BibitemShut {NoStop}%
\bibitem [{\citenamefont {Summers}(2012)}]{summers2012perspective}%
  \BibitemOpen
  \bibfield  {author} {\bibinfo {author} {\bibfnamefont {S.~J.}\ \bibnamefont {Summers}},\ }\href {https://arxiv.org/abs/1203.3991} {\bibfield  {journal} {\bibinfo  {journal} {arXiv preprint arXiv:1203.3991}\ } (\bibinfo {year} {2012})}\BibitemShut {NoStop}%
\bibitem [{\citenamefont {Giovannetti}\ and\ \citenamefont {Fazio}(2005)}]{giovanetti2005capacity-spin}%
  \BibitemOpen
  \bibfield  {author} {\bibinfo {author} {\bibfnamefont {V.}~\bibnamefont {Giovannetti}}\ and\ \bibinfo {author} {\bibfnamefont {R.}~\bibnamefont {Fazio}},\ }\href {\doibase 10.1103/PhysRevA.71.032314} {\bibfield  {journal} {\bibinfo  {journal} {Phys. Rev. A}\ }\textbf {\bibinfo {volume} {71}},\ \bibinfo {pages} {032314} (\bibinfo {year} {2005})}\BibitemShut {NoStop}%
\bibitem [{\citenamefont {Bose}(2003)}]{bose2003spin-chain}%
  \BibitemOpen
  \bibfield  {author} {\bibinfo {author} {\bibfnamefont {S.}~\bibnamefont {Bose}},\ }\href {\doibase 10.1103/PhysRevLett.91.207901} {\bibfield  {journal} {\bibinfo  {journal} {Phys. Rev. Lett.}\ }\textbf {\bibinfo {volume} {91}},\ \bibinfo {pages} {207901} (\bibinfo {year} {2003})}\BibitemShut {NoStop}%
\bibitem [{\citenamefont {Bayat}\ \emph {et~al.}(2008)\citenamefont {Bayat}, \citenamefont {Burgarth}, \citenamefont {Mancini},\ and\ \citenamefont {Bose}}]{bayat2008spin-chain}%
  \BibitemOpen
  \bibfield  {author} {\bibinfo {author} {\bibfnamefont {A.}~\bibnamefont {Bayat}}, \bibinfo {author} {\bibfnamefont {D.}~\bibnamefont {Burgarth}}, \bibinfo {author} {\bibfnamefont {S.}~\bibnamefont {Mancini}}, \ and\ \bibinfo {author} {\bibfnamefont {S.}~\bibnamefont {Bose}},\ }\href {\doibase 10.1103/PhysRevA.77.050306} {\bibfield  {journal} {\bibinfo  {journal} {Phys. Rev. A}\ }\textbf {\bibinfo {volume} {77}},\ \bibinfo {pages} {050306} (\bibinfo {year} {2008})}\BibitemShut {NoStop}%
\end{thebibliography}%
\end{document}